\documentclass[12pt, a4paper]{article}
\usepackage{siunitx,booktabs}
\usepackage{amsmath,amsthm,amsfonts,amscd,amssymb} % Some packages to write mathematics.
\usepackage{verbatim}      	    % Allows quoting source with commands.
\usepackage{makeidx}       	    % Package to make an index.
\usepackage{url}
\usepackage{rotating}	
\usepackage{appendix}
\usepackage{endnotes}	
\usepackage{setspace}
\usepackage{tabu}
\usepackage[left=1.25in, right=1.25in, top=1.25in, bottom=1.25in]{geometry}%
\usepackage{booktabs}
\usepackage{tabularx}
\usepackage{graphicx}
\usepackage{scalefnt}
\usepackage{float}
\usepackage{titlesec}
\usepackage{caption}
\usepackage{siunitx}
\usepackage{subfig}
\usepackage{longtable}
\usepackage{siunitx}
\usepackage{enumitem}

\usepackage{dcolumn}
\newcolumntype{d}[1]{D{.}{.}{#1}}

\floatstyle{ruled}
\newfloat{algorithm}{tbp}{loa}
\providecommand{\algorithmname}{Algorithm}
\floatname{algorithm}{\protect\algorithmname}

\usepackage[round]{natbib}
\usepackage[hyperindex,breaklinks]{hyperref}
\hypersetup{colorlinks=true,       % false: boxed links; true: colored links
   linkcolor=red,
    citecolor=blue,        % color of links to bibliography
    filecolor=magenta,      % color of file links
    urlcolor=cyan           % color of external links
}

\titleformat*{\section}{\large\bfseries\filcenter}%
\titleformat*{\subsection}{\normalsize\bfseries}
\titleformat*{\subsubsection}{\normalsize\itshape}

\titlelabel{\thetitle.\quad}

\newtheorem{proposition}{Proposition}
\newtheorem{lemma}{Lemma}
\newtheorem{theorem}{Theorem}
\newtheorem{assumption}{Assumption}

\usepackage{xr}  

\begin{document}
\title{Efficient Likelihood-based Estimation via Annealing for Dynamic Structural Macrofinance Models}
\author{\normalsize {Andras Fulop}\thanks{ESSEC Business School, Paris-Singapore. Email: fulop@essec.edu.}   \and
            \normalsize {Jeremy Heng}\thanks{ESSEC Business School, Paris-Singapore. Email: heng@essec.edu.} \and
           \normalsize  {Junye Li}\thanks{Fudan University. Email: li\_junye@fudan.edu.cn}}
\date{\small{First version: January 2021; This version: December 2021}}
\maketitle

\begin{abstract}
Most solved dynamic structural macrofinance models are non-linear and/or non-Gaussian state-space models with high-dimensional and complex structures. We propose an annealed controlled sequential Monte Carlo method that delivers numerically stable and low variance estimators of the likelihood function. The method relies on an annealing procedure to gradually introduce information from observations and constructs globally optimal proposal distributions by solving associated optimal control problems that yield zero variance likelihood estimators. To perform parameter inference, we develop a new adaptive SMC$^2$ algorithm that employs likelihood estimators from annealed controlled sequential Monte Carlo. We provide a theoretical stability analysis that elucidates the advantages of our methodology and asymptotic results concerning the consistency and convergence rates of our SMC$^2$ estimators. We illustrate the strengths of our proposed methodology by estimating two popular macrofinance models: a non-linear new Keynesian dynamic stochastic general equilibrium model and a non-linear non-Gaussian consumption-based long-run risk model.

%\noindent \textbf{JEL Classification}: C11, C51, C82, E44
\end{abstract}

\clearpage

\section{Introduction}

In macroeconomics and finance, dynamic structural models provide a convenient theoretical framework to explain economic fluctuations, government policies, and asset prices movements.\footnote{For example, dynamic stochastic general equilibrium (DSGE) models, spanning from real business cycle models \citep{rbc1982} to new Keynesian models \citep*[see, e.g.,][]{christiano2005, smets2003}, are commonly used to explain and predict comovements of aggregate macroeconomic fundamentals over the business cycle; consumption-based asset pricing models seek to relate agents' consumption to asset prices and explore fundamental determinants of asset prices \citep{lucas1978}.} A common feature of these models is that agents' decision rules are derived from assumptions of preferences and technologies by solving intertemporal optimization problems. In empirical applications, these models need to be numerically solved and estimated on real macroeconomic and financial data. For a long time, the literature has been following an approximate approach that these models are first log-linearized and cast into approximate linear state-space models and are then estimated using either likelihood-based methods or moment-based methods.\footnote{The estimation of structural models in macroeconomics and finance follow very different directions. In macroeconomics, the log-linearized models are usually estimated via Bayesian MCMC methods; see, for example, \citet*{herbst2016} and \citet*{schorfheide2016chapter}. In finance, the typical practice involves formulating a set of first-order optimality conditions that map onto moment-based estimation of the key parameters of interest; see, for example, \citet*{bansal2007res}, \citet*{bansal2016jme}, and \citet*{gallant2019rfs}.}

It is now known that second-order approximation errors in linear solutions of dynamic economic models have first-order effects on likelihood estimation \citep*{fernandez2006ecta, fernandez2007estimating}, and that errors in the likelihood estimation are compounded with the sample size. In recent work, \citet*{pohl2018higher} show that log-linearization of consumption-based asset pricing models with long-run risks may result in economically significant errors when state variables are persistent. However, when models are solved using more accurate numerical methods such as high-order perturbation methods \citep*{schmitt2004} or projection methods \citep{judd1992projection}, the solved models become non-linear and/or non-Gaussian state-space models with high-dimensional and complex structures. The aim of this article is to propose an efficient econometric toolbox based on sequential Monte Carlo (SMC) methods that facilitate likelihood-based estimation when non-linear numerical methods are used to solve models.

Our contribution is threefold. First, we propose a novel methodology to efficiently estimate the likelihood function by building on state-of-the-art SMC methods, also known as particle filters. While particle filters have been applied to estimate the likelihood of dynamic economic models \citep[see, e.g.,][]{fernandez2007estimating, deJong2013Restud, herbst2019tempered, fulop2020bayesian}, the successful application of this approach in model estimation remains challenging due to the large variance of the resulting likelihood estimator. A recent work by \citet{heng2020controlled} aims to address this issue using ideas from the optimal control literature.
Their proposed methodology, termed as controlled SMC, constructs \emph{globally optimal} proposal distributions that take the {entire sequence of observations} into account using approximate dynamic programming schemes.
However, for dynamic structural macrofinance models with complex structures and highly informative observations,
controlled SMC can easily fail due to serious numerical instabilities, particularly when the initial proposal distributions are far from optimality.

We develop a new methodology, termed as annealed controlled SMC, which prevents such numerical issues from manifesting. The central idea is to rely on annealing to gradually introduce information from the observations and iteratively refine the proposal distributions as the inverse temperature increases. The original controlled SMC method of \citet{heng2020controlled} can be seen as a special case of our approach with the inverse temperature fixed at one. Our proposed methodology yields a zero-variance likelihood estimator when the optimal proposal distributions are attainable. Practical implementation requires function approximations, hence resulting in sub-optimal proposal distributions. We provide a theoretical analysis to characterize the quality of our proposal distributions, which elucidates the properties and advantages of annealed controlled SMC. Our results reveal the importance of adopting the annealing procedure. The use of annealing has been explored in online settings to construct \emph{locally optimal} proposal distributions that take the next observation into account \citep{herbst2019tempered,godsill2001improvement}. In an offline context where globally optimal proposals are desired, the use of annealing to alleviate numerical instabilities was mentioned in earlier work by \citet{scharth2016particle}, but the main and crucial difference is that their proposal learning procedure does not exploit the proposals learnt at lower inverse temperatures.

Second, to perform full Bayesian inference, we employ annealed controlled SMC within a novel adaptive SMC$^2$ algorithm that
recursively approximates a sequence of annealed posterior distributions of the latent states and model parameters. Compared to existing annealed SMC$^2$ routines such as \citet*{duan2015density} and \citet*{svensson2018learning}, our adaptive approach has two key features.
The flexibility of our framework allows us to utilize computationally inexpensive SMC methods (e.g., bootstrap particle filter) at low inverse temperatures, thereby efficiently ruling out unlikely parameters at early stages of the algorithm, and switch to annealed controlled SMC method, which is more computationally demanding but accurate, at higher inverse temperatures when the algorithm considers more promising regions of the parameter space.
Moreover, under suitable assumptions, we establish law of large numbers and central limit theorems for our estimators of posterior expectations and the model evidence as the number of parameter particles goes to infinity, for any given number of state particles larger than one.

%Firstly, it regenerates the proposal-refining routine in each tempering step, allowing us to utilize fast routines (e.g., bootstrap filter) at the initial stages of the algorithm when the tempering coefficient on the observation density is low.
%The rationale here is that at low inverse temperatures, the performance of uncontrolled SMC would be adequate as the influence of the observations are limited. This allows us to efficiently rule out very unlikely parameters at early iterations of the algorithm. For larger inverse temperatures $\lambda_i>\lambda_*$, it is worthwhile spending the computational overhead to learn the optimal policies on more promising regions of the parameter space.
%Second, under suitable assumptions, we are able to prove laws of large numbers and central limit theorems for our estimators of model evidence and posterior expectations as the number of parameter particles goes to infinity for any given number of state particles larger than 1.

Compared to MCMC methods or particle MCMC methods \citep*{andrieu2010particle} that sample from the posterior distribution using a single Markov chain, our algorithm shares the many advantages of SMC samplers \citep{del2006sequential,dai2020invitation}. This includes parallelism over the parameter particles, automated tuning of the inverse temperatures and proposal transitions in the parameter space, and an estimator of the model evidence that facilitiates model comparison.

Third, we apply our adaptive SMC$^2$ with annealed controlled SMC nested to estimate two popular dynamic structural models in macroeconomics and finance. The first model is a prototypical new Keynesian dynamic stochastic general equilibrium model (DSGE) that has been studied in \citet{woodford2003}, \citet{an2007bayesian} and \citet{herbst2016}; the second is a consumption-based long-run risk asset pricing model, in which consumption volatility is modelled as an autoregressive gamma process \citep{gourieroux2006autoregressive} instead of an autoregressive process that is typically adopted in the standard long-run risk model \citep*[see, e.g.,][]{bansal2004risks, bansal2012empirical}. The first application is described in the main article and the second in the Appendix for the sake of brevity.

We conduct extensive simulation studies based on the above new Keynesian DSGE model that is solved using the log-linearized method \citep{sims2002solving} or the second-order perturbation method with pruning \citep*{schmitt2004, schmitt2007optimal, fernandez2018res}. For the linearized model, whose likelihood function can be evaluated exactly using a Kalman filter, we find that annealed controlled SMC with moderate number of particles delivers a likelihood estimator with negligible variance. In contrast to the bootstrap particle filter \citep{gordon1993novel}, the performance of annealed controlled SMC is very robust to the magnitude of the standard deviation of the measurement error.
For the non-linear model, whose likelihood function is intractable, we find that the variance of annealed controlled SMC log-likelihood estimator with modest number of particles is about four orders of magnitude smaller than that of the bootstrap particle filter with very large number of particles, and that the magnitude of variance reduction holds even for very small standard deviation of the measurement error.
Furthermore, we show that the likelihood estimate from the approximate log-linearized model is much smaller than that of the desired non-linear model obtained with annealed controlled SMC.

We then employ our proposed methodology to estimate the new Keynesian DSGE model on real macroeconomic and financial data. While the log-linearized approximation of this model has been investigated in \citet{an2007bayesian} and empirically estimated in \citet{herbst2016}, the corresponding non-linear model has not been fully studied yet in the literature.\footnote{In a recent paper by \citet{aruoba2020wp}, a variant of this new Keynesian DSGE model is solved using a piecewise-linear approximation method and is estimated using a particle MCMC method.}
We estimate the resulting non-linear model under the second-order perturbation method with pruning based on three observables: quarterly per capita GDP growth rate, quarterly inflation, and quarterly annualized interest rate. To avoid the issue of having a zero lower bound on the interest rate, we focus on the pre-crisis sample, ranging from 1983:Q1 to 2007:Q4 with a total of 100 observations.

We estimate the model in two settings: when the standard deviations of the measurement errors are fixed at 20\% of the corresponding sample standard deviations, as is commonly considered in the literature, and when they are treated as free parameters to be inferred using the data.
%\textcolor{red}{\textbf{Jeremy: this is quite abrupt and not really concrete} ``Our method performs quite well in estimating this model, as in both cases, the parameter acceptance rate fluctuates around 25\%, ranging from about 0.20 to 0.40, a typical range suggested by the Bayesian literature.''}
We notice that while the estimates suggest a low degree of price rigidity when the standard deviations of the measurement errors are fixed, the opposite conclusion is reached when they are treated as free parameters. This suggests the practice of fixing the standard deviations of the measurement errors may distort model implications.
Furthermore, we find that the estimated standard deviations are very different from the fixed values and suggest that the model fits the interest rate better than the output growth rate and the inflation rate.

The remainder of the paper is organized as follows. Section \ref{sec: ssm} describes the models of interest and its state-space representation. Section \ref{sec: likelihood} introduces our annealed controlled SMC method to efficiently estimate the likelihood of these models. Section \ref{sec: parameter_inference} details our adaptive SMC$^2$ algorithm for parameter inference. Section \ref{sec:apps} presents an application on a non-linear new Keynesian DSGE model; another application on a non-linear and non-Gaussian consumption-based long-run risk asset pricing model is given in the Appendix. Finally, Section \ref{sec: conclusion} concludes the paper.
Supplementary details and proofs of all theoretical results are provided in the Appendix.

%%%%%%%%%%%%%%%%%%%%%%%%%%
\section{State-space representation}\label{sec: ssm}

In most dynamic structural macrofinance models, the agents' decision rules are derived from assumptions of preferences and technologies by solving intertemporal optimization problems. For estimation and empirical applications, these models are first numerically solved and then cast into the framework of state-space models. For time $t=1,\ldots,T$, let $s_t\in\mathbb{S}\subseteq\mathbb{R}^d$ denote a vector of latent state variables, which may include both endogenous and exogenous variables, and $y_t\in\mathbb{Y}\subseteq\mathbb{R}^{d_y}$ denote a vector of observations. Given a vector of unknown parameters $\theta\in\Theta\subseteq\mathbb{R}^{d_{\theta}}$, solving a dynamic structural model gives the following relations
     \begin{eqnarray}\label{ssm:state_obs}
           s_0=\Phi_{\theta}^{(0)}(\varepsilon_0),\quad s_t = \Phi_{\theta}(s_{t-1}, \varepsilon_t),\quad
           y_t = \Psi_{\theta}(s_{t-1}, s_{t}, u_t),
     \end{eqnarray}
where $\Phi_{\theta}^{(0)}$, $\Phi_{\theta}$ and $\Psi_{\theta}$ are non-linear functions from the model solution,
$\varepsilon_t$ and $u_t$ represent exogenous shocks and observation noise, respectively.
When model-implied observables are deterministically related to state variables,
$u_t$ usually captures the measurement error of observations.

We view the sequence of states $(s_t)_{t=0}^T$ as a Markov chain on the state-space $\mathbb{S}$ evolving according to
\begin{equation}\label{eqn:latent_process}
	s_0 \sim \mu_{\theta},\quad s_t \sim f_{\theta}(\cdot|s_{t-1}),\quad t=1,\ldots,T,
\end{equation}
where $\mu_{\theta}(ds_0)$ denotes an initial distribution and $f_{\theta}(ds_t | s_{t-1})$ a Markov transition kernel on $\mathbb{S}$ that depend on parameters $\theta$. For most models in the literature, these measures either admit densities with respect to the Lebesgue measure on $\mathbb{R}^d$, or are supported on a lower-dimensional subspace when there are both endogenous and exogenous state variables.

The sequence of observations $(y_t)_{t=0}^T$ are assumed to be conditionally independent given the latent process $(s_t)_{t=0}^T$ and are distributed according to
\begin{equation}\label{eqn:observation_model}
	y_t \sim g_{\theta}(\cdot |s_{t-1},s_t), \quad t=1,\ldots,T,
\end{equation}
where $g_{\theta}(y_t |s_{t-1},s_t)$ denotes the observation density that is also parameter-dependent. The dependence on the latent states at the previous and current times in Equation \eqref{eqn:observation_model} is convenient when modeling asset returns data; longer time dependencies can also be accommodated in our proposed methodology. \bigskip
%
%\noindent \emph{Example 1: DSGE models.} DSGE models
%
%\noindent \emph{Example 2: Long-run risk models.}

\section{Likelihood estimation}\label{sec: likelihood}

Given an observation sequence $y_{1:T}=(y_t)_{t=1}^T\in\mathbb{Y}^{T}$, the complete likelihood is
{\small
\begin{equation}\label{eqn:likelihood}
	p(ds_{0:T}, y_{1:T} | \theta) = p(ds_{0:T}|\theta)p(y_{1:T} | s_{0:T}, \theta) = \mu_{\theta}(ds_0)\prod_{t=1}^Tf_{\theta}(ds_t | s_{t-1}) g_{\theta}(y_t |s_{t-1},s_t).
\end{equation}}To perform parameter inference, we have to integrate out the latent process to compute the likelihood function
\begin{equation}\label{eqn:likelihood}
	p(y_{1:T} | \theta) = \int_{\mathbb{S}^{T+1}}p(ds_{0:T}, y_{1:T} | \theta),
\end{equation}
and to estimate the latent states, we need to characterize the smoothing distribution
\begin{equation}\label{eqn:smoothing_distribution}
	p(ds_{0:T} | y_{1:T}, \theta) = \frac{p(ds_{0:T},y_{1:T} | \theta)}{p(y_{1:T} | \theta)}.
\end{equation}

To facilitate computation and prevent numerical instabilities that often arise when working with dynamic structural macrofinance models,
it will be beneficial to gradually introduce the influence of the observations $y_{1:T}$. We do so by defining, for an inverse temperature $\lambda\in[0,1]$, the annealed distribution
\begin{eqnarray}\label{eqn:tempered_likelihood}
	p(ds_{0:T}, y_{1:T} | \theta, \lambda) &=& p(ds_{0:T}|\theta)p(y_{1:T} | s_{0:T}, \theta,\lambda)\notag\\
	 &=& \mu_{\theta}(ds_0)\prod_{t=1}^Tf_{\theta}(ds_t | s_{t-1}) g_{\theta}(y_t |s_{t-1},s_t)^{\lambda},
\end{eqnarray}
and the corresponding likelihood and smoothing distribution as
{\small
\begin{equation}\label{eqn:tempered_smoothing}
	p(y_{1:T} | \theta, \lambda) = \int_{\mathbb{S}^{T+1}}p(ds_{0:T}, y_{1:T} | \theta, \lambda), \quad
	p(ds_{0:T} | y_{1:T}, \theta, \lambda) = \frac{p(ds_{0:T},y_{1:T} | \theta, \lambda)}{p(y_{1:T} | \theta, \lambda)}.
\end{equation}}

By gradually increasing $\lambda$, we introduce a path of distributions between the law of the latent process
$p(ds_{0:T}|\theta)$ and the desired smoothing distribution $p(ds_{0:T} | y_{1:T}, \theta)$ in Equation \eqref{eqn:smoothing_distribution}. In the context of state-space models, similar motivations can be found in \citet*{godsill2001improvement}, \citet*{svensson2018learning}, and \citet*{herbst2019tempered}. We note that the likelihood $p(y_{1:T} | \theta, \lambda)$ in Equation \eqref{eqn:tempered_smoothing} is different but related to the tempered likelihood $p(y_{1:T} | \theta)^{\lambda}$, considered in \citet*{duan2015density}, via Jensen's inequality.

\subsection{Sequential Monte Carlo}\label{sec:SMC}

For most models of practical interest, the quantities in Equation \eqref{eqn:tempered_smoothing} are intractable and we have to rely on Monte Carlo approximations. Sequential Monte Carlo (SMC) methods, also known as particle filters, can provide state-of-the-art approximations by simulating an interacting particle system of size $N\in\mathbb{N}$ \citep*{doucet2001introduction,chopin2020introduction}. In what follows, we present a generic description of SMC that encompasses several algorithms in a common framework.

At the initial time, we sample $N$ independent states $(s_0^{(n)})_{n=1}^N$ from a proposal distribution
$q_0(ds_0|\theta,\lambda)$ on $\mathbb{S}$. The states are then assigned normalized weights $(W_0^{(n)})_{n=1}^N$ that sum
to one according to a weight function $W_0^{(n)}\propto w_0(s_0^{(n)};\theta,\lambda)$. We then sample from the weighted particle approximation $\sum_{n=1}^NW_0^{(n)}\delta_{s_0^{(n)}}(ds_0)$ to multiply states with high weights and discard states
that are unlikely. This operation, known as resampling, can be seen as sampling ancestor indexes $(a_0^{(n)} )_{n=1}^N$ from a distribution $r(\cdot|W_0^{(1)},\ldots,W_0^{(N)})$ on $\{1,\ldots,N\}^N$.
We will consider multinomial resampling where $(a_0^{(n)} )_{n=1}^N$ are independent samples
from the categorical distribution on $\{1,\ldots,N\}$ with probabilities $(W_0^{(n)})_{n=1}^N$.
For time $t=1,\ldots,T$, we move the resampled states using a proposal transition kernel
$q_t(ds_t|s_{t-1},\theta,\lambda)$ on $\mathbb{S}$, i.e. sample $s_t^{(n)}\sim q_t(\cdot|s_{t-1}^{a_{t-1}^{(n)}},\theta,\lambda)$
independently for $n=1,\ldots,N$. These samples are then assigned normalized weights $(W_t^{(n)})_{n=1}^N$
according to the weight function $W_t^{(n)}\propto w_t(s_{t-1}^{a_{t-1}^{(n)}},s_t^{(n)};\theta,\lambda)$.
If $t<T$, we perform resampling by drawing the ancestor indexes $(a_t^{(n)})_{n=1}^N$ from $r(\cdot|W_t^{(1)},\ldots,W_t^{(N)})$. We assume that the smoothing distribution $p(ds_{0:T}|y_{1:T},\theta,\lambda)$ at inverse temperature $\lambda\in[0,1]$ is absolutely continuous with respect to the law of the proposal process
\begin{equation}\label{eqn:proposal_law}
	q(ds_{0:T}|\theta,\lambda)=q_0(ds_0|\theta,\lambda)\prod_{t=1}^Tq_t(ds_t|s_{t-1},\theta,\lambda)
\end{equation}
with density written as $p(s_{0:T}|y_{1:T},\theta,\lambda)/q(s_{0:T}|\theta,\lambda)$,
and that the choice of weight functions satisfy
\begin{equation}\label{eqn:weights_requirement}
	w_0(s_0;\theta,\lambda)\prod_{t=1}^Tw_t(s_{t-1}, s_t;\theta,\lambda) = \frac{p(s_{0:T},y_{1:T}|\theta,\lambda)}{q(s_{0:T}|\theta,\lambda)}.
\end{equation}

The complexity of SMC methods is $\mathcal{O}(NT)$, and its memory requirement is also $\mathcal{O}(NT)$
if all states $(s_t^{(n)})_{t=0,n=1}^{T,N}$ and ancestor indexes $(a_t^{(n)})_{t=0,n=1}^{T-1,N}$ are stored.
The latter can be lowered to $\mathcal{O}(T+N\log N)$ using efficient implementations \citep*{jacob2015path}.
Given the output of SMC, an unbiased estimator of the likelihood $p(y_{1:T}|\theta,\lambda)$
at the inverse temperature $\lambda\in[0,1]$ is
\begin{equation}\label{eqn:smc_likelihood}
	\hat{p}(y_{1:T}|\theta,\lambda) = \left\lbrace\frac{1}{N}\sum_{n=1}^Nw_0(s_0^{(n)};\theta,\lambda)\right\rbrace
	\left\lbrace\prod_{t=1}^T\frac{1}{N}\sum_{n=1}^Nw_t(s_{t-1}^{(a_{t-1}^{(n)})}, s_t^{(n)};\theta,\lambda)\right\rbrace,
\end{equation}
and a weighted particle approximation of the corresponding smoothing distribution is given by
\begin{equation}\label{eqn:smc_smoothing_approx}
	\hat{p}(ds_{0:T}|y_{1:T},\theta,\lambda) = \sum_{n=1}^NW_T^{(n)}\delta_{s_{0:T}^{(n)}}(ds_{0:T}).
\end{equation}

In the above, each trajectory $s_{0:T}^{(n)}$ is formed by tracing the ancestral lineage of $s_T^{(n)}$, i.e.
$s_{0:T}^{(n)}=(s_t^{(l_t^{(n)})})_{t=0}^T$ with particle indexes $(l_t^{(n)})_{t=0}^T$ given by the backward recursion
$l_T^{(n)}=n$ and $l_t^{(n)} = a_t^{(l_{t+1}^{(n)})}$ for $t=T-1,\ldots,0$.
Using the particle approximation in Equation \eqref{eqn:smc_smoothing_approx}, smoothing expectations
$\int_{\mathbb{S}^{T+1}} \varphi(s_{0:T})p(ds_{0:T}|y_{1:T},\theta,\lambda)$ for any
integrable function $\varphi:\mathbb{S}^{T+1}\rightarrow\mathbb{R}$ can be approximated by the weighted average
$\sum_{n=1}^NW_T^{(n)}\varphi(s_{0:T}^{(n)})$.
Convergence properties of these estimators in the limit of the number of particles $N\rightarrow\infty$
are well-understood \citep{del2004feynman,chopin2004central}. However, a successful implementation of SMC in practice crucially relies on the choice of proposal distributions in Equation \eqref{eqn:proposal_law}.
Poor choices will require prohibitively large number of samples to obtain adequate likelihood and state estimators.

In our framework, the bootstrap particle filter (BPF) of \citet*{gordon1993novel} corresponds to having
the law of the latent process in Equation \eqref{eqn:latent_process} as proposal, i.e.
$q_0(ds_0|\theta)=\mu_{\theta}(ds_0)$, $q_t(ds_t|s_{t-1},\theta)=f_{\theta}(ds_t|s_{t-1})$ for $t=1,\ldots,T$,
and the weight functions $w_0(s_0)=1$, $w_t(s_{t-1},s_t;\theta,\lambda)=g_{\theta}(y_t|s_{t-1},s_t)^{\lambda}$ for $t=1,\ldots,T$. It is straightforward to verify that these choices satisfy Equation \eqref{eqn:weights_requirement}.
Although simple to implement, the efficiency of BPF estimators can be particularly poor in practice if the observations are informative, unless the inverse temperature $\lambda$ is small enough to limit the influence of
the observations. The fully adapted auxiliary particle filter (APF) introduced by \citet*{pitt1999filtering} and \citet*{carpenter1999improved} can give better performance by constructing proposals that take the current observation into account. This corresponds to selecting the \emph{locally optimal} proposals
$q_0(ds_0|\theta)=\mu_{\theta}(ds_0)$, $q_t(ds_t|s_{t-1},\theta,\lambda)=p(ds_t|s_{t-1},y_t,\theta,\lambda)$
for $t=1,\ldots,T$, and the weight functions $w_0(s_0)=1$,
$w_t(s_{t-1};\theta,\lambda)=p(y_t|s_{t-1},\theta,\lambda)$ for $t=1,\ldots,T$.
Exact implementation of APF is not feasible for many non-linear and/or non-Gaussian state-space models,
such as the ones to be considered in this paper, as the proposal transitions and weight functions are intractable.
Various tractable approximations of APF have been considered \citep[see, e.g.,][]{doucet2000sequential, finke2020limit}.

\subsection{Controlled sequential Monte Carlo}

In this paper, we develop a novel methodology that can significantly outperform both BPF and APF,
by constructing \emph{globally optimal} proposal distributions that take the entire observation sequence $y_{1:T}$ into account. Our approach extends the controlled SMC methodology of \citet{heng2020controlled} that in turn builds upon the works by \citet*{richard2007efficient, scharth2016particle} and \citet*{guarniero2017iterated}. %that can significantly outperform both BPF and APF

In what follows, suppose that we have a given SMC method, defined by a specific choice of proposals $(q_t)_{t=0}^T$ and weight functions $(w_t)_{t=0}^T$ satisfying Equation \eqref{eqn:weights_requirement}. For notational ease, we will suppress their dependence on the parameters $\theta\in\Theta$ and the inverse temperature $\lambda\in[0,1]$. The key idea is to express the desired smoothing distribution $p(ds_{0:T} | y_{1:T}, \theta, \lambda) = p(ds_{0} | y_{1:T}, \theta, \lambda)
\prod_{t=1}^Tp(ds_{t} | s_{t-1}, y_{t:T}, \theta, \lambda)$ as
\begin{align}\label{eqn:rewrite_smoothing}
	p(ds_{0} | y_{1:T}, \theta, \lambda) = \frac{q_0(ds_0)\psi_0^*(s_0)}{q_0(\psi_0^*)},\quad
	p(ds_{t} | s_{t-1}, y_{t:T}, \theta, \lambda) = \frac{q_t(ds_t|s_{t-1})\psi_t^*(s_{t-1},s_t)}{q_t(\psi_t^*|s_{t-1})}
\end{align}
for $t=1,\ldots,T$, where the sequence of functions $\psi^*=(\psi_t^*)_{t=0}^T$ are defined by the backward recursion
\begin{align}\label{eqn:optimal_psi}
	\psi_T^*(s_{T-1},s_T) ~&=~ w_T(s_{T-1},s_T),\notag\\
	\psi_t^*(s_{t-1},s_t) ~&=~ w_t(s_{t-1},s_t)q_{t+1}(\psi_{t+1}^*|s_{t}),\quad t = T-1,\ldots,1,\\
	\psi_0^*(s_0) ~&=~ w_0(s_0)q_1(\psi_1^*|s_0).\notag
\end{align}
The notation $q_0(\psi_0^*)=\int_{\mathbb{S}}q_0(ds_0)\psi_0^*(s_0)$ denotes the expectation of $\psi_0^*$
with respect to the distribution $q_0$, and $q_t(\psi_t^*|s_{t-1})=\int_{\mathbb{S}}q_t(ds_t|s_{t-1})\psi_t^*(s_{t-1},s_t)$
denotes the conditional expectation of $\psi_t^*$ under the Markov transition kernel $q_t$.
In the case of the BPF, given that $\psi^*$ admits the following probabilistic interpretation
\begin{equation}
	\psi_0^*(s_0)=p(y_{1:T}|s_0,\theta,\lambda),\quad\psi_t^*(s_{t-1},s_t)=p(y_{t:T}|s_{t-1},s_t,\theta,\lambda),\quad t=1,\ldots,T,
\end{equation}
it is sometimes referred to as the backward information filter \citep*{briers2010smoothing}.

By obtaining an approximation $\psi=(\psi_t)_{t=0}^T$ of the backward recursion in Equation \eqref{eqn:optimal_psi}
using an approximate dynamic programming method that will be discussed in Section \ref{sec:ADP},
we can construct a new proposal distribution
\begin{align}\label{eqn:law_proposal}
 q^{\psi}(ds_{0:T})=q_0^{\psi}(ds_0)\prod_{t=1}^Tq_t^{\psi}(ds_t|s_{t-1})
\end{align}
by mimicking Equation \eqref{eqn:rewrite_smoothing}, i.e. define
\begin{align}\label{eqn:proposal_psi}
	q_0^{\psi}(ds_{0}) = \frac{q_0(ds_0)\psi_0(s_0)}{q_0(\psi_0)},\quad
	q_t^{\psi}(ds_{t} | s_{t-1}) = \frac{q_t(ds_t|s_{t-1})\psi_t(s_{t-1},s_t)}{q_t(\psi_t|s_{t-1})},
	\quad t=1,\ldots,T.
\end{align}
Following the terminology in \citet{heng2020controlled}, we will refer to a sequence of non-negative and bounded
functions $\psi$ as a policy and $\psi^*$ as the optimal policy. As the choice of policy is specific to the application of interest, we defer these discussions until Section \ref{sec:apps} and assume $\psi$ is such that the proposals $(q_t^{\psi})_{t=0}^T$ in Equation \eqref{eqn:proposal_psi} can be sampled from, and the expectations $q_0(\psi_0)$, $q_t(\psi_t|s_{t-1})$ for $t=1,\ldots,T$ can be evaluated. To employ these newly constructed proposals within SMC, the appropriate weight functions are given by
\begin{align}\label{eqn:weights_psi}
	w_0^{\psi}(s_0) ~&=~ \frac{q_0(\psi_0)q_1(\psi_1|s_0)}{\psi_0(s_0)},\notag\\
	w_t^{\psi}(s_{t-1},s_t) ~&=~ \frac{w_t(s_{t-1},s_t)q_{t+1}(\psi_{t+1}|s_t)}{\psi_t(s_{t-1},s_t)},\quad t=1,\ldots,T-1,\\
	w_T^{\psi}(s_{T-1},s_T) ~&=~ \frac{w_T(s_{T-1},s_T)}{\psi_T(s_{T-1},s_T)},\notag
\end{align}
which satisfy $w_0^{\psi}(s_0)\prod_{t=1}^Tw_t^{\psi}(s_{t-1}, s_t) = p(s_{0:T},y_{1:T}|\theta,\lambda)/q^{\psi}(s_{0:T}|\theta,\lambda)$.
An algorithmic description of the resulting SMC method is detailed in Algorithm \ref{alg:cSMC}.
To distinguish between the initial SMC method and the new SMC method induced by a policy,
we will refer to the former as \emph{uncontrolled} SMC and the latter as \emph{controlled} SMC.
From the output of Algorithm \ref{alg:cSMC}, we have an unbiased estimator of the likelihood $p(y_{1:T}|\theta,\lambda)$
(Step 3) and an approximate sample from the smoothing distribution $p(ds_{0:T} | y_{1:T}, \theta, \lambda)$
by sampling from the weighted particle approximation in Equation \eqref{eqn:smc_smoothing_approx}
(or equivalently selecting an ancestral lineage at Step 4).
%\begin{align}\label{eqn:csmc_likelihood}
%	\hat{p}(y_{1:T}|\theta,\lambda) = \left\lbrace\frac{1}{N}\sum_{n=1}^Nw_0^{\psi}(s_0^{(n)})\right\rbrace
%	\left\lbrace\prod_{t=1}^T\frac{1}{N}\sum_{n=1}^Nw_t^{\psi}(s_{t-1}^{(a_{t-1}^{(n)})}, s_t^{(n)})\right\rbrace,
%\end{align}
%and \eqref{eqn:smc_smoothing_approx} gives a weighted particle approximation of the corresponding smoothing distribution.
The efficiency of these approximations will ultimately depend on how well the chosen policy $\psi$
approximates the optimal policy $\psi^*$. A more precise characterization of this relationship
will be given in Section \ref{sec:analysis1}. We note that the choice $\psi=\psi^*$ is optimal as Algorithm \ref{alg:cSMC}
would yield a zero variance likelihood estimator, for any number of particles $N$, and an exact trajectory from the smoothing distribution.
\begin{algorithm}[h]
{\footnotesize
\protect\caption{\footnotesize{Controlled sequential Monte Carlo at inverse temperature $\lambda\in[0,1]$} \label{alg:cSMC}}

\textbf{Input}: number of particles $N$ and policy $\psi=(\psi_t)_{t=0}^T$.

(1) For time $t=0$ and particle $n=1,\ldots,N$.
\ignorespacesafterend
\begin{description}[itemsep=0pt,parsep=0pt,topsep=0pt,labelindent=0.5cm]
	\item (1a) Sample state $s_{0}^{(n)}\sim q_0^{\psi}$.
	\item (1b) Compute normalized weights $W_0^{(n)}=w_0^{\psi}(s_{0}^{(n)})/\sum_{m=1}^Nw_0^{\psi}(s_{0}^{(m)})$.
\end{description}

%\textcompwordmark{}

(2) For time $t=1,\ldots,T$ and particle $n=1,\ldots,N$.

\begin{description}[itemsep=0pt,parsep=0pt,topsep=0pt,labelindent=0.5cm]
	\item (2a) Sample ancestor $a_{t-1}^{(n)}\sim r(\cdot|W_{t-1}^{(1)},\ldots,W_{t-1}^{(N)})$.

	\item (2b) Sample state $s_t^{(n)} \sim q_t^{\psi}(\cdot|s_{t-1}^{(a_{t-1}^{(n)})})$.

	\item (2c) Compute normalized weights $W_t^{(n)}=w_t^{\psi}(s_{t-1}^{(a_{t-1}^{(n)})},s_{t}^{(n)})/
	\sum_{m=1}^Nw_t^{\psi}(s_{t-1}^{(a_{t-1}^{(m)})},s_{t}^{(m)})$.
\end{description}

%\textcompwordmark{}

(3) Compute likelihood estimator $\hat{p}(y_{1:T}|\theta,\lambda)=\{\frac{1}{N}\sum_{n=1}^Nw_0^{\psi}(s_0^{(n)})\}
	\{\prod_{t=1}^T\frac{1}{N}\sum_{n=1}^Nw_t^{\psi}(s_{t-1}^{(a_{t-1}^{(n)})}, s_t^{(n)})\}$.

(4) Sample an ancestor $l_T\sim r(\cdot|W_{T}^{(1)},\ldots,W_{T}^{(N)})$ and set
ancestral lineage as $l_t = a_t^{(l_{t+1})}$ for $t=T-1,\ldots,0$.

\textbf{Output}: states $(s_t^{(n)})_{t=0,n=1}^{T,N}$, ancestors $(a_t^{(n)})_{t=0,n=1}^{T-1,N}$,
likelihood estimator $\hat{p}(y_{1:T}|\theta,\lambda)$ and trajectory $(s_t^{(l_t)})_{t=0}^T$.}
\end{algorithm}

\subsection{Policy refinement with approximate dynamic programming}\label{sec:ADP}

Suppose we have a policy $\psi=(\psi_t)_{t=0}^T$ that denotes our current approximation of the optimal policy
$\psi^*=(\psi_t^*)_{t=0}^T$ defined in Equation \eqref{eqn:optimal_psi}. If we define $\phi^*=(\phi_t^*)_{t=0}^T$ using the backward recursion
\begin{eqnarray}\label{eqn:optimal_phi}
	\phi_T^*(s_{T-1},s_T) &=& w_T^{\psi}(s_{T-1},s_T) ,\notag\\
	\phi_t^*(s_{t-1},s_t) &=& w_t^{\psi}(s_{t-1},s_t)q_{t+1}^{\psi}(\phi_{t+1}^*|s_{t}),\quad t = T-1,\ldots,1,\\
	\phi_0^*(s_0) &=& w_0^{\psi}(s_0)q_1^{\psi}(\phi_1^*|s_0),\notag
\end{eqnarray}
it can be shown that $\psi^*=\psi\cdot\phi^*=(\psi_t\cdot\phi_t^*)_{t=0}^T$, where $\psi_t\cdot\phi_t^*$ denotes
pointwise multiplication of two functions \citep[see, Proposition 1,][]{heng2020controlled}. This result shows how policy refinements can be performed and identifies $\phi^*$ as the optimal refinement of the current policy $\psi$.
The optimal refinement can be viewed as the solution of an associated Kullback--Leibler optimal control problem
%that minimizes
%\begin{equation}\label{eqn:reverse_KL}
%	\mathrm{KL}\left(q^{\psi\cdot\phi}(ds_{0:T}) ~|~ p(ds_{0:T}|y_{1:T},\theta,\lambda) \right)
%	= \int_{\mathbb{S}^{T+1}}\log\left(\frac{q^{\psi\cdot\phi}(s_{0:T})}{p(s_{0:T}|y_{1:T},\theta,\lambda )}\right)q^{\psi\cdot\phi}(ds_{0:T})
%\end{equation}
%over admissible refinements $\phi$,
and Equation \eqref{eqn:optimal_phi} is the corresponding dynamic programming recursion.
As noted by \citet{heng2020controlled}, drawing this connection allows one to exploit approximate dynamic programming (ADP) methods to approximate the optimal refinement. The following is an ADP scheme to approximate $\phi^*$ by combining function approximation and iterating the backward recursion in Equation \eqref{eqn:optimal_phi}.

Let $(s_t^{(n)})_{t=0,n=1}^{T,N}$ and $(a_t^{(n)})_{t=0,n=1}^{T-1,N}$ denote the states and ancestors from running controlled SMC with the current policy $\psi$. At the terminal time $T$, we approximate $\phi_T^*= w_T^{\psi}$ by solving the least squares problem
\begin{align}\label{eqn:least_squares_terminal}
	\phi_T=\arg\min_{f\in\mathbb{F}_T}\sum_{n=1}^N\left( \log f(s_{T-1}^{(a_{T-1}^{(n)})},s_T^{(n)}) -
	\log w_T^{\psi}(s_{T-1}^{(a_{T-1}^{(n)})},s_T^{(n)})\right)^2,
\end{align}
where $\mathbb{F}_T$ is a function class to be specified. By plugging in the approximation $\phi_T\approx\phi_T^*$
in the iterate $\phi_{T-1}^*=w_{T-1}^{\psi}q_T^{\psi}(\phi_T^*)$,
we approximate the function $\varphi_{T-1}=w_{T-1}^{\psi}q_T^{\psi}(\phi_T)$ using the least squares problem
\begin{align}\label{eqn:least_squares}
	\phi_{T-1}=\arg\min_{f\in\mathbb{F}_{T-1}}\sum_{n=1}^N\left( \log f(s_{T-2}^{(a_{T-2}^{(n)})},s_{T-1}^{(n)}) -
	\log \varphi_{T-1}(s_{T-2}^{(a_{T-2}^{(n)})},s_{T-1}^{(n)})\right)^2,
\end{align}
where $\mathbb{F}_{T-1}$ is another function class to be chosen. We then proceed in the same manner until
the initial time $0$ to obtain a sequence of functions $\phi=(\phi_t)_{t=0}^T$ approximating the optimal refinement $\phi^*$.
The refined policy is then given by the update $\psi\cdot\phi=(\psi_t\cdot\phi_t)_{t=0}^T$.

The above ADP scheme is summarized in Algorithm \ref{alg:ADP}, where we also give alternative expressions involving
the proposals $(q_t)_{t=0}^T$ and weight functions $(w_t)_{t=0}^T$ of the uncontrolled SMC,
which we will use for our numerical implementation.
The cost of running Algorithm \ref{alg:ADP} is $\mathcal{O}(T(NC_{\mathrm{evaluate}} + C_{\mathrm{approx}}))$,
where $C_{\mathrm{evaluate}}(d)$ is the cost of evaluating one of the weight functions in Equation \eqref{eqn:weights_psi},
and $C_{\mathrm{approx}}(N,d)$ is the cost of each least squares approximation.
In the case of linear least squares, $C_{\mathrm{approx}}$ would be linear in $N$.
By selecting parametric function classes $(\mathbb{F}_t)_{t=0}^T$ that depend on coefficients $(\beta_t)_{t=0}^T$,
only the estimated coefficients parameterizing the policies have to be stored in practice.
We will also choose function classes that are closed under multiplication, so that
the refined policy $\psi\cdot\phi$ also lie in the same function classes,
with coefficients that can be easily updated.
The richness of these function classes will determine the quality of $\psi\cdot\phi$
as an approximation of the optimal policy $\psi^*$.
\begin{algorithm}[h]
\protect\caption{\footnotesize{Approximate dynamic programming at inverse temperature $\lambda\in[0,1]$} \label{alg:ADP}}
{\footnotesize
\textbf{Input}: current policy $\psi=(\psi_t)_{t=0}^T$ and output of controlled SMC (Algorithm \ref{alg:cSMC}).

(1) For time $t=T,\ldots,1$,

\begin{description}[itemsep=0pt,parsep=0pt,topsep=0pt,labelindent=0.5cm]

	\item (1a) If $t=T$, for particle $n=1,\ldots,N$, set
$ \varphi_T(s_{T-1}^{(a_{T-1}^{(n)})},s_T^{(n)}) = w_T^{\psi}(s_{T-1}^{(a_{T-1}^{(n)})},s_T^{(n)})
= \frac{w_T(s_{T-1}^{(a_{T-1}^{(n)})},s_T^{(n)})}{\psi_T(s_{T-1}^{(a_{T-1}^{(n)})},s_T^{(n)})}$.

	\item (1b) If $t<T$, for particle $n=1,\ldots,N$, set
$$\varphi_t(s_{t-1}^{(a_{t-1}^{(n)})},s_t^{(n)}) =
w_t^{\psi}(s_{t-1}^{(a_{t-1}^{(n)})},s_t^{(n)})q_{t+1}^{\psi}(\phi_{t+1}|s_t^{(n)}) =
\frac{w_t(s_{t-1}^{(a_{t-1}^{(n)})},s_t^{(n)})q_{t+1}(\psi_{t+1}\cdot\phi_{t+1}|s_t^{(n)})}
{\psi_t(s_{t-1}^{(a_{t-1}^{(n)})},s_t^{(n)})}.$$

	\item (1c) Fit the function
$ \phi_t = \arg\min_{f\in\mathbb{F}_t}\sum_{n=1}^N\left(\log f(s_{t-1}^{(a_{t-1}^{(n)})},s_t^{(n)})
-\log\varphi_t(s_{t-1}^{(a_{t-1}^{(n)})},s_t^{(n)}) \right)^2.$

\end{description}

%\textcompwordmark{}

(2) For time $t=0$.

\begin{description}[itemsep=0pt,parsep=0pt,topsep=0pt,labelindent=0.5cm]
	\item (2a) For particle $n=1,\ldots,N$, set
	$\varphi_0(s_0^{(n)})= w_0^{\psi}(s_0^{(n)})q_1^{\psi}(\phi_1|s_0^{(n)}) =
	q_0(\psi_0)q_1(\psi_1\cdot\phi_1|s_0^{(n)})/\psi_0(s_0^{(n)}).$

	\item (2b) Fit the function
	$ \phi_0 = \arg\min_{f\in\mathbb{F}_0}\sum_{n=1}^N\left(\log f(s_0^{(n)})
	-\log\varphi_0(s_0^{(n)}) \right)^2.$

\end{description}

\textbf{Output}: refined policy $\psi\cdot\phi=(\psi_t\cdot\phi_t)_{t=0}^T$.}
\end{algorithm}

We can then run a controlled SMC (Algorithm \ref{alg:cSMC}) with proposals defined by the refined policy $\psi\cdot\phi$ to obtain better likelihood and state estimates. Using this SMC output, one could consider another round of policy refinement
with ADP (Algorithm \ref{alg:ADP}) to produce more efficient SMC estimates. This iterative procedure that alternates between ADP and SMC, with the inverse temperature fixed at the desired level of $\lambda=1$,
is studied by \citet{heng2020controlled}. However, for complex state-space models with strong non-linearities and highly informative observations, typically of dynamic structural macrofinance models, this approach can easily fail when the initial policy is far from optimality as the states used in the ADP approximation would be in the tails of the smoothing distribution in Equation \eqref{eqn:smoothing_distribution}. To prevent such numerical instabilities from manifesting, we propose a novel iterative procedure in the following section that performs policy refinement as the inverse temperature $\lambda$ gradually increases.

\subsection{Annealed controlled sequential Monte Carlo}\label{sec:ACSMC}

Let $(\lambda_i)_{i=0}^I$ denote an increasing inverse temperature schedule with $\lambda_0=0$
and $\lambda_I\leq 1$ that will be pre-specified. Although the inverse temperature of $\lambda=1$ is desired
to estimate the quantities in Equations \eqref{eqn:likelihood} and \eqref{eqn:smoothing_distribution},
approximating the intermediate quantities in Equation \eqref{eqn:tempered_smoothing} for $0\leq\lambda<1$ is also
useful when we consider parameter inference in Section \ref{sec: parameter_inference}.

We begin by running the uncontrolled SMC at inverse temperature $\lambda_0=0$, defined by
proposals $(q_t)_{t=0}^T$ and weight functions $(w_t)_{t=0}^T$,
and initializing the policy $\psi^{(0)}=(\psi_t^{(0)})_{t=0}^T$ as constant one functions, i.e. $\psi_t^{(0)}=1$ for $t=0,\ldots,T$.
In the case of the BPF, this simply generates $N$ trajectories
from the latent process of Equation \eqref{eqn:latent_process} and initializes the policy at optimality for $\lambda_0=0$.
Subsequently, for iteration $i=1,\ldots,I$, we construct new proposals for the next inverse temperature $\lambda_i$ by
refining the policy $\psi^{(i-1)}$ from the previous inverse temperature $\lambda_{i-1}$ using ADP.
More precisely, we employ Algorithm \ref{alg:ADP} using $\psi^{(i-1)}$ as the current policy and the previous SMC output
at $\lambda_{i-1}$. Writing $\psi^{(i)}$ as the refined policy, we then run a controlled SMC (Algorithm \ref{alg:cSMC})
at inverse temperature $\lambda_i$ with new proposals defined by $\psi^{(i)}$.
We provide an algorithmic summary of the above methodology in Algorithm \ref{alg:ac-SMC},
which we refer to as annealed controlled SMC (AC-SMC).
\begin{algorithm}[h]
\protect\caption{\footnotesize{Annealed controlled sequential Monte Carlo }\label{alg:ac-SMC}}
{\footnotesize
\textbf{Input}: number of particles $N$ and inverse temperature schedule $(\lambda_i)_{i=0}^I$.

(1) For iteration $i=0$.

\begin{description}[itemsep=0pt,parsep=0pt,topsep=0pt,labelindent=0.5cm]
	\item (1a) Set policy $\psi^{(0)}$ as constant one functions.

	\item (1b) Run uncontrolled SMC at inverse temperature $\lambda_0$ with proposals $(q_t)_{t=0}^T$ and weights $(w_t)_{t=0}^T$.

\end{description}

(2) For iteration $i=1,\ldots,I$.

\begin{description}[itemsep=0pt,parsep=0pt,topsep=0pt,labelindent=0.5cm]

	\item (2a) Run ADP (Algorithm \ref{alg:ADP}) at inverse temperature $\lambda_i$ with $\psi^{(i-1)}$ as current policy
	and previous SMC output to obtain refined policy $\psi^{(i)}$.

	\item (2b) Run controlled SMC (Algorithm \ref{alg:cSMC}) at inverse temperature $\lambda_i$ with policy $\psi^{(i)}$.

\end{description}
\textbf{Output}: policy $\psi^{(I)}$ and controlled SMC output at inverse temperature $\lambda_I$.}
\end{algorithm}

We now explain the rationale behind our proposed methodology and how it differs from existing works.
Suppose $\psi^{(i-1)}$ is a good approximation of the optimal policy at $\lambda_{i-1}$, which is the case at initialization.
The states generated by the resulting controlled SMC at $\lambda_{i-1}$ (Step 2b) will be in regions of high probability mass
under the smoothing distribution $p(ds_{0:T} | y_{1:T}, \theta, \lambda_{i-1})$.
If the inverse temperature increase is small, these states should also be in high probability regions under
$p(ds_{0:T} | y_{1:T}, \theta, \lambda_{i})$ and therefore are good support points to learn the refined policy
$\psi^{(i)}$ using the ADP algorithm (Step 2a). The optimal refinement of $\psi^{(i-1)}$
is given by Equation \eqref{eqn:optimal_phi} and can be rewritten as
\begin{align}\label{eqn:rewrite_optimal_phi}
	\phi_T^*(s_{T-1},s_T;\theta,\lambda_i) ~&= w_T^{\psi^{(i-1)}}(s_{T-1},s_T;\theta,\lambda_{i-1})
	\frac{w_T(s_{T-1},s_T;\theta,\lambda_i)}{w_T(s_{T-1},s_T;\theta,\lambda_{i-1})} ,\\
	\phi_t^*(s_{t-1},s_t;\theta,\lambda_i) ~&=~ w_t^{\psi^{(i-1)}}(s_{t-1},s_t;\theta,\lambda_{i-1})
	\frac{w_t(s_{t-1},s_t;\theta,\lambda_{i})}{w_t(s_{t-1},s_t;\theta,\lambda_{i-1})}
	q_{t+1}^{\psi^{(i-1)}}(\phi_{t+1}^*|s_{t},\theta),\notag\\
	\phi_0^*(s_0;\theta,\lambda_i) ~&=~ w_0^{\psi^{(i-1)}}(s_0;\theta)
	q_1^{\psi^{(i-1)}}(\phi_1^*|s_0,\theta),\notag
\end{align}
for $t = T-1,\ldots,1$. In the preceding equations, we reintroduce parameter and temperature dependence for clarity
and assume that the choice of initial proposals $(q_t)_{t=0}^T$ are not temperature dependent.

The weight functions $w_0^{\psi^{(i-1)}}(s_0;\theta)$ and $w_t^{\psi^{(i-1)}}(s_{t-1},s_t;\theta,\lambda_{i-1})$
for $t=1,\ldots,T$ are determined by the quality of the previous ADP approximation and is equal to the
residual of the least squares approximation at time $t=0,1,\ldots,T$ in the logarithmic scale.
These weight functions should therefore be close to constant functions as we have assumed
near optimality of $\psi^{(i-1)}$. The ratio of weights $w_t(s_{t-1},s_t;\theta,\lambda_{i})/w_t(s_{t-1},s_t;\theta,\lambda_{i-1})$, 	
which accounts for the increase in inverse temperature, reduces to $g_{\theta}(y_t|s_{t-1},s_t)^{\lambda_i-\lambda_{i-1}}$
in the case of the BPF. If the inverse temperature increment is small, we can expect these ratios to be
well-behaved with small fluctuations. By an inductive argument, the conditional expectation
$q_{t}^{\psi^{(i-1)}}(\phi_{t}^*|s_{t-1},\theta)$ should also have the same behaviour.
Hence we can expect the optimal refinement of $\psi^{(i-1)}$ to be well-approximated by simple functions.
Assuming that the chosen function classes are rich enough and the number of samples $N$ is sufficiently large
to obtain good least squares estimation, a good approximation of the optimal policy at $\lambda_{i}$ will then
ensure stability of the next iteration.

In contrast to the policy refinement procedure in \citet{heng2020controlled}, which is based on repeated least squares fitting
of residuals with the inverse temperature fixed at $\lambda=1$, our methodology can be seen as an extension
that allows one to incorporate changes in temperature. The use of annealing to alleviate numerical instabilities was mentioned in earlier work by \citet{scharth2016particle}, but the main and crucial difference is that their policy learning procedure does not change across iterations to exploit the policies learnt at lower inverse temperatures.

In practical implementations of AC-SMC (Algorithm \ref{alg:ac-SMC}), we can monitor the performance of controlled SMC
at each inverse temperature (Step 2b) to evaluate the quality of the ADP approximation (Step 2a).
Using the above-mentioned relationship between the weight functions of controlled SMC and the residuals in ADP,
we can inspect the variance of the SMC weights $(W_t^{(n)})_{n=1}^N$, using for example the effective sample size
criterion \citep*{kong1994sequential}, defined at each time $t=0,1,\ldots,T$ as $1/\sum_{n=1}^N(W_t^{(n)})^2$.
Using the output of Algorithm \ref{alg:ac-SMC}, we obtain an unbiased estimator of the likelihood $p(y_{1:T}|\theta,\lambda_I)$,
and an approximate sample from the smoothing distribution
$p(ds_{0:T}|y_{1:T},\theta,\lambda_I)$.

\subsection{Analysis of annealed controlled sequential Monte Carlo}\label{sec:analysis1}
We analyze the performance of AC-SMC at each iteration, by supposing that we have a current policy $\psi$,
and an approximation $\phi$ of the optimal refinement of $\psi$,
obtained using ADP (Algorithm \ref{alg:ADP}) at inverse temperature $\lambda\in[0,1]$.
To measure the effect policy refinement has on controlled SMC (Algorithm \ref{alg:cSMC}), we consider the
Kullback--Leibler (KL) divergence from the proposal distribution $q^{\psi\cdot\phi}(ds_{0:T})$, defined in
Equations \eqref{eqn:law_proposal} and \eqref{eqn:proposal_psi}, to
the smoothing distribution $p(ds_{0:T}|y_{1:T},\theta,\lambda)$ in Equation \eqref{eqn:tempered_smoothing}, defined as
{\small
\begin{align}
\mathrm{KL}\left(p(ds_{0:T}|y_{1:T},\theta,\lambda)  ~|~ q^{\psi\cdot\phi}(ds_{0:T}) \right)=\int_{\mathbb{S}^{T+1}}\log\left(\frac{p(s_{0:T}|y_{1:T},\theta,\lambda )}{q^{\psi\cdot\phi}(s_{0:T})}\right){p(s_{0:T}|y_{1:T},\theta,\lambda )}ds_{0:T}.	
\end{align}}This choice of KL divergence is motivated by the fact that it characterizes the quality of our proposal distribution in the context of importance sampling \citep{chatterjee2018sample}, i.e., small KL divergence is both sufficient and necessary for
good importance sampling approximations.

We first introduce some notation. We denote by $\mu_t^*$ and $\mu_t^{\psi\cdot\phi}$ the marginal distributions
of the smoothing and proposal distributions at time $t=0,\ldots,T$, respectively.
Let $(q_t^*)_{t=0}^T$ denote the optimal proposals under the optimal policy $\psi^*$,
and define for $t=1,\ldots,T$ the joint smoothing distribution
$(\mu_{t-1}^*\times q_t^*)(ds_{t-1},ds_t)=\mu_{t-1}^*(ds_{t-1})q_t^*(ds_t|s_{t-1})$
on $\mathbb{S}\times\mathbb{S}$.
Note that the above quantities depend on the parameters $\theta\in\Theta$ and inverse temperature $\lambda\in[0,1]$, even though these dependencies are not made explicit for notational simplicity.
Our first result provides a decomposition of the KL divergence,
in terms of the logarithmic differences between
$\phi^*$ and $\phi$ under the marginal distributions
$\xi_0^*=q_0^*$, $\xi_0^{\psi\cdot\phi}=q_0^{\psi\cdot\phi}$, and
$\xi_t^*=\mu_{t-1}^*\times q_t^*$,
$\xi_t^{\psi\cdot\phi}=\mu_{t-1}^*\times q_t^{\psi\cdot\phi}$ for $t=1,\ldots,T$.

In what follows, all proofs are provided in the Appendix \ref{sec:proofsanneal}.

\begin{proposition}\label{prop:KL_decomp}
For any current policy $\psi = (\psi_t)_{t=0}^T$ and an approximation $\phi = (\phi_t)_{t=0}$ of $\phi^* = (\phi_t^*)_{t=0}$, the optimal refinement of $\psi$, the
KL divergence from $q^{\psi\cdot\phi}(ds_{0:T})$ to $p(ds_{0:T}|y_{1:T},\theta,\lambda)$ satisfies
{\small
\begin{align}\label{eqn:KL_decomposition}
	&\mathrm{KL}\left(p(ds_{0:T}|y_{1:T},\theta,\lambda)  ~|~ q^{\psi\cdot\phi}(ds_{0:T}) \right)
	\leq \sum_{t=0}^T\xi_t^*(\log(\phi_t^*/\phi_t)) + \xi_t^{\psi\cdot\phi}(\log(\phi_t/\phi_t^*)).
\end{align}}
\end{proposition}

Next, we will introduce some assumptions which are needed to derive upper bounds
of the terms in Equation \eqref{eqn:KL_decomposition}.
The states $(s_t^{(n)})_{t=0,n=1}^{T,N}$ and ancestors $(a_t^{(n)})_{t=0,n=1}^{T-1,N}$ from
controlled SMC with current policy $\psi$ define the following empirical measures
{\small
\begin{align}
	\nu_0^{\psi,N}(ds_0)=\frac{1}{N}\sum_{n=1}^N\delta_{s_0^{(n)}}(ds_0),\quad
	\nu_t^{\psi,N}(ds_{t-1},ds_t)=\frac{1}{N}\sum_{n=1}^N\delta_{(s_{t-1}^{a_{t-1}^{(n)}},s_t^{(n)})}(ds_{t-1},ds_t),
\end{align}}for $t=1,\ldots,T$, that are used in the least squares approximations in Equations \eqref{eqn:least_squares_terminal} and \eqref{eqn:least_squares}.
As the number of particles $N\rightarrow\infty$, these measures $(v_t^{\psi,N})_{t=0}^T$ are consistent approximations of $(v_t^{\psi})_{t=0}^T$, defined recursively as
\begin{align}
	\nu_0^{\psi}(ds_0) ~&=~ q_0^{\psi}(ds_0),\quad
	\nu_1^{\psi}(ds_0,ds_1) ~=~ \frac{\nu_0^{\psi}(ds_0)w_0^{\psi}(s_0)}
	{\nu_0^{\psi}(w_0^{\psi})}q_1^{\psi}(ds_1|s_{0}),\\
	\nu_t^{\psi}(ds_{t-1},ds_t) ~&=~ \frac{\int_{\mathbb{S}}\nu_{t-1}^{\psi}(ds_{t-2},ds_{t-1})
	w_{t-1}^{\psi}(s_{t-2},s_{t-1})}{\nu_{t-1}^{\psi}(w_{t-1}^{\psi})}q_t^{\psi}(ds_t|s_{t-1}),\quad t=2,\ldots,T.\notag
\end{align}
For each time $t=0,\ldots,T$, we define the $L^2$-norm $\| \varphi \|_{L^2(\nu_t^{\psi})}= \nu_t^{\psi}(\varphi^2)^{1/2}$ and the $L^2(\nu_t^{\psi})$ space\footnote{With equivalent classes defined by functions that agree $\nu_t^{\psi}$-almost everywhere.}
as the set of measurable functions $\varphi$ with $\| \varphi \|_{L^2(\nu_t^{\psi})}<\infty$.
To study ADP (Algorithm \ref{alg:ADP}) in the infinite particle regime, we define $L^2$-projections under
the function class $\mathbb{F}_t$ and the distribution $\nu_t^{\psi}$
\begin{align}\label{eqn:L2_projection}
	P_t^{\psi}\varphi = \arg\min_{f\in\mathbb{F}_t}\|\log f - \log \varphi\|_{L^2(\nu_t^{\psi})}^2,
\end{align}
for $\log \varphi\in L^2(\nu_t^{\psi})$. The following assumption concerns our choice of function classes
$(\mathbb{F}_t)_{t=0}^T$ within ADP.

\begin{assumption}\label{ass:function_classes}
	The function classes $(\mathbb{F}_t)_{t=0}^T$ satisfy:
	\begin{enumerate}[label=(\roman*)]
		\item $\log\mathbb{F}_t$ is a closed linear subspace of $L^2(\nu_t^{\psi})$ for $t=0,\ldots,T$;
		
		\item $\sup_{f\in \mathbb{G}_t^{\psi}}\|\log P_t^{\psi}f - \log f\|_{L^2(\nu_t^{\psi})}
		\leq e_t^{\psi}<\infty$ for $t=0,\ldots,T$, where $\mathbb{G}_t^{\psi}=\{w_t^{\psi}q_{t+1}^{\psi}(\varphi)
		: \varphi\in\mathbb{F}_{t+1}\}$ for $ t=0,\ldots,T-1$ and $\mathbb{G}_T^{\psi}=\{w_T^{\psi}\}$.
	\end{enumerate}
\end{assumption}
Assumption \ref{ass:function_classes}$(i)$ ensures the existence of a unique projection in Equation \eqref{eqn:L2_projection},
which can be relaxed by letting $L^2$-projections denote the set of minimizers.
The residual errors $(e_t^{\psi})_{t=0}^T$ in Assumption \ref{ass:function_classes}$(ii)$ describes
the flexibility of the chosen function classes for our purpose of learning policies.
The next assumption pertains to the relationships between the distributions
$\xi_t^{*}$, $\xi_t^{\psi\cdot\phi}$ and $\nu_t^{\psi}$.
For distributions $\mu$ and $\nu$ defined on a common measurable space, we will write $\mu\ll\nu$ if $\mu$ is
absolutely continuous with respect to $\nu$.
\begin{assumption}\label{ass:marginals}
	There exist positive constants $(C_t)_{t=0}^T$ and $(M_t)_{t=0}^T$ such that:
	\begin{enumerate}[label=(\roman*)]
		\item $\xi_t^{*}\ll\nu_t^{\psi}$ with density satisfying $\xi_t^{*}/\nu_t^{\psi}\leq C_t$ for $t=0,\ldots,T$;
		
		\item $\xi_t^{\psi\cdot\phi}\ll\xi_t^{*}$
		with density satisfying $\xi_t^{\psi\cdot\phi}/\xi_t^{*}\leq M_t$ for $t=0,\ldots,T$.
		
	\end{enumerate}
\end{assumption}

Assumption \ref{ass:marginals}$(i)$ requires the current policy $\psi$ to induce a reasonably good approximation
of the smoothing marginals.
As such a condition is unlikely to be satisfied when $\psi$ is given by constant one functions and
the inverse temperature $\lambda=1$, this motivates the use of AC-SMC which allows policy refinement
as $\lambda$ gradually increases.
Assumption \ref{ass:marginals}$(ii)$ is a condition on the quality of the proposals
$(q_t^{\psi\cdot\phi})_{t=0}^T$ under the refined policy $\psi\cdot\phi$.
We now state our main result, which gives recursive bounds of the logarithmic differences
appearing in the KL upper bound in Equation \eqref{eqn:KL_decomposition}.

\begin{theorem}\label{thm:policy_learning}
Under Assumptions \ref{ass:function_classes} and \ref{ass:marginals},
$\varepsilon_t^{*}=\xi_t^*(\log(\phi_t^*/\phi_t))$ and $\varepsilon_t^{\psi\cdot\phi}=\xi_t^{\psi\cdot\phi}(\log(\phi_t/\phi_t^*))$ for $t=0,\ldots,T$ satisfy
the backward recursions
\begin{align}\label{eqn:error_recursion}
	\varepsilon_t^{*} \leq \varepsilon_{t+1}^{*} + C_t e_t^{\psi},\quad
	\varepsilon_{t}^{\psi\cdot\phi} \leq M_t\varepsilon_{t+1}^{\psi\cdot\phi} + C_t M_t e_t^{\psi} ,
	\quad t=0,\ldots,T-1,
\end{align}
with $\varepsilon_T^{*}\leq C_Te_T^{\psi}$ and $\varepsilon_t^{\psi\cdot\phi}\leq C_TM_Te_T^{\psi}$ at
the terminal time $T$.
\end{theorem}

Equation \eqref{eqn:error_recursion} shows how the residual errors $(e_t^{\psi})_{t=0}^T$
in the $N=\infty$ limit propagate backward in time. These errors can be small if the function classes
are sufficiently rich, in which case, Theorem \ref{thm:policy_learning} and Proposition \ref{prop:KL_decomp}
would imply good performance of AC-SMC.
%Our results are complementary to \citet[Section 5.1]{heng2020controlled}, where function approximation errors are studied in the natural scale.

% discuss result

\section{Parameter inference}\label{sec: parameter_inference}

In this section, we consider the problem of parameter and state inference in the Bayesian framework.
Let $p(d\theta)=p(\theta)d\theta$ denote our prior distribution on the parameter space $\Theta$.
We develop a new methodology to approximate the posterior distribution of the parameters and latent states
\begin{equation}\label{eqn:posterior_parameters_states}
	p(d\theta, ds_{0:T} | y_{1:T}) = p(d\theta|y_{1:T}) p(ds_{0:T} |y_{1:T}, \theta)
	= \frac{p(d\theta)p(y_{1:T} | \theta)}{p(y_{1:T})}p(ds_{0:T} |y_{1:T}, \theta),
\end{equation}
and the model evidence $p(y_{1:T})=\int_{\Theta}p(d\theta)p(y_{1:T} | \theta)$. Our approach is to
employ an SMC sampler \citep{del2006sequential} to sequentially approximate distributions and their normalizing constants along the path
\begin{eqnarray}\label{eqn:path_parameters_states}
	p(d\theta, ds_{0:T} | y_{1:T},\lambda) &=& p(d\theta|y_{1:T},\lambda) p(ds_{0:T} |y_{1:T}, \theta,\lambda)\notag\\
	& =&
	\frac{p(d\theta)p(y_{1:T} | \theta, \lambda)}{p(y_{1:T}|\lambda)}p(ds_{0:T} |y_{1:T}, \theta,\lambda)
\end{eqnarray}
for $\lambda\in[0,1]$, in which a nested AC-SMC method (Algorithm \ref{alg:ac-SMC}) is exploited to approximate the likelihood
$p(y_{1:T}|\theta,\lambda)$ and smoothing distribution $p(ds_{0:T} |y_{1:T}, \theta,\lambda)$. As the inverse temperature $\lambda$ increases from zero to one, Equation \eqref{eqn:path_parameters_states} bridges between the prior distribution
$p(d\theta)p(ds_{0:T}|\theta)$ and the posterior distribution $p(d\theta, ds_{0:T} | y_{1:T})$ in Equation
\eqref{eqn:posterior_parameters_states}, and the normalization constant $p(y_{1:T}|\lambda)=\int_{\Theta}p(d\theta)p(y_{1:T} | \theta,\lambda)$ goes from one to the model evidence $p(y_{1:T})$.

Our approach falls in the class of SMC$^2$ methods developed by \citet*{fulop2013efficient}, \citet*{chopin2013smc2}, and \citet{duan2015density}, but differs in important aspects. In Section \ref{subsec:smc2}, we provide a description of our adaptive SMC$^2$ algorithm and discuss how it differs from and relates to the existing literature. Theoretical justifications of the algorithm are presented in Section \ref{subsec:consistency}.

\subsection{Adaptive SMC$^2$}\label{subsec:smc2}

Our proposed SMC$^2$ methodology, detailed in Algorithm \ref{alg:SMC-sq}, builds on AC-SMC in Algorithm \ref{alg:ac-SMC} and a conditional  implementation of controlled SMC (Algorithm \ref{alg:condSMC} in Appendix \ref{appendix:cond_smc}).
\begin{algorithm}[h]
\protect\caption{\footnotesize{Adaptive sequential Monte Carlo$^2$} \label{alg:SMC-sq}}
{\footnotesize
\textbf{Input}: number of parameter particles $P$, state particles $N$ and ESS threshold $\kappa_{\mathrm{ESS}}$.

(1) For iteration $i=0$.

\begin{description}[itemsep=0pt,parsep=0pt,topsep=0pt,labelindent=0.5cm]

	\item (1a) Set inverse temperature as $\lambda_0=0$.

	\item (1b) For particle $p=1,\ldots,P$, sample a parameter $\theta^{(p)}\sim p(d\theta)$ from the prior and a trajectory
	$s_{0:T}^{(p)}=(s_{t}^{(p)})_{t=0}^T\sim p(ds_{0:T}|\theta^{(p)})$ from the latent process.
	
%	\item (1c) For particle $p=1,\ldots,P$, set policy $\psi^{(p)}$ as constant one functions.

\end{description}

(2) For iteration $i=1,2,\ldots$, if $\mathrm{ESS}\left(\Omega_i^{(1)}(1),\ldots,\Omega_i^{(P)}(1)\right) \geq \kappa_{\mathrm{ESS}} P$, set
      next inverse temperature as $\lambda_i = 1$, else determine $\lambda_i$ as the value of $\lambda\in(\lambda_{i-1},1)$
	that solves $\mathrm{ESS}\left(\Omega_i^{(1)}(\lambda),\ldots,\Omega_i^{(P)}(\lambda)\right) = \kappa_{\mathrm{ESS}} P$ using bisection method.

(3) For particle $p=1,\ldots,P$.
	
	\begin{description}[itemsep=0pt,parsep=0pt,topsep=0pt,labelindent=0.5cm]
		\item (3a) Compute unnormalized weights $\omega_i^{(p)}= \omega_i^{(p)}(\lambda_i)$ and
		normalized weights $\Omega_i^{(p)}= \Omega_i^{(p)}(\lambda_i)$.

		\item (3b) Sample ancestor $\alpha_{i}^{(p)}\sim r(\cdot|\Omega_{i}^{(1)},\ldots,\Omega_{i}^{(P)})$, and update parameter as $\theta^{(p)}=\theta^{(\alpha_{i}^{(p)})}$ and trajectory as $s_{0:T}^{(p)}=s_{0:T}^{(\alpha_{i}^{(p)})}$.
		
		\item (3c) Learn policy $\psi^{(p)}$ for parameter $\theta^{(p)}$ at inverse temperature $\lambda_i$.
%		Run ADP (Algorithm \ref{alg:ADP}) at inverse temperature $\lambda_i$ with $\psi^{(i-1)}$ as current policy
%	and previous SMC output to obtain refined policy $\psi^{(i)}$.

		\item (3d) Run conditional SMC (Algorithm \ref{alg:condSMC} in the Appendix) with policy $\psi^{(p)}$ and reference trajectory $s_{0:T}^{(p)}$
		at parameter $\theta^{(p)}$ and inverse temperature $\lambda_i$ to update likelihood estimator
		$\hat{p}(y_{1:T}|\theta^{(p)},\lambda_i)$ and trajectory $s_{0:T}^{(p)}$.
		
       \end{description}
		
(4) Construct proposal transition kernel $h_i$ based on current set of parameters $(\theta^{(p)})_{p=1}^P$. For $k = 1,\ldots, K$ PMMH moves.
		\begin{description}[itemsep=0pt,parsep=0pt,topsep=0pt,labelindent=0.5cm]
			\item (4a) Sample a parameter proposal $\theta^{(*)}\sim h_i(\cdot |\theta^{(p)})$.
			\item (4b) Learn policy $\psi^{(*)}$ for proposed parameter $\theta^{(*)}$ at inverse temperature $\lambda_i$.
			\item (4c) Run controlled SMC (Algorithm \ref{alg:cSMC}) with policy $\psi^{(*)}$ at proposed parameter $\theta^{(*)}$
			and inverse temperature $\lambda_i$ to obtain likelihood estimator
			$\hat{p}(y_{1:T}|\theta^{(*)},\lambda_i)$ and trajectory $s_{0:T}^{(*)}$.
		
			\item (4d) Update parameter $\theta^{(p)}=\theta^{(*)}$, trajectory $s_{0:T}^{(p)}=s_{0:T}^{(*)}$ and
			likelihood estimator $\hat{p}(y_{1:T}|\theta^{(p)},\lambda_i) = \hat{p}(y_{1:T}|\theta^{(*)},\lambda_i)$ with probability
			$\alpha(\theta^{(*)}|\theta^{(p)},\lambda_i)$ in Equation \eqref{eqn:PMMH_acceptprob}.		
			
		\end{description}

(5) If $\lambda_i=1$, set number of iterations as $I=i$ and terminate iteration.\\
(6) Compute the model evidence estimator $\hat{p}(y_{1:T})$ as in Equation \eqref{pevidencehat}.
	
\textbf{Output}: parameters $(\theta^{(p)})_{p=1}^P$, trajectories $(s_{0:T}^{(p)})_{p=1}^P$ and model evidence estimator $\hat{p}(y_{1:T})$.}
\end{algorithm}

In Step 1, we initialize the algorithm at inverse temperature $\lambda_0=0$ by sampling $P$ parameters
and state trajectories $(\theta^{(p)},s_{0:T}^{(p)})_{p=1}^P$ from the prior distribution $p(d\theta, ds_{0:T} | y_{1:T},\lambda_0)=p(d\theta)p(ds_{0:T}|\theta)$. Subsequently, for iteration $i\geq 1$, we determine the next inverse temperature $\lambda_i\in(0,1]$ in Step 2, so that the previous bridging distribution $p(d\theta, ds_{0:T} | y_{1:T},\lambda_{i-1})$ provides a good importance sampling approximation of the next bridging distribution $p(d\theta, ds_{0:T} | y_{1:T},\lambda_i)$.
In this context, we measure the quality of importance sampling approximations using the effective sample size (ESS) criterion
\begin{align}\label{eqn:adapt_weights}
\mathrm{ESS}(\Omega_i^{(1)}(\lambda),\ldots,\Omega_i^{(P)}(\lambda))=1 {\Big/} \sum_{p=1}^P\Omega_i^{(p)}(\lambda)^2,
\end{align}
which is based on the unnormalized weights
\begin{align}\label{eqn:adapt_weights}
	\omega_i^{(p)}(\lambda)
	= \omega(\theta^{(p)},s_{0:T}^{(p)}|\lambda_{i-1},\lambda)
	= \prod_{t=1}^Tg_{\theta^{(p)}}(y_t|s_{t-1}^{(p)},s_t^{(p)})^{\lambda-\lambda_{i-1}}
%\Omega_i^{(p)}(\lambda) = \frac{\omega_i^{(p)}(\lambda)}{\sum_{j=1}^P\omega_i^{(j)}(\lambda)}
\end{align}
and the normalized weights $\Omega_i^{(p)}(\lambda)=\omega_i^{(p)}(\lambda)/\sum_{j=1}^P\omega_i^{(j)}(\lambda)$.
The ESS lies between $1$ to $P$: the lower bound is attained when one sample has all the normalized weight,
while the upper bound is achieved when all samples have uniform weights.

We adapt the next inverse temperature using the following scheme:
\begin{align}\label{eqn:adapt_inversetemp}
	\lambda_i = \inf\left\lbrace\lambda\in(\lambda_{i-1},1] : \mathrm{ESS}\left(\Omega_i^{(1)}(\lambda),\ldots,\Omega_i^{(P)}(\lambda)\right)
	= \kappa_{\mathrm{ESS}} P \right\rbrace,
\end{align}
with the convention $\inf\emptyset = 1$ to accommodate the terminal iteration,
where $\kappa_{\mathrm{ESS}}\in(0,1)$ is a pre-specified threshold that controls the amount of weight degeneracy.
We refer readers to \citet[Section 3]{dai2020invitation} for discussions on the choice of $\kappa_{\mathrm{ESS}}$ and its impact
on the number of iterations required to reach the desired inverse temperature of $\lambda=1$. In contrast to the adaptation scheme proposed in \citet*{svensson2018learning}, which requires storing the entire SMC output for each parameter $\theta^{(p)}$, the unnormalized weight in Equation \eqref{eqn:adapt_weights} is considerably simpler as it only depends on
a single trajectory. Moreover, as it can be shown that our ESS criterion is a strictly decreasing and continuous function of $\lambda\in(\lambda_{i-1},1]$, Equation \eqref{eqn:adapt_inversetemp} can be implemented using a simple bisection routine. After determining $\lambda_i$, we compute the resulting importance weights (Step 3a) and perform resampling to focus our computational effort on more likely parameters and trajectories (Step 3b). For notational ease, we have used the same notation $(\theta^{(p)},s_{0:T}^{(p)})_{p=1}^P$ to denote the resulting set of samples after any operation.

In Step 3c, we then learn a policy $\psi^{(p)}$ for each resampled parameter $\theta^{(p)}$ to construct good proposal distributions that approximate the smoothing distribution $p(ds_{0:T} | y_{1:T},\theta^{(p)},\lambda_i)$ at inverse temperature $\lambda_i$.
This policy learning step is left intentionally general in the algorithm to accommodate various strategies
for different applications. A generic and useful recipe is to set $\psi^{(p)}$ as constant one functions for inverse temperatures $\lambda_i$ that are below a small pre-specified level $\lambda_*\in(0,1)$. The rationale here is that at low inverse temperatures, the performance of uncontrolled SMC would be adequate as the influence of the observations are limited. This allows us to efficiently rule out very unlikely parameters at early iterations of the algorithm. For larger inverse temperatures $\lambda_i>\lambda_*$, it is worthwhile spending the computational overhead to learn the optimal policies in more promising regions of the parameter space.
%In particular, we will employ AC-SMC (Algorithm \ref{alg:ac-SMC}) with a pre-specified inverse temperature schedule from zero to $\lambda_i$ for the LRR model; or LC-SMC (Algorithm \ref{alg:lc-SMC} in Appendix \ref{appendix:dsge_lc_smc}) with a pre-specified linearization schedule from one to zero for the DSGE model\footnote{In view of Step 4b of Algorithm \ref{alg:SMC-sq},
%Step 2b of Algorithm \ref{alg:ac-SMC} or \ref{alg:lc-SMC} is unnecessary in the last policy refinement iteration.}.

In Step 3d, given the policy $\psi^{(p)}$ for each parameter $\theta^{(p)}$, we then run a conditional implementation of controlled SMC
(Algorithm \ref{alg:condSMC} in Appendix \ref{appendix:cond_smc}) to obtain a likelihood estimator $\hat{p}(y_{1:T}|\theta^{(p)},\lambda_i)$ and a new trajectory $s_{0:T}^{(p)}$.
The distinctive feature in the conditional implementation
%The main difference between conditional SMC and controlled SMC (Algorithm \ref{alg:cSMC})
is that the input reference trajectory is conditioned to survive all resampling steps \citep*{andrieu2010particle}, which is necessary for the validity of our approach. The use of conditional SMC within SMC$^2$ was also considered by \citet*{chopin2013smc2}, but for the purpose of increasing the number of state particles $N$ as more observations are assimilated. As we are constructing better SMC proposals by learning optimal policies, the choice of $N$ is less crucial in our setting. Compared to the tempered likelihood approach of \citet{duan2015density}, our proposed methodology has the flexibility to alter the SMC configuration via conditional SMC and has the benefits of annealing. Note that at this stage, we only have to store the resampled parameters $(\theta^{(p)})_{p=1}^P$,
updated trajectories $(s_{0:T}^{(p)})_{p=1}^P$, and likelihood estimators $(\hat{p}(y_{1:T}|\theta^{(p)},\lambda_i))_{p=1}^P$.

Compared to having a single Markov chain to sample from the posterior distribution of Equation \eqref{eqn:posterior_parameters_states}, the ability to select tuning parameters of a proposal transition kernel $h_i$, based on existing samples that approximate $p(d\theta|y_{1:T},\lambda_i)$ in Step 4, is an advantage of the SMC$^2$ approach. We refer readers to \citet[Section 2.2]{dai2020invitation} and references therein for discussions of various adaptation rules. Our numerical experiments employ Gaussian random walk proposals with covariance matrices that are estimated using the current set of parameter particles $(\theta^{(p)})_{p=1}^P$. Step 4 describes $K\in\mathbb{N}$ particle marginal Metropolis--Hastings (PMMH) moves \citep*{andrieu2010particle} to improve the sample diversity of the resampled parameters. The combination of these steps define a Markov transition kernel that has $p(d\theta, ds_{0:T} | y_{1:T},\lambda_i)$ as its invariant distribution. As the efficiency of PMMH moves crucially depends on the variance of the likelihood estimator \citep{doucet2015efficient,sherlock2015efficiency}, we learn the optimal policy for each
proposed parameter $\theta^{(*)}$ (Step 4b), and run controlled SMC with the resulting policy $\psi^{(*)}$ (Step 4c) to obtain a lower variance likelihood estimator $\hat{p}(y_{1:T}|\theta^{(*)},\lambda_i)$ and a trajectory $s_{0:T}^{(*)}$ with a law that is closer to the
smoothing distribution $p(ds_{0:T} | y_{1:T}, \theta^{(*)}, \lambda_i)$. In Step 4d, the proposed parameter $\theta^{(*)}$, trajectory $s_{0:T}^{(*)}$, and likelihood estimator $\hat{p}(y_{1:T}|\theta^{(*)},\lambda_i)$ are then accepted according to the Metropolis--Hastings acceptance probability
\begin{equation}\label{eqn:PMMH_acceptprob}
			\alpha(\theta^{(*)}|\theta^{(p)},\lambda_i) = \min\left\lbrace 1, ~\frac{p(\theta^{(*)}) \hat{p}(y_{1:T}|\theta^{(*)},\lambda_i)  h_i(\theta^{(p)}|\theta^{(*)})}
			{p(\theta^{(p)}) \hat{p}(y_{1:T}|\theta^{(p)},\lambda_i) h_i(\theta^{(*)}|\theta^{(p)})} \right\rbrace.
\end{equation}

When the desired inverse temperature of $\lambda=1$ is reached, we terminate the SMC$^2$ algorithm (Step 5). The number of iterations $I$ required to complete the algorithm is random and depends on the choice of ESS threshold $\kappa_{\mathrm{ESS}}\in(0,1)$. Using the parameters and trajectories $(\theta^{(p)},s_{0:T}^{(p)})_{p=1}^P$ outputted by the algorithm, we can approximate posterior expectations $\pi(\varphi)=\int_{\Theta\times\mathbb{S}^{T+1}}\varphi(\theta,s_{0:T})p(d\theta, ds_{0:T} | y_{1:T})$, for any integrable function $\varphi:\Theta\times\mathbb{S}^{T+1}\rightarrow\mathbb{R}$, with the sample average $\hat{\pi}(\varphi)=P^{-1}\sum_{p=1}^P\varphi(\theta^{(p)},s_{0:T}^{(p)})$. As a by-product of the algorithm, we also have an estimator of the model evidence $\hat{p}(y_{1:T})$,
computed using the unnormalized weights in Step 6
\begin{equation}\label{pevidencehat}
    \hat{p}(y_{1:T})=\prod_{i=1}^I \left(\frac{1}{P}\sum_{p=1}^P \omega_i^{(p)} \right).
\end{equation}
Since estimators of the model evidence cannot be easily obtained with a single PMMH chain targeting Equation \eqref{eqn:posterior_parameters_states}, this is another strength of our SMC$^2$ approach.
In the next section, we will establish consistency properties of our estimators of posterior expectations and
the model evidence in the limit of the number of parameters particles $P\rightarrow\infty$, for any choice of the number of state particles $N>1$.

\subsection{Consistency properties of adaptive SMC$^2$ estimators}\label{subsec:consistency}
This section concerns asymptotic properties of our estimators of expectations under the posterior
distribution in Equation \eqref{eqn:posterior_parameters_states} and the model evidence $p(y_{1:T})$.
In the following, the notations $\stackrel{p.}{\rightarrow}$ and $\stackrel{d.}{\rightarrow}$ denote
convergence in probability and distribution, respectively.
We will establish the weak law of large numbers (WLLN)
\begin{align}\label{eqn:WLLN}
	\hat{\pi}(\varphi)\stackrel{p.}{\longrightarrow}\pi(\varphi),\quad
	\hat{p}(y_{1:T})\stackrel{p.}{\longrightarrow} p(y_{1:T}),
\end{align}
as the number of parameter particles $P\rightarrow\infty$, for any number of state particles $N>1$.
To quantify the rate of convergence, we will also seek the central limit theorems (CLT)
\begin{align}\label{eqn:CLT}
	\sqrt{P}(\hat{\pi}(\varphi)-\pi(\varphi))\stackrel{d.}{\longrightarrow}\mathcal{N}(0,\sigma^2(\varphi)), \quad
	\sqrt{P}(\hat{p}(y_{1:T})-p(y_{1:T}))\stackrel{d.}{\longrightarrow}\mathcal{N}(0,\sigma^2),
\end{align}
as $P\rightarrow\infty$ for any fixed $N>1$, where we denote the normal distribution
with mean vector $\mu$ and covariance matrix $\Sigma$ as $\mathcal{N}(\mu,\Sigma)$ and its density
by $x\mapsto\mathcal{N}(x;\mu,\Sigma)$.
Although we will not give explicit expressions of the asymptotic variances $\sigma^2(\varphi)$
and $\sigma^2$ due to the complexity of our SMC$^2$ algorithm,
it is clear how they can be derived in our proofs, provided in the Appendix \ref{sec:proofsmc2}.

We first consider Algorithm \ref{alg:SMC-sq} without adaptation, i.e.,
the inverse temperature schedule $(\lambda_i)_{i=0}^I$
and the proposal transition kernels $(h_{i})_{i=1}^{I}$ are pre-specified and not
determined on the fly in Steps 2 and 4.
This analysis serves to elucidate various aspects of our SMC$^2$ algorithm
without the added complication of adaptation.

\begin{theorem}\label{thm:smc_sq}
The estimators generated by SMC$^2$ in Algorithm \ref{alg:SMC-sq} without adaptation in Steps 2 and 4
satisfy the WLLNs and CLTs in Equations \eqref{eqn:WLLN} and \eqref{eqn:CLT}
for any bounded and measurable function $\varphi:\Theta\times\mathbb{S}^{T+1}\rightarrow\mathbb{R}$.
\end{theorem}

To study the adaptive SMC$^2$ algorithm, we first note that the adaptation scheme
in Step 2 and Equation \eqref{eqn:adapt_inversetemp} should be seen as a finite sample approximation of
a limiting deterministic inverse temperature schedule $(\lambda_i^*)$, defined by
the following scheme
\begin{align}
	\lambda_i^* = \inf\left\lbrace\lambda\in(\lambda_{i-1}^*,1] : \chi^2(p(d\theta, ds_{0:T} | y_{1:T},\lambda)~|~
	p(d\theta, ds_{0:T} | y_{1:T},\lambda_{i-1}^*))
	= \kappa_{\mathrm{ESS}}^{-1}-1 \right\rbrace,
\end{align}
initialized at $\lambda_0^*=0$, where the above $\chi^2$-divergence
\begin{align}
	&\chi^2(p(d\theta, ds_{0:T} | y_{1:T},\lambda)~|~
	p(d\theta, ds_{0:T} | y_{1:T},\lambda_{i-1}^*)) \\
	&= \frac{\int_{\Theta\times\mathbb{S}^{T+1}}\omega(\theta,s_{0:T}|\lambda_{i-1}^*,\lambda)^2
	p(d\theta, ds_{0:T} | y_{1:T},\lambda_{i-1}^*)}
	{(\int_{\Theta\times\mathbb{S}^{T+1}}\omega(\theta,s_{0:T}|\lambda_{i-1}^*,\lambda)
	p(d\theta, ds_{0:T} | y_{1:T},\lambda_{i-1}^*))^2} -1\notag
\end{align}
depends on the importance weight $\omega$ defined in Equation \eqref{eqn:adapt_weights}.
Next, we formalize the adaptive nature of Step 4
using the framework of \citet{beskos2016convergence}.
At iteration $i$, we consider a parametric family of proposal transition kernels
$h_i$, indexed by tuning parameters $\xi\in\mathbb{R}^{s}$.
The algorithm determines these tuning parameters by computing
the sample average $\xi_i=P^{-1}\sum_{p=1}^PS_i(\theta^{(p)})$ of a summary statistic
$S_i:\Theta\rightarrow\mathbb{R}^{s}$, based on a current set of
parameter particles $(\theta^{(p)})_{p=1}^P$ approximating $p(d\theta|y_{1:T},\lambda_{i-1})$.
Let $m_i$ denote the resulting Markov transition kernel on $\Theta\times\mathbb{S}^{T+1}$
by composing Steps 3c, 3d and 4 of Algorithm \ref{alg:SMC-sq}, which will be shown to have
$p(d\theta,ds_{0:T}|y_{1:T},\lambda_{i})$ as its invariant distribution.
We also define the non-negative kernel
$Q_i(d\tilde{\theta},d\tilde{s}_{0:T}|\theta,s_{0:T})=\omega(\theta,s_{0:T}|\lambda_{i-1},\lambda_i)
m_i(d\tilde{\theta},d\tilde{s}_{0:T}|\theta,s_{0:T})$ on $\Theta\times\mathbb{S}^{T+1}$.
To stress the dependence of $Q_i$ on the quantities $\zeta_i=(\lambda_{i-1},\lambda_i,\xi_i)$,
we will write $Q_{i,\zeta_i}$.
Our assumptions will involve the behaviour of the function $\zeta_i\mapsto Q_{i,\zeta_i}$
and its gradient $\zeta_i\mapsto \nabla_{\zeta_i}Q_{i,\zeta_i}$
at $\zeta_i^*=(\lambda_{i-1}^*,\lambda_i^*,\xi_i^*)$, which corresponds to
an idealized algorithm where the limiting inverse temperature schedule
$(\lambda_i^*)$ is employed, and tuning parameters are determined by the expectation
$\xi_i^*=\int_{\Theta}S_i(\theta)p(d\theta|y_{1:T},\lambda_{i-1})$.

\begin{assumption}\label{ass:adaptive_smc}
For time $t=1,\ldots,T$ and iteration $i\geq 1$,
the observation density $g_{\theta}(y_t|s_{t-1},s_t)$, summary statistic $S_i$,
importance weight $\omega(\theta,s_{0:T}|\lambda_{i-1},\lambda_i)$
and the non-negative kernel $Q_i$ satisfy:
	\begin{enumerate}[label=(\roman*)]
		\item  $(\theta,s_{t-1},s_t)\mapsto\log g_{\theta}(y_t|s_{t-1},s_t)$
		is bounded on $\Theta\times\mathbb{S}\times\mathbb{S}$ for any $y_t\in\mathbb{Y}$;
		
		\item $S_i:\Theta\rightarrow\mathbb{R}^{s}$ is bounded on $\Theta$;
		
		\item $(\theta,s_{0:T},\lambda_{i-1},\lambda_i)\mapsto
		\omega(\theta,s_{0:T}|\lambda_{i-1},\lambda_i)$ is continuous at $(\lambda_{i-1}^*,\lambda_i^*)$
		uniformly on $\Theta\times\mathbb{S}^{T+1}$;
		
		\item $(\theta,s_{0:T},\zeta_i)\mapsto Q_i(\varphi|\theta,s_{0:T})$ is continuous at $\zeta_i^*$
		uniformly on $\Theta\times\mathbb{S}^{T+1}$
		for any bounded and measurable function
		$\varphi:\Theta\times\mathbb{S}^{T+1}\rightarrow\mathbb{R}$;
		
		\item $(\theta,s_{0:T},\zeta_i)\mapsto \nabla_{\zeta_i}Q_i(\varphi|\theta,s_{0:T})$ is well-defined,
		bounded and continuous at $\zeta_i^*$ uniformly on $\Theta\times\mathbb{S}^{T+1}$
		for any function $\varphi:\Theta\times\mathbb{S}^{T+1}\rightarrow\mathbb{R}$.
	\end{enumerate}

\end{assumption}

\begin{theorem}\label{thm:adaptive_smc_sq}
Under Assumption \ref{ass:adaptive_smc},
the estimators generated by adaptive SMC$^2$ in Algorithm \ref{alg:SMC-sq}
satisfy the WLLNs and CLTs in Equations \eqref{eqn:WLLN} and \eqref{eqn:CLT}
for any bounded and measurable function $\varphi:\Theta\times\mathbb{S}^{T+1}\rightarrow\mathbb{R}$.
\end{theorem}

The crux of our arguments to establish Theorems \ref{thm:smc_sq} and \ref{thm:adaptive_smc_sq} is to first cast
the rather involved SMC$^2$ algorithm, without and with adaptation,
as a particular SMC sampler and adaptive SMC sampler in the frameworks
of \citet{del2006sequential} and \citet{beskos2016convergence} that
operate on an extended space. This enables us to then invoke convergence results
for standard SMC methods \citep{del2004feynman,chopin2004central} and apply them
to our SMC$^2$ estimators. Such an approach that exploits the specific properties of Algorithm \ref{alg:SMC-sq}
is not applicable to the SMC$^2$ method of \citet{duan2015density}, which is based on the tempered likelihood.

\section{Applications}\label{sec:apps}

%\subsection{A Prototypical DSGE Model}
%Dynamic stochastic general equilibrium (DSGE) models have been widely used in macroeconomic research and in central banks for forecasting and policy-making. We consider a prototypical small-scale New Keynesian DSGE model as described in \citet{an2007bayesian}.

Dynamic stochastic general equilibrium (DSGE) models have been widely used in macroeconomic research and in central banks for forecasting and policy-making. In this section, we consider a prototypical New Keynesian DSGE model that has been studied in \citet{woodford2003}, \citet{an2007bayesian}, and \citet{herbst2016}. This model has become a benchmark specification for the analysis of monetary policy. Variants of this model have also been studied in \citet*{aruoba2018res} and \citet{aruoba2020wp}.

In the Appendix \ref{LRRsmc2}, we provide another application to estimate a non-linear and non-Gaussian consumption-based long-run risk asset pricing model, in which consumption volatility is modelled using an autoregressive gamma process instead of an autoregressive process that is usually adopted in the standard long-run risk model \citep*[see,][]{bansal2004risks, bansal2012empirical}. This model has also been studied by \citet{fulop2020bayesian}.

\subsection{Model setup}

The model economy consists of a representative household, a final goods producing firm, a continuum of intermediate goods producing firms, and a monetary/fiscal authority. For the purpose of self-containedness, we provide the detailed model description in what follows. \bigskip

\noindent {\emph{Household.}} We consider a representative household who maximizes the following expected utility
\begin{equation}
   \mathbb{E}_t\left[ \sum_{s=0}^{\infty} \beta^s\left( \frac{(\mathsf{C}_{t+s}/\mathsf{A}_{t+s})^{1-\tau} - 1}{1-\tau} +
   \chi_M\log\left(\frac{\mathsf{M}_{t+s}}{\mathsf{P}_{t+s}}\right) - \mathsf{H}_{t+s}\right)\right],
\end{equation}
subject to the budget constraint
\begin{equation}
   \mathsf{P}_t\mathsf{C}_t + \mathsf{B}_t + \mathsf{M}_t + \mathsf{T}_t = \mathsf{P}_t\mathsf{W}_t\mathsf{H}_t
   + \mathsf{R}_{t-1}\mathsf{B}_{t-1} + \mathsf{M}_{t-1} + \mathsf{P}_t\mathsf{D}_t + \mathsf{P}_t\mathsf{SC}_t,
\end{equation}
where $\mathbb{E}_t$ denotes conditional expectation given information up to time $t$, and $\beta \in(0,1)$ is the discount factor. The household derives utility from consumption $\mathsf{C}_t$ relative to a habit shock and real money balances $\mathsf{M}_t/\mathsf{P}_t$, with $\mathsf{P}_t$ denoting the price of the final good, and derives disutility from hours worked $\mathsf{H}_t$. $\tau$ captures the household's level of risk aversion and its inverse, $1/\tau$, is the intertemporal elasticity of substitution. $\chi_M$ is a scale factor that determines steady-state real money balances. The household receives the real wage $\mathsf{W}_t$ in exchange for labor and has access to the bond market where $\mathsf{B}_t$ nominal government bonds are traded with gross interest $\mathsf{R}_t$. Furthermore, he/she receives residual real profits $\mathsf{D}_t$ from firms and has to pay lump-sum taxes $\mathsf{T}_t$. $\mathsf{SC}_t$ is the net cash inflow from trading a full set of state-contingent securities. %\textcolor{red}{Jeremy: $\chi_M$ is not defined }
\bigskip

\noindent {\emph{Firms.}} The final goods producing firms generate aggregate output $\mathsf{Y}_t$ by combining a continuum of intermediate goods $\mathsf{Y}_t(j)$ for $j\in[0,1]$. Under the assumption of perfect competition and free entry, the demand for intermediate goods with price $\mathsf{P}_t(j)$ is given by
\begin{equation}
	\mathsf{Y}_t(j) = \left(\frac{\mathsf{P}_t(j)}{\mathsf{P}_t}\right)^{-1/\nu}\mathsf{Y}_t,
\end{equation}
and the price of the final good is
\begin{equation}
	\mathsf{P}_t = \left( \int_0^1 \mathsf{P}_t(j)^{\frac{\nu-1}{\nu}}dj\right)^{\frac{\nu}{\nu-1}},
\end{equation}
where $1/\nu > 1$ represents the elasticity of demand for each intermediate good.

Intermediate good $j$ is produced by a monopolist who has the following linear production technology
\begin{equation}
     \mathsf{Y}_t(j) = \mathsf{A}_t\mathsf{N}_t(j),
\end{equation}
where $\mathsf{N}_t(j)$ is the labor input of firm $j$ and $\mathsf{A}_t$ is an exogenous productivity process that is common to all firms and evolves according to
\begin{equation}
    \log \mathsf{A}_t = \log \gamma + \log \mathsf{A}_{t-1} + \log \mathsf{z}_t, \quad \log \mathsf{z}_t = \rho_z \log \mathsf{z}_{t-1} + \varepsilon_{z,t}.
\end{equation}
In the above, $\mathsf{z}_t$ captures exogenous fluctuations of the technology growth rate and $\varepsilon_{z,t}\sim\mathcal{N}(0,\sigma_z^2)$.

Firms face nominal price rigidities in terms of quadratic price adjustment costs
\begin{equation}
	\mathsf{AC}_t(j) = \frac{\phi}{2}\left(\frac{\mathsf{P}_t(j)}{\mathsf{P}_{t-1}(j)} - \pi\right)^2\mathsf{Y}_t(j),
\end{equation}
where $\phi$ governs the price stickiness in the economy and $\pi$ is the steady-state rate of inflation $\pi_t$, defined as $\pi_t = \mathsf{P}_t/\mathsf{P}_{t-1}$. Each firm chooses its labor input $\mathsf{N}_t(j)$ and price $\mathsf{P}_t(j)$ to maximize the present value of its future profits
\begin{equation}
    \mathbb{E}_t\left[ \sum_{s=0}^{\infty} \beta^s\mathsf{Q}_{t+s|t}\left( \frac{\mathsf{P}_{t+s}(j)}{\mathsf{P}_{t+s}}\mathsf{Y}_{t+s}(j) - \mathsf{W}_{t+s}\mathsf{N}_{t+s}(j) - \mathsf{AC}_{t+s}(j)\right)\right],
\end{equation}
where $\mathsf{Q}_{t+s|t}$ is the time $t$ value of a unit of the consumption good in period $t+s$ to the household, which is treated as exogenous by the firm.\bigskip

\noindent \emph{Government policies.} The government consumes a stochastic fraction of aggregate output and its spending is assumed to evolve according to
\begin{equation}
	\mathsf{G}_t = \left(1 - \frac{1}{\mathsf{g}_t} \right)\mathsf{Y}_t,
\end{equation}
where $\mathsf{g}_t$ is an exogenous process and is assumed to follow
\begin{equation}
   \log \mathsf{g}_t = (1-\rho_g)\log \mathsf{g} + \rho_g \log \mathsf{g}_{t-1} + \varepsilon_{g,t},
\end{equation}
with $\varepsilon_{g,t}\sim\mathcal{N}(0,\sigma_g^2)$.
%We further assume that all shocks $\epsilon_{z,t}$, $\epsilon_{R,t}$, and $\epsilon_{g,t}$ are mutually independent.

A central bank sets the interest rate by an interest rate feedback rule
\begin{equation}
     \mathsf{R}_t = \mathsf{R}_t^{*1-\rho_R}\mathsf{R}_{t-1}^{\rho_R}\exp({\varepsilon_{R,t}}),
\end{equation}
where $\varepsilon_{R,t}\sim\mathcal{N}(0,\sigma_R^2)$ is a monetary policy shock and $\mathsf{R}_t^*$ is the nominal target rate
\begin{equation}
     \mathsf{R}^*_t = \mathsf{r}\pi^*\left(\frac{\pi_t}{\pi^*} \right)^{\psi_1}\left(\frac{\mathsf{Y}_t}{\mathsf{Y}_t^*} \right)^{\psi_2},
\end{equation}
where $\mathsf{r}$ is the steady-state real interest rate, $\pi^*$ is the target inflation rate, which coincides in equilibrium with the steady-state inflation rate $\pi$, and $Y_t^*$ is the level of output when there are no nominal rigidities ($\phi = 0$).

The government levies a lump-sum taxes to finance any shortfalls in government revenues. Its budget constraint is given by
\begin{equation}
   \mathsf{P}_t\mathsf{G}_t + \mathsf{M}_{t-1} + \mathsf{R}_{t-1}\mathsf{B}_{t-1} = \mathsf{T}_t + \mathsf{M}_t + \mathsf{B}_t.
\end{equation}

\noindent \emph{Summary of equilibrium conditions.} We express the model in terms of detrended variables, $\mathsf{c}_t = \mathsf{C}_t/\mathsf{A}_t$ and $\mathsf{y}_t = \mathsf{Y}_t/\mathsf{A}_t$. It can be shown that if the innovations $(\varepsilon_{z,t},\varepsilon_{R,t},\varepsilon_{g,t})$ are zero at all times, the model economy has a unique steady-state, which is given as follows
\begin{equation}
  \pi = \pi^*, \quad \mathsf{r} = \gamma/\beta, \quad \mathsf{R} = \mathsf{r}\pi^*, \quad \mathsf{c} = (1-\nu)^{1/\tau}, \quad \mathsf{y} = \mathsf{g}\mathsf{c}.
\end{equation}
Denoting $\hat{\mathsf{x}}_t = \log(\mathsf{x}_t/\mathsf{x})$ as the percentage deviation of a variable $\mathsf{x}_t$ from its steady-state $\mathsf{x}$, the model's equilibrium conditions can be summarized as follows
\begin{align}
     &\beta \mathbb{E}_t\Big[\exp(-\tau(\hat{\mathsf{c}}_{t+1} - \hat{\mathsf{c}}_{t}) + \hat{\mathsf{R}}_t - \hat{\mathsf{z}}_{t+1} -\hat{\pi}_{t+1})\Big]= 1,  \\
     &\exp(\hat{\mathsf{c}}_{t} -\hat{\mathsf{y}}_{t}) = \exp(-\hat{\mathsf{g}}_t) - \frac{\phi\pi^2\mathsf{g}}{2}(\exp(\hat{\pi}_{t}) - 1)^2,\\
     % \end{align}
    % \begin{align}
     &\frac{1-\nu}{\nu\phi\pi^2} (\exp(\tau \hat{\mathsf{c}}_{t}) - 1) = (\exp(\hat{\pi}_{t}) - 1)\Big[(1-\frac{1}{2\nu})\exp(\hat{\pi}_{t}) + \frac{1}{2\nu}\Big]\notag\\
     &-\beta \mathbb{E}_t\Big[(\exp(\hat{\pi}_{t+1}) - 1)\exp(-\tau(\hat{\mathsf{c}}_{t+1} -\hat{\mathsf{c}}_{t}) + \hat{\mathsf{y}}_{t+1} - \hat{\mathsf{y}}_{t} +\hat{\pi}_{t+1})\Big],\\
          &\hat{\mathsf{R}}_t = \rho_R \hat{\mathsf{R}}_{t-1} + (1-\rho_R)\psi_1\hat{\pi}_t + (1-\rho_R)\psi_2(\hat{\mathsf{y}}_{t} -\hat{\mathsf{g}}_{t}) + \varepsilon_{R,t},\\
     &\hat{\mathsf{g}}_t = \rho_g\hat{\mathsf{g}}_{t-1} + \varepsilon_{g,t},\\
     &\hat{\mathsf{z}}_t = \rho_z\hat{\mathsf{z}}_{t-1} + \varepsilon_{z,t}.
     \end{align}

\noindent \emph{Model solution.} The above equilibrium conditions form a non-linear rational expectations system in the state variable
     \begin{equation}
         s_t = (\hat{\mathsf{y}}_t, \hat{\mathsf{c}}_t, \hat{\pi}_t, \hat{\mathsf{R}}_t, \varepsilon_{R,t}, \hat{\mathsf{g}}_t,\hat{\mathsf{z}}_t),
     \end{equation}
in which the first four are endogenous state variables and the last three are exogenous state variables. This system has to be solved numerically before estimating the model using data; the resulting solution has the form in Equation \eqref{ssm:state_obs}. In empirical studies, linear approximation methods \citep[e.g.,][]{sims2002solving} are very popular as they result in a linear state-space representation of the model that can be easily estimated by evaluating the likelihood function using the Kalman filter \citep[see,][]{an2007bayesian, herbst2016}.

However, as shown in \citet*{fernandez2006ecta}, the second-order approximation errors in the solutions of dynamic economic models have first-order effects on the likelihood estimation. Furthermore, errors in the likelihood estimation are compounded with the size of the sample. Therefore, it is desirable to solve the model using more accurate methods. For the purpose of illustrating our new econometric method, after taking both accuracy and computational cost in account, we choose a second-order perturbation method with pruning \citep*{schmitt2004, schmitt2007optimal, fernandez2018res} to solve the above non-linear rational expectations system.\bigskip

\noindent \emph{Measurement equations.} Finally, we complete the model by relating the state variables $s_t$ to a set of observables as in Equation \eqref{ssm:state_obs}. Following \citet{an2007bayesian} and \citet{herbst2016}, we assume that the following observations are available: quarter-to-quarter per capita GDP growth rates (YGR), annualized quarter-to-quarter inflation rates (INF), and annualized nominal interest rates (INT), which are in percentages and relate to the state variables as follows
\begin{eqnarray}
        \textrm{YGR}_t &=& \gamma^{(Q)} + 100 (\hat{\mathsf{y}}_t - \hat{\mathsf{y}}_{t-1}+\hat{\mathsf{z}}_t) + u_{1,t}, \label{obsy}\\
        \textrm{INF}_t &=& \pi^{(A)} + 400\hat{\pi}_t + u_{2,t}, \label{obspi}\\
        \textrm{INT}_t &=& \pi^{(A)} + \mathsf{r}^{(A)} + 4\gamma^{(Q)} + 400\hat{R}_t + u_{3,t},\label{obsr}
\end{eqnarray}
where $u_t = (u_{1,t}, u_{2,t}, u_{3,t})\sim\mathcal{N}(0_3,\Sigma_u)$, which captures the measurement errors of the observables,
has a zero mean vector $0_3$ and diagonal covariance matrix $\Sigma_u = \textrm{diag}(\sigma_{e,y}^2, \sigma_{e,\pi}^2, \sigma_{e,R}^2)$.
The $d_{\theta}=18$ unknown parameters $\theta$ to be inferred include the structural parameters and the standard deviations of the measurement errors
\begin{equation}
    \theta = (\tau, \nu, \kappa, 1/\mathsf{g}, \psi_1, \psi_2, \mathsf{r}^{(A)}, \pi^{(A)}, \gamma^{(Q)}, \rho_R, \rho_g, \rho_z, \sigma_R, \sigma_g, \sigma_z, \sigma_{e,y}, \sigma_{e,\pi}, \sigma_{e,R}),
\end{equation}
where $\kappa=\tau(1-\nu)/\nu\pi^2\phi$, and $(\mathsf{r}^{(A)},\pi^{(A)},\gamma^{(Q)})$ are related to the model steady-states via
\begin{equation}
   \beta = \frac{1}{1+\mathsf{r}^{(A)}/400}, \quad \pi = 1 + \frac{\pi^{(A)}}{400}, \quad \gamma = 1 +  \frac{\gamma^{(Q)}}{100}.
\end{equation}

\subsection{Implementation details}

\noindent \emph{State-space model.} Solving the model using the second-order perturbation method as discussed above introduces non-linearities in the resulting state-space model. The latent state variables $s_t = (x_t,z_t)\in \mathbb{S}=\mathbb{R}^{d_x}\times\mathbb{R}^{d_z}$ contain $d_x=4$ endogenous variables $x_t$ and $d_z=3$ exogenous variables $z_t$. The time evolution of the endogenous variables $(x_t)_{t=0}^T$ is given by a deterministic mapping $x_t = X_{\theta}(x_{t-1},z_t)$ for $t=1,\ldots,T$ with initialization at $x_0=X_{\theta}(0_{d_x},z_0)$, where $X_{\theta}:\mathbb{S}\rightarrow\mathbb{R}^{d_x}$ has the form
\begin{equation}\label{eqn:dsge_xtransition}
	X_{\theta}(s)=c(\theta)+L(\theta)s+Q_{\theta}(s),\quad s\in\mathbb{S}.
\end{equation}
In the above, the constant is $c(\theta)\in\mathbb{R}^{d_x}$, the linear term is $L(\theta)\in\mathbb{R}^{d_x\times d}$, and
the quadratic term $Q_{\theta}:\mathbb{S}\rightarrow\mathbb{R}^{d_x}$ is defined as $Q_{\theta}(s)_i=s^\top Q_i(\theta)s$
where $Q_i(\theta)\in\mathbb{R}^{d\times d}$ for $i=1,\ldots,d_x$. The exogenous variables $(z_t)_{t=0}^T$ are modelled as a vector autoregressive process
\begin{equation}\label{eqn:dsge_ztransition}
	z_t = \rho(\theta) z_{t-1}+ \Sigma(\theta)\varepsilon_t,\quad\varepsilon_t\sim\mathcal{N}(0_{d_z},I_{d_z}),
	\quad t=1,\ldots,T,
\end{equation}
with initialization $z_0 = \Sigma(\theta)\varepsilon_0, \varepsilon_0\sim\mathcal{N}(0_{d_z},I_{d_z})$,
where the matrices are $\rho(\theta),\Sigma(\theta)\in\mathbb{R}^{d_z\times d_z}$ and $I_{d_z}$ denotes the identity matrix of size $d_z$. We can then write the time evolution of the latent states $(s_t)_{t=0}^T$ as
\begin{equation}\label{eqn:dsge_state_transition}
	s_t = A(\theta)s_{t-1} + B(\theta)\varepsilon_t + c_{\theta}(s_{t-1},\varepsilon_t),
	\quad \varepsilon_t\sim\mathcal{N}(0_{d_z},I_{d_z}),
\end{equation}
for $t=1,\ldots,T$, with the initial state
\begin{equation}\label{eqn:dsge_initial_state}
	s_0 = B(\theta)\varepsilon_0 + c_{\theta}((0_{d_x},\Sigma(\theta)\varepsilon_0),\varepsilon_0),
	\quad\varepsilon_0\sim\mathcal{N}(0_{d_z},I_{d_z}).
\end{equation}
Expressions for the model matrices $A(\theta)\in\mathbb{R}^{d\times d}$, $B(\theta)\in\mathbb{R}^{d\times d_z}$
and the function $c_{\theta}:\mathbb{S}\times\mathbb{R}^{d_z}\rightarrow\mathbb{S}$ in terms of the above quantities
are given in the Appendix \ref{appendix:dsge_ssm}. In the framework of Equation \eqref{eqn:latent_process}, this corresponds to an initial distribution and a Markov transition kernel that are partially degenerate
\begin{equation}
\begin{aligned}\label{eqn:dsge_transitions}
	\mu_{\theta}(ds_0)&=\delta_{\Phi_{\theta}^{(0)}(\varepsilon_0)}(ds_0)\mathcal{N}(\varepsilon_0;0_{d_z},I_{d_z})d\varepsilon_0,\\
	f_{\theta}(ds_t | s_{t-1}) &= \delta_{\Phi_{\theta}(s_{t-1},\varepsilon_{t})}(ds_t)\mathcal{N}(\varepsilon_t;0_{d_z},I_{d_z})d\varepsilon_t,
\end{aligned}
\end{equation}
for $t=1,\ldots,T$, where $s_0=\Phi_{\theta}^{(0)}(\varepsilon_0)$ and $s_t=\Phi_{\theta}(s_{t-1},\varepsilon_{t})$ denote the mappings in Equations \eqref{eqn:dsge_state_transition} and \eqref{eqn:dsge_initial_state}. The log-linearized model can be easily obtained by simply suppressing the $c_{\theta}$ terms in Equations \eqref{eqn:dsge_state_transition} and \eqref{eqn:dsge_initial_state}.

There are three observables $y_t\in\mathbb{Y}=\mathbb{R}^{d_y}$ with $d_y=3$. For each quarter $t=1,\ldots,T$, the relation between observables and latent states is described by the model
\begin{align}
	g_{\theta}(y_t | s_t) = \mathcal{N}(y_t; d(\theta) + E(\theta)s_t, F(\theta)),
\end{align}
for some model vector $d(\theta)\in\mathbb{R}^{d_y}$, model matrix $E(\theta)\in\mathbb{R}^{d_y\times d}$
and covariance matrix $F(\theta)\in\mathbb{R}^{d_y\times d_y}$; see Equations \eqref{obsy}, \eqref{obspi}, and \eqref{obsr}.\bigskip
%{\color{red}{restriction on parameters whose linear dynamics has stationary distribution, discuss projection on compact set?}}
%\citep{fernandez2007estimating,herbst2015bayesian}.

\noindent \emph{Annealed controlled SMC.}  The uncontrolled SMC within AC-SMC is taken as the BPF; see Section \ref{sec:SMC} for the corresponding choice of proposals $(q_t)_{t=0}^T$ and weight functions $(w_t)_{t=0}^T$. Owing to the partially degenerate nature of the model's initial distribution and Markov transitions in Equation \eqref{eqn:dsge_transitions}, here we elaborate how to specify policies. The treatment we propose is necessary for models with both endogenous and exogenous state variables and to the best of our knowledge, has not been considered in earlier works.

Instead of adopting a parameterization purely in terms of the latent variables $(s_t)_{t=0}^T$, we will also introduce dependence on the noise variables $(\varepsilon_t)_{t=0}^T$ in the function classes
\begin{equation}
\begin{aligned}\label{eqn:dsge_quadratic_functionclass}
	\mathbb{F}_0 &= \left\lbrace \psi_t(\varepsilon_0)=\exp(-Q_0(\varepsilon_0;\beta_0)) : |\beta_0|^2\leq \xi\right\rbrace,\\
	\mathbb{F}_t &= \left\lbrace \psi_t(s_{t-1},\varepsilon_t)=\exp(-Q(s_{t-1},\varepsilon_t;\beta_t)) :
	\beta_t=\Lambda_{\theta}(\tilde{\beta}_t), |\tilde{\beta}_t|^2\leq \xi\right\rbrace,
\end{aligned}
\end{equation}
for $t=1,\ldots,T$, where $Q_0(z;\beta)=z^\top A z + z^\top b + c$ is a quadratic function with coefficients
$\beta=(A,b,c)\in\mathbb{R}^{d_z\times d_z}_{\mathrm{sym}}\times\mathbb{R}^{d_z}\times\mathbb{R}$
and $Q(s_{t-1},\varepsilon_t;\beta_t)$ has the form of
\begin{equation}\label{eqn:quadratic_function}
	Q(z,z';\beta) = {z}'^\top A z' + {z}'^\top b + {z}'^\top C z + z^{\top} D z + z^{\top}e + f,
\end{equation}
with coefficients $\beta=(A,b,C,D,e,f)\in\mathbb{R}^{d_z\times d_z}_{\mathrm{sym}}\times\mathbb{R}^{d_z}\times \mathbb{R}^{d_z\times d}\times \mathbb{R}_\mathrm{sym}^{d\times d}\times\mathbb{R}^d\times\mathbb{R}$.
The notation $\mathbb{R}_\mathrm{sym}^{d\times d}$ refers to the space of real symmetric matrices of size $d\times d$,
and $|\beta|$ denotes the Euclidean norm of $\beta$ as a vector in Euclidean space.
Under this specification, we can initialize the policy in AC-SMC by having all coefficients equal to zero.
Functions in Equation \eqref{eqn:dsge_quadratic_functionclass} can be fitted using ridge regression in the logarithmic scale, and $\xi\in(0,\infty)$ controls the amount of shrinkage. Since policy approximations ultimately construct SMC proposal distributions, we take the view that careful selection of shrinkage within ADP is unnecessary. In our numerical implementation, we will fix $\xi$ as a suitably large value to safeguard against ill-conditioned Gram matrices.

The choice of function classes in Equation \eqref{eqn:dsge_quadratic_functionclass} is well-specified when log-linearization approximations are used to solve for the equilibrium conditions; see the Appendix \ref{appendix:dsge_lqg} for the optimal policy of the linearized model. When we employ second-order approximation methods that induce non-linearities in the model, this choice can provide good policy approximations despite being misspecified. However, our numerical findings reveal that these function classes are too flexible if the coefficients are not appropriately constrained. In practice, overfitting at each step of the ADP algorithm (Algorithm \ref{alg:ADP}) could result in numerical instabilities after several steps of iterating overfitted functions. To prevent overfitting, instead of performing regression using the variables $(s_{t-1},\varepsilon_t)$, we reduce the dimension of the feature space by working with the linearized state $\tilde{s}_t = A(\theta)s_{t-1} + B(\theta)\varepsilon_t$, and fitting a quadratic function $Q_0(\tilde{s}_t;\tilde{\beta}_t)$ with coefficients $\tilde{\beta}_t=(\tilde{A}_t,\tilde{b}_t,\tilde{c}_t)\in\mathbb{R}_\mathrm{sym}^{d\times d}\times\mathbb{R}^d\times\mathbb{R}$. The desired coefficients $\beta_t=(A_t,b_t,C_t,D_t,e_t,f_t)$ are then obtained by equating coefficients in the equality $Q_0(\tilde{s}_t;\tilde{\beta}_t)=Q(s_{t-1},\varepsilon_t;\beta_t)$, which defines the mapping $\Lambda_{\theta}$ (see, Appendix \ref{appendix:dsge_dimreduction} for details). We note that regularization is necessary in this setting, as the linearized states (and hence its features) are supported on a lower-dimensional subspace.

After fitting a policy $\psi=(\psi_t)_{t=0}^T$, the new proposal transitions $(q_t^{\psi})_{t=1}^T$ involve sampling the noise $\varepsilon_t$ from a distribution that depends on the previous state $s_{t-1}$, and applying the map $s_t=\Phi_{\theta}(s_{t-1},\varepsilon_{t})$ to obtain the next state (similarly for the new initial distribution $q_0^{\psi}$). As the map $\Phi_{\theta}(s_{t-1},\varepsilon_{t})$ may be the result of a complex numerical routine, the state transition density is generally intractable; hence including the noise variables in our policy parameterization is key to facilitate sampling from the new proposal transitions.
We refer readers to the Appendix \ref{appendix:dsge_csmc} for the precise form of the proposal transitions and analytical expressions to evaluate the expectations of the policy appearing in the weight functions of Equation \eqref{eqn:weights_psi} and Algorithm \ref{alg:ADP}.
Lastly, as the form of the functions in Equation \eqref{eqn:dsge_quadratic_functionclass} are closed under multiplication,
policy refinements can be performed by updating the coefficients (see, Appendix \ref{appendix:dsge_constraints}).

\subsection{Simulations}

We perform Monte Carlo simulations to examine the likelihood estimator of AC-SMC.
We first consider the log-linearized model that is linear and Gaussian. Therefore, the likelihood function can be evaluated exactly using a Kalman filter.
We generate a sequence of quarterly observations of YGR, INF and INT from this model for 500 time periods.
The true model parameter values are the same as those used in \citet[Table 3]{an2007bayesian}.
For the purpose of comparison, we also include simulation results for the BPF with various numbers of particles.
It is well-understood that the performance of the BPF can be very poor when the standard deviations of the measurement errors are small.
To investigate how the efficiency of AC-SMC behaves with respect to the magnitude of the standard deviations of the measurement errors,
we choose these standard deviations to be 5\%, 10\%, 15\%, or 20\% of the standard deviations of the simulated data and add measurement noise to the data accordingly.
For each measurement error setting, we implement 100 independent repetitions of BPF and AC-SMC on a fixed simulated data sample
and compute the sample means and variances of their log-likelihood estimates.

Table \ref{table_likelihood_linear} summarizes the performance of these particle filters. For AC-SMC, we set the number of particles as $N=1,024$, and for BPF, we choose the number of particles $N$ equal to 4,096, 8,192, or 16,384.
It is striking that for all measurement error settings, the sample mean of the AC-SMC log-likelihood estimates matches the true log-likelihood value computed using a Kalman filter, and the sample variance is very small, particularly in the case of larger measurement errors.
In contrast, these numerical results show that the BPF log-likelihood estimator is very inefficient, and especially so in the case of small measurement errors.
For example, in the setting where the standard deviations of the measurement errors are 5\% of their sample analogue,
the true log-likelihood value is -2,091.4;
the sample mean and variance of log-likelihood estimates are -2,091.4 and $2.3\times 10^{-4}$ respectively for AC-SMC,
and -2,094.6 and 416.9 respectively for BPF with $N=16,384$ number of particles.
In terms of computational cost, AC-SMC with $N=1,024$ number of particles is comparable to that of BPF with $N=8,192$ number of particles.
\begin{table}[!t]
\begin{center}
\caption{\textbf{{\protect\footnotesize{Log-likelihood Estimates: Linear Model}}}}\label{table_likelihood_linear}
{\onehalfspacing\footnotesize
\resizebox{\textwidth}{!}{\begin{tabular*}{1.2\textwidth}{@{\extracolsep{\fill}}cccccccccccc}
\toprule
                      & \multicolumn{5}{c}{ME: 5\%}&& \multicolumn{5}{c}{ME: 10\%}   \\
\cline{2-6} \cline{8-12}
                      &     & \multicolumn{3}{c}{BPF} &  &&  & \multicolumn{3}{c}{BPF} & \\
                        \cline{3-5}\cline{9-11}
                    &    KF          & $2^{12}$ & $2^{13}$ & $2^{14}$ & AC-SMC &&   KF       & $2^{12}$ & $2^{13}$ & $2^{14}$ & AC-SMC\\
 \hline
 Mean           & -2091.4 &	-2441.1 &	-2287.4 &	-2194.6 &	-2091.4 && -2162.2	&-2217.7	&-2193.2	&-2179.6	&-2162.2\\
Var               &	            &   3061.2  &	1491.2  &	416.9   &	$2.3\times 10^{-4}$  &&		&150.9	&132.9	&45.9	&$1.0\times 10^{-5}$\\
Time             &              &	1.26	     &2.58	     &4.93	    &2.59	      &&               &   1.3	& 2.65	&5.04	&2.59	\\
\hline
                      & \multicolumn{5}{c}{ME: 15\%}&& \multicolumn{5}{c}{ME: 20\%}   \\
\cline{2-6} \cline{8-12}
                      &     & \multicolumn{3}{c}{BPF} &  &&  & \multicolumn{3}{c}{BPF} & \\
                        \cline{3-5}\cline{9-11}
                    &    KF          & $2^{12}$ & $2^{13}$ & $2^{14}$ & AC-SMC &&   KF       & $2^{12}$ & $2^{13}$ & $2^{14}$ & AC-SMC\\
 \hline
Mean	    & -2265.7     &-2297.3	&-2284.9	&-2276.9	&-2265.7	&&-2311.2	     &-2323.0    &-2317.9	&-2314.4	&-2311.2\\
Var    	    &                  &80.7	&51.3	&29.8	&$3.4\times 10^{-6}$	&&	              & 30.8	       &17.7         &	7.3	&$9.9\times 10^{-7}$\\
Time		    &                  &1.32	&2.66	&5.05	&2.54        &&		      &1.35	       &2.97	         &5.3	        &2.7\\
\bottomrule
\end{tabular*}}}
\end{center}
{\footnotesize \emph{Note:}
We generate a sequence of quarterly observations of YGR, INF and INT from the linear model for 500 time periods.
The true model parameter values are the same as those used in \citet[Table 3]{an2007bayesian}.
For each measurement error (ME) setting, we implement 100 independent repetitions of BPF and AC-SMC, based on the same sequence of simulated data,
and compute the sample means and variances of their log-likelihood estimates.
The true log-likelihood value is computed using a Kalman filter (KF).
For AC-SMC, we set the number of particles as $N=1,024$, and for BPF, we choose the number of particles $N$ equal to 4,096, 8,192, or 16,384.
The runtime is measured in seconds.}
\end{table}

We now consider the non-linear model that is obtained by solving the DSGE model with the second-order perturbation method.
Note that the likelihood function of this non-linear model is intractable.
As before, we simulate data from the model and consider four measurement error settings.
Table \ref{table_likelihood_nonlinear} reports the sample means and variances of log-likelihood estimates from 100 independent repetitions of BPF and AC-SMC.
In this simulation study, we set the number of particles in AC-SMC as $N=1,024$, and consider the number of particles $N$ equal to 16,384, 32,768, or 65,536 in BPF.

The performance of AC-SMC is also significantly better than that of BPF for this non-linear model.
For example, in the challenging regime of small measurement errors,
the sample variance of the log-likelihood estimates of AC-SMC is as small as 0.065, compared to 182.7 for BPF with $N=65,536$ particles.
Moreover, the gains in terms of variance reduction is always of several orders of magnitude across the four measurement error settings.
To take computational cost into account, we observe that the cost of AC-SMC with $N=1,024$ particles is comparable to that of BPF with
a number of particles between 32,768 and 65,536.
Lastly, we note that the likelihood estimate obtained by approximating the non-linear model with log-linearization
and running the Kalman filter is very biased compared to the estimates from BPF and AC-SMC.
\begin{table}[!t]
\begin{center}
\caption{\textbf{{\protect\footnotesize{Log-Likelihood Estimates: Non-linear Model}}}}\label{table_likelihood_nonlinear}
{\onehalfspacing\footnotesize
\resizebox{\textwidth}{!}{\begin{tabular*}{1.2\textwidth}{@{\extracolsep{\fill}}cccccccccccc}
\toprule
                      & \multicolumn{5}{c}{ME: 5\%}&& \multicolumn{5}{c}{ME: 10\%}   \\
\cline{2-6} \cline{8-12}
                      &     & \multicolumn{3}{c}{BPF} &  &&  & \multicolumn{3}{c}{BPF} & \\
                        \cline{3-5}\cline{9-11}
                    &    KF          & $2^{14}$ & $2^{15}$ & $2^{16}$ & AC-SMC &&   KF       & $2^{14}$ & $2^{15}$ & $2^{16}$ & AC-SMC\\
 \hline
 Mean  & -2204.4 & -2192.7 & -2147.3 & -2120.7 & -2089.4 && -2212.4 & -2164.5 & -2159.5 & -2156.2 & -2152.2 \\
 Var     &       & 570.1 & 325.2 & 182.7 & $6.5\times 10^{-2}$ &&       & 34.4  & 14.6  & 10.8  & $3.2\times 10^{-4}$ \\
 Time  &       & 5.12  & 13.17 & 30.68 & 20.58 &&       & 5.03  & 13.23 & 30.97 & 20.42 \\
\hline
                      & \multicolumn{5}{c}{ME: 15\%}&& \multicolumn{5}{c}{ME: 20\%}   \\
\cline{2-6} \cline{8-12}
                      &     & \multicolumn{3}{c}{BPF} &  &&  & \multicolumn{3}{c}{BPF} & \\
                        \cline{3-5}\cline{9-11}
                    &    KF          & $2^{14}$ & $2^{15}$ & $2^{16}$ & AC-SMC &&   KF       & $2^{14}$ & $2^{15}$ & $2^{16}$ & AC-SMC\\
 \hline
 Mean  & -2257.3 & -2222.1 & -2219.1 & -2218.2 & -2216.9 && -2359.7 & -2347.7 & -2346.6 & -2345.5 & -2344.3 \\
 Var      &       & 13.8  & 6.7   & 2.7   & $3.9\times 10^{-4}$ &&       & 9.2   & 4.4   & 1.9   & $5.0\times 10^{-4}$ \\
 Time   &       & 5.04  & 13.21 & 31.16 & 20.78 &&       & 5.06  & 13.22 & 31.25 & 20.59 \\
\bottomrule
\end{tabular*}}}
\end{center}
{\footnotesize \emph{Note:}
We generate a sequence of quarterly observations of YGR, INF and INT from the non-linear model for 500 time periods.
The true model parameter values are the same as those used in \citet[Table 3]{an2007bayesian}.
For each measurement error (ME) setting, we implement 100 independent repetitions of BPF and AC-SMC, based on the same sequence of simulated data,
and compute the sample means and variances of their log-likelihood estimates.
We set the number of particles in AC-SMC as $N=1,024$, and consider the number of particles $N$ equal to 16,384, 32,768, or 65,536 in BPF. The runtime is measured in seconds.}
\end{table}

\subsection{Real data application}\label{sec:dsge_realdata}

We now apply our adaptive SMC$^2$ algorithm with AC-SMC nested to estimate the above new Keynesian DSGE model using real data. While the log-linearized approximation of this model has been extensively investigated \citep[see, e.g.,][]{an2007bayesian, herbst2016}, the corresponding non-linear model has not been fully studied yet in the literature. We estimate the resulting non-linear model from applying the second-order perturbation method with pruning based on three observables as discussed in Equations \eqref{obsy}, \eqref{obspi}, and \eqref{obsr}: quarterly per capita GDP growth rate (YGR), quarterly inflation (INF), and the annualized federal funds rate as a proxy of interest rate (INT).
To avoid the issue of having a zero lower bound on the interest rate, we focus on the pre-crisis sample, ranging from 1983:Q1 to 2007:Q4 with a total of 100 observations. The time series of these three observables are displayed in Figure \ref{fig:data}.
\begin{figure}[!t]
\begin{center}
\includegraphics[width=1\textwidth, height=12cm]{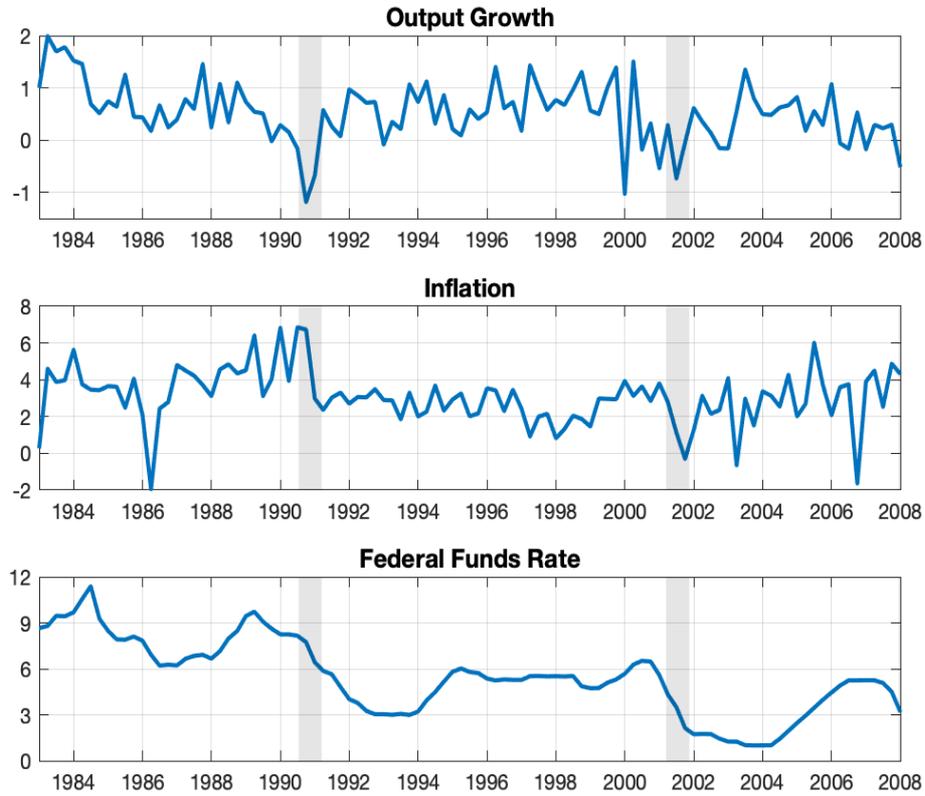}
\vspace{-1cm}
\caption{\textbf{\small {Time Series Data}}}\label{fig:data}
\end{center}
{\footnotesize\emph{Note:} Time series data include quarter-to-quarter per capita GDP growth rates (YGR), annualized quarter-to-quarter inflation rates (INF), and annualized nominal interest rates (INT), all of which are in percentages. To avoid the issue of having a zero lower bound on the interest rate, we focus on the pre-crisis sample, ranging from 1983:Q1 to 2007:Q4 with a total of 100 observations. The regions shaded in grey correspond to NBER recession periods.}
\end{figure}

Our adaptive SMC$^2$ algorithm requires initializing particles from the prior distribution of the model parameters.
Given the complex behaviour of the likelihood function implied by the DSGE model, choosing a conjugate prior is not feasible here.
We adopt a prior distribution that is component-wise independent, with marginal distributions that are normal distributions for real-valued parameters,
truncated Normal distributions for positive parameters, and uniform distributions for bounded parameters.
Under these choices, simulating from the prior distribution is straightforward.
The hyper-parameters of the prior distribution are selected by following the existing literature in macroeconomics.
The left column of Table \ref{table_estimates} details the exact distributional form, the support, and the hyper-parameters of the prior distribution for each model parameter.

We estimate the model in two settings: when the standard deviations of the measurement errors are fixed at 20\% of the corresponding sample standard deviations, as is commonly considered in the literature, and when they are treated as free parameters to be inferred using the data.
The algorithmic configuration of our adaptive SMC$^2$ involves $P=1,024$ parameter particles, $N=1,024$ state particles within AC-SMC, and an ESS threshold of $\kappa_{\mathrm{ESS}}=0.5$ to adaptively determine the inverse temperature schedule.
To ensure adequate sample diversity after resampling, we apply PMMH moves until the cumulative acceptance rate reaches two.
%In implementation, we set the number of state particles equal to 1,024 in AC-SMC and the number of parameter particles at 1,024.
%We select the tempering schedule so that the parameter ESS falls slightly below $N/2$ (at each distribution-bridging step),
Figure \ref{fig:acceptrate} displays the acceptance rate of the final PMMH move at each annealing iteration.
We observe acceptance rates ranging from around 0.2 to 0.4 in both settings, revealing adequate rejuvenation of the sample diversity
\citep{doucet2015efficient,sherlock2015efficiency}.
%An important quantity in our proposed methodology is the acceptance rate, which measures the efficiency of each rejuvenation step.
%Figure \ref{fig:acceptrate} display the acceptance rate of the last PMMH iteration at each bridging distribution. In both cases, we notice that the acceptance rate fluctuates around 25\%, ranging from about 0.20 to 0.40, a typical range suggested by the Bayesian literature.
\begin{figure}[t]
\begin{center}
\includegraphics[width=1\textwidth, height=6.5cm]{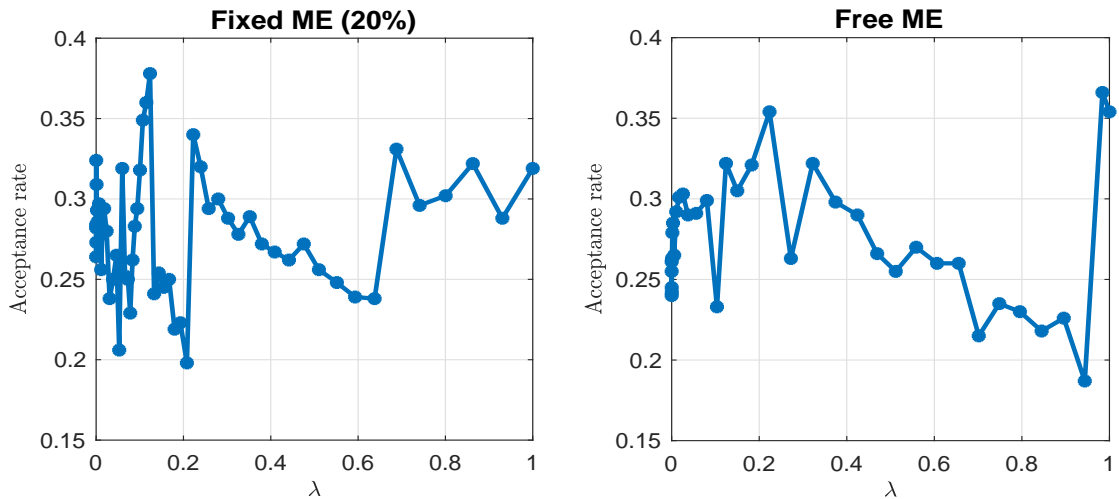}
\caption{\textbf{\protect\small {Acceptance Rates}}}\label{fig:acceptrate}
\end{center}
\end{figure}

The middle column of Table \ref{table_estimates} reports the parameter estimates of the model when standard deviations of the measurement errors are fixed. The risk-aversion coefficient has a posterior mean of 2.07 and a posterior standard deviation of 0.42.
As the posterior mean of $\kappa$ is 0.76, which is relatively large, this suggests a low degree of price rigidity and a small effect of monetary policy shocks on output.
The annualized steady-state growth rate of the economy is 3.2\%, the annualized steady-state inflation rate is 4.8\%, and the annualized steady-state nominal interest rate is 8.4\%. The steady-state ratio of $\mathsf{c}/\mathsf{y}$, which is equal to $1/\mathsf{g}$, is around 0.28.
The two exogenous processes, given by the government spending shock and the productivity shock, are close to unit root as the estimated persistence parameters are 0.98 and 0.99, respectively. This suggests that innovations to these processes have a long-lasting effect.
\begin{table}[!t]
\begin{center}
\caption{\textbf{{\protect\footnotesize{Prior Distributions and Posterior Estimates of Parameters}}}}\label{table_estimates}
{\onehalfspacing\footnotesize
\begin{tabular*}{\textwidth}{@{\extracolsep{\fill}}cllcccc}
\toprule
                        &  \multicolumn{2}{c}{Priors} &  \multicolumn{2}{c}{Fixed ME (20\%)}  & \multicolumn{2}{c}{Free ME}   \\
                              \cline{2-3}                                 \cline{4-5}                       \cline{6-7}

Parameter             & \multicolumn{1}{c}{Support}& \multicolumn{1}{c}{Distribution} &  Mean  &  Std     & Mean  & Std     \\
\hline
$\tau$                 & $(0,\infty)$ & $\mathcal{TN}(2.00, 0.50)$ & 2.07 & 0.42  &    1.73 & 0.43 \\
$\nu$                  & $(0, 1)$               & $\mathcal{U}(0, 1)$          & 0.43 & 0.27  &    0.68 & 0.20  \\
$\kappa$            & $(0,\infty)$ & $\mathcal{TN}(0.20, 0.20)$ & 0.76 & 0.15 &     0.06 & 0.04 \\
$1/\mathsf{g}$                  & $(0, 1)$ & $\mathcal{U}(0, 1)$& 0.28 & 0.27  &    0.53 & 0.21  \\
$\psi_1$              & $(0,\infty)$ & $\mathcal{TN}(1.50, 0.25)$& 1.73 & 0.18  &    1.80 & 0.23  \\
$\psi_2$              & $(0,\infty)$ & $\mathcal{TN}(0.50, 0.25)$& 0.44 & 0.22  &     0.67 & 0.23  \\
$\mathsf{r}^{(A)}$             & $(0,\infty)$ & $\mathcal{TN}(0.80, 0.50)$& 0.41 & 0.22   &    0.51 & 0.21  \\
$\pi^{(A)}$           & $(0,\infty)$ & $\mathcal{TN}(4.00, 2.00)$& 4.81 & 0.61  &     3.65 & 0.30   \\
$\gamma^{(Q)}$ & $\mathbb{R}$ & $\mathcal{N}(0.40, 0.20)$& 0.79 & 0.10   &    0.68 & 0.09   \\
$\rho_R$            & $(0, 1)$ & $\mathcal{U}(0, 1)$& 0.83 & 0.03   &     0.64 & 0.05  \\
$\rho_g$             & $(0, 1)$ & $\mathcal{U}(0, 1)$& 0.98 & 0.02   &    0.91 & 0.16   \\
$\rho_z$             & $(0, 1)$ & $\mathcal{U}(0, 1)$& 0.99 & 0.01   &    0.97 & 0.01  \\
100$\sigma_R$    & $(0,\infty)$ & $\mathcal{TN}(0.30, 4.00)$& 0.18 & 0.02&    0.03 & 0.01   \\
100$\sigma_{g}$ & $(0,\infty)$ & $\mathcal{TN}(0.40, 4.00)$& 0.62 & 0.05&    0.02 & 0.01   \\
100$\sigma_{z}$ & $(0,\infty)$ & $\mathcal{TN}(0.40, 4.00)$& 0.08 & 0.02&    0.01 & 0.00  \\
\hline
$\sigma_{e,y}$   & $(0,\infty)$ & $\mathcal{TN}(0.12, 1.00)$& 0.12 & ------&    0.50 & 0.05   \\
$\sigma_{e,\pi}$  & $(0,\infty)$ & $\mathcal{TN}(0.30, 1.00)$& 0.30 & ------&    1.37 & 0.10   \\
$\sigma_{e, R}$   & $(0,\infty)$ & $\mathcal{TN}(0.49, 1.00)$& 0.49 & ------&    0.13 & 0.03  \\
\hline
% LLH                    &            &      &  \multicolumn{2}{c}{[-403.0, -392.1]} & \multicolumn{2}{c}{[-357.8, -345.4]} \\
 Model Evidence                 &            &      &  \multicolumn{2}{c}{-442.8} & \multicolumn{2}{c}{-407.6} \\
\bottomrule
\end{tabular*}}
\end{center}
{\footnotesize \emph{Note:} This table details the prior distributions and posterior estimates of model parameters.
The left column provides the exact distributional form, the support, and the hyper-parameters of the prior distribution for each model parameter. $\mathcal{N}$ stands for the normal distribution, $\mathcal{TN}$ the truncated normal distribution, and $\mathcal{U}$ the uniform distribution.
The model is estimated using our adaptive SMC$^2$ algorithm with $P=1,024$ parameter particles and $N=1,024$ state particles within AC-SMC.
The middle and right columns report the posterior means and standard deviations of model parameters when the standard deviations of the measurement errors fixed at 20\% of the sample standard deviations and treated as free parameters, respectively.}
\end{table}

We now examine how the parameter estimates change when the standard deviations of the measurement errors are treated as free parameters.
From the right column of Table \ref{table_estimates}, the most striking change is that the posterior mean of $\kappa$ drops from 0.76 to 0.06,
which implies a much larger degree of price rigidity and a stronger effect of monetary policy shocks on output.
The annualized steady-state growth rate of the economy is 2.7\%, the annualized steady-state inflation rate is 3.7\%, the annualized steady-state nominal interest rate is 6.7\%, and the steady-state ratio of $\mathsf{c}/\mathsf{y}$ is around 0.53.
The government spending shock and the productivity shock are now slightly less persistent, with much smaller variations.

The estimated standard deviations of the measurement errors in the right column are quite different from the fixed values given in the middle column,
and suggest that the model fits the interest rate better than the output growth rate and the inflation rate.
Comparing the estimated model evidence in the two settings, reported in the last row of Table \ref{table_estimates},
suggests that the common practice of fixing the standard deviations of the measurement errors at 20\% of their sample analogue does deteriorate the model fit.
\begin{figure}[t]
\begin{center}
\includegraphics[width=1\textwidth, height=11cm]{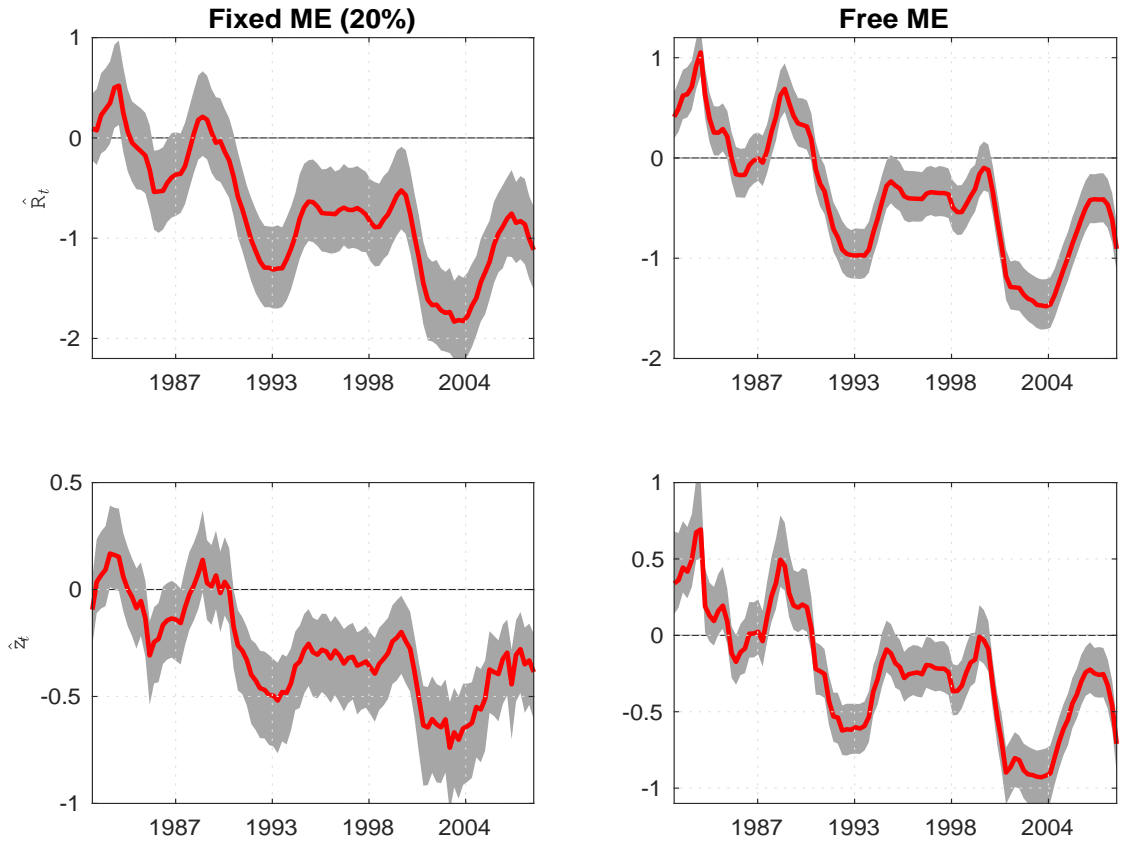}
\caption{\textbf{\protect\small { Selected Smoothed Latent States}}}\label{fig:latentstate}
\end{center}
{\footnotesize\emph{Note:} This figure illustrates (5, 50, 95)\%-quantiles of one smoothed endogenous state, the nominal interest rate ($\hat{\mathsf{R}}_{t}$), and of one smoothed exogenous state, the productivity shock ($\hat{\mathsf{z}}_t$).
The left and right panels correspond to fixing the standard deviations of the measurement errors at 20\% of their sample analogue and treating them as free parameters, respectively.}
\end{figure}

Using the output of our adaptive SMC$^2$ algorithm, we can examine the smoothing distribution of the latent states,
such as the exogenous states which are usually used in policy analysis.
Figure \ref{fig:latentstate} illustrates the (5, 50, 95)\%-quantiles of one smoothed endogenous state, the nominal interest rate $\hat{\mathsf{R}}_{t}$, and of one smoothed exogenous state, the productivity shock $\hat{\mathsf{z}}_t$.
The left and right panels correspond to fixing the standard deviations of the measurement errors and treating them as free parameters, respectively.
The differences between the smoothing distributions of these state variables are apparent.

\section{Concluding remarks}\label{sec: conclusion}

The literature on estimation and empirical applications of dynamic structural macrofinance models has mostly been relying on approximations using log-linearization to cast these models into linear Gaussian state-space models. It is now known that second-order approximation errors in the linear solutions can have first-order effects on likelihood estimation, and that errors in likelihood estimation can be compounded with the sample size. However, when models are solved using more accurate numerical methods such as high-order perturbation methods or projection methods, the solved models become non-linear and/or non-Gaussian state-space models with high-dimensional and complex structures.

In this paper, we propose a novel methodology to efficiently estimate the likelihood function of such state-space models by building on state-of-the-art SMC methods. While particle filters have been applied to estimate the likelihood of dynamic economic models, the successful application of this approach remains challenging due to the large variance of the resulting likelihood estimator. We develop an annealed controlled SMC method that delivers numerically stable and low variance estimators of the likelihood function, by adopting an annealing procedure to gradually introduce information from observations and construct globally optimal proposal distributions using approximate dynamic programming schemes. We provide a theoretical analysis to characterize the quality of our proposal distributions, which elucidates various properties of annealed controlled SMC.

To perform parameter inference, we develop a new adaptive SMC$^2$ algorithm that employs likelihood estimators from annealed controlled SMC. Under suitable assumptions, we establish law of large numbers and central limit theorems for our estimators of posterior expectations and the model evidence. We illustrate the strengths of our adaptive SMC$^2$ algorithm by estimating two popular macrofinance models: a prototypical new Keynesian dynamic stochastic general equilibrium model and a non-linear non-Gaussian consumption-based long-run risk asset pricing model.

We believe that the methods developed in this paper are readily applicable to DSGE models exhibiting strong non-linearities for instance due to the presence of effective lower bounds on interest rates targeted by monetary policy. Further as researchers move to the estimation of larger scale non-linear DSGE models with richer state spaces, the curse of dimensionality of particle filters (see e.g. \citet{bengtsson2008curse}) will call for the use of more efficient approaches, such as the one developed here.

\section*{Acknowledgements}
This work was funded by CY Initiative of Excellence (grant "Investissements d'Avenir" ANR-16-IDEX-0008).

\bibliography{ref}

\clearpage
\begin{center}
{\Large \bf Appendix}
\end{center}
\appendix

\renewcommand{\theequation}{\thesection.\arabic{equation}}

\setcounter{equation}{0}

\section{Proofs for annealed controlled sequential Monte Carlo}\label{sec:proofsanneal}

In this section, we detail the proofs of the results in Section \ref{sec:analysis1} on
annealed controlled sequential Monte Carlo.

\begin{proof}[Proof of Proposition \ref{prop:KL_decomp}]
Using the relation $\psi^*=\psi\cdot\phi^*$, which holds as $\phi^*$ is the optimal refinement of $\psi$ \citep[Proposition 1]{heng2020controlled},
and the form of the smoothing distribution $p(ds_{0:T}|y_{1:T},\theta,\lambda)$ in Equation \eqref{eqn:rewrite_smoothing},
we can decompose the log-density as
\begin{align}\label{eqn:decomp_exp}
	\log\left(\frac{p(s_{0:T} | y_{1:T},\theta,\lambda)}{q^{\psi\cdot\phi}(s_{0:T})}\right)
	&= \log\left(\frac{\phi_0^*(s_0)}{\phi_0(s_0)}\right) +
	\log\left(\frac{q_0(\psi_0\cdot\phi_0)}{q_0(\psi_0\cdot\phi_0^*)}\right) \notag\\
	&+ \sum_{t=1}^T\log\left(\frac{\phi_t^*(s_{t-1},s_t)}{\phi_t(s_{t-1},s_t)}\right)
	+ \sum_{t=1}^T\log\left(\frac{q_t(\psi_t\cdot\phi_t|s_{t-1})}{q_t(\psi_t\cdot\phi_t^*|s_{t-1})}\right).
\end{align}
To upper bound the expectations and conditional expectations in \eqref{eqn:decomp_exp}, we apply the log-sum inequality to obtain
\begin{align}
	\log\left(\frac{q_0(\psi_0\cdot\phi_0)}{q_0(\psi_0\cdot\phi_0^*)}\right)
	\leq \int_{\mathbb{S}}\frac{q_0(ds_0)\psi_0(s_0)\phi_0(s_0)}{q_0(\psi_0\cdot\phi_0)}
	\log\left(\frac{\phi_0(s_0)}{\phi_0^*(s_0)}\right) = q_0^{\psi\cdot\phi}(\log(\phi_0/\phi_0^*))
\end{align}
and
\begin{align}
	\log\left(\frac{q_t(\psi_t\cdot\phi_t|s_{t-1})}{q_t(\psi_t\cdot\phi_t^*|s_{t-1})}\right)
	&\leq \int_{\mathbb{S}}\frac{q_t(ds_t|s_{t-1})\psi_t(s_{t-1},s_t)\phi_t(s_{t-1},s_t)}{q_t(\psi_t\cdot\phi_t|s_{t-1})}
	\log\left(\frac{\phi_t(s_{t-1},s_t)}{\phi_t^*(s_{t-1},s_t)}\right) \notag\\
	&= q_t^{\psi\cdot\phi}(\log(\phi_t/\phi_t^*)|s_{t-1}).
\end{align}
Hence the KL divergence from $q^{\psi\cdot\phi}(ds_{0:T})$ to $p(ds_{0:T}|y_{1:T},\theta,\lambda)$
can be upper bounded by
\begin{align}
	&\mathrm{KL}\left(p(ds_{0:T}|y_{1:T},\theta,\lambda)  ~|~ q^{\psi\cdot\phi}(ds_{0:T}) \right)
	= \int_{\mathbb{S}^{T+1}} \log\left(\frac{p(s_{0:T} | y_{1:T},\theta,\lambda)}{q^{\psi\cdot\phi}(s_{0:T})}\right)
	p(ds_{0:T}|y_{1:T},\theta,\lambda)\notag \\
	&\leq q_0^*(\log(\phi_0^*/\phi_0)) + q_0^{\psi\cdot\phi}(\log(\phi_0/\phi_0^*)) \notag\\
	& + \sum_{t=1}^T(\mu_{t-1}^{*}\times q_t^{*})(\log(\phi_t^*/\phi_t))
	+ \sum_{t=1}^T(\mu_{t-1}^{*}\times q_t^{\psi\cdot\phi})(\log(\phi_t/\phi_t^*)).
\end{align}
Equation \eqref{eqn:KL_decomposition} follows by recalling that
$\xi_0^*=q_0^*$, $\xi_0^{\psi\cdot\phi}=q_0^{\psi\cdot\phi}$, and
$\xi_t^*=\mu_{t-1}^*\times q_t^*$, $\xi_t^{\psi\cdot\phi}=\mu_{t-1}^*\times q_t^{\psi\cdot\phi}$
for $t=1,\ldots,T$.
\end{proof}

\begin{proof}[Proof of Theorem \ref{thm:policy_learning}]
We first establish the backward recursion for the sequence $(\varepsilon_t^*)_{t=0}^T$.
By Assumption \ref{ass:marginals}$(i)$ and Jensen's inequality, the difference at the terminal time $T$ satisfies
\begin{align}\label{eqn:proof_terminaltime}
	\varepsilon_T^{*} \leq C_T \nu_T^{\psi}(\log(\phi_T^*/\phi_T))
	\leq C_T \| \log \phi_T - \log \phi_T^*\|_{L^2(\nu_T^{\psi})}.
\end{align}
Noting that $\phi_T=P_T^{\psi} w_T^{\psi}$, $\phi_T^*=w_T^{\psi}$ and $w_T^{\psi}\in\mathbb{G}_T^{\psi}$,
we have $\varepsilon_T^{*}\leq C_Te_T^{\psi}$ by applying Assumption \ref{ass:function_classes}$(ii)$.
For time $t=0,\ldots,T-1$, we consider the decomposition
\begin{align}\label{eqn:proof_decomp1}
	\varepsilon_t^{*} = \xi_t^{*}(\log(\phi_t^*/\tilde{\phi}_t))
						   + \xi_t^{*}(\log(\tilde{\phi}_t/\phi_t)),
\end{align}
with $\tilde{\phi}_t = w_t^{\psi} q_{t+1}^{\psi}(\phi_{t+1})$.
By the log-sum inequality, we have
\begin{align}
	\log(\phi_t^*/\tilde{\phi}_t) = \log\left(\frac{q_{t+1}^{\psi}(\phi_{t+1}^*)}{q_{t+1}^{\psi}(\phi_{t+1})}\right)
	= \log\left(\frac{q_{t+1}(\psi_{t+1}\phi_{t+1}^*)}{q_{t+1}(\psi_{t+1}\phi_{t+1})}\right)
	\leq q_{t+1}^{*}(\log(\phi_{t+1}^*/\phi_{t+1})).
\end{align}
Hence using the definition of the distributions $(\xi_t^*)_{t=0}^T$, the first term in Equation \eqref{eqn:proof_decomp1} satisfies
\begin{align}\label{eqn:decomp1_term1}
	\xi_t^{*}(\log(\phi_t^*/\tilde{\phi}_t)) \leq \xi_t^{*}(q_{t+1}^{*}(\log(\phi_{t+1}^*/\phi_{t+1})))
	= \xi_{t+1}^{*}(\log(\phi_{t+1}^*/\phi_{t+1})) = \varepsilon_{t+1}^*.
\end{align}
For the second term in Equation \eqref{eqn:proof_decomp1}, as before we apply
Assumption \ref{ass:marginals}$(i)$ and Jensen's inequality to obtain
\begin{align}\label{eqn:proof_consider1}
	\xi_t^{*}(\log(\tilde{\phi}_t/\phi_t)) \leq C_t \nu_t^{\psi}(\log(\tilde{\phi}_t/\phi_t))
	\leq C_t \| \log \phi_t - \log \tilde{\phi}_t \|_{L^2(\nu_t^{\psi})}.
\end{align}
Since $\phi_t = P_t^{\psi}\tilde{\phi}_t$ and $\tilde{\phi}_t\in\mathbb{G}_t^{\psi}$,
applying Assumption \ref{ass:function_classes}$(ii)$ gives
\begin{align}\label{eqn:decomp1_term2}
	\xi_t^{*}(\log(\tilde{\phi}_t/\phi_t)) \leq C_t e_t^{\psi}.
\end{align}
Combining \eqref{eqn:proof_decomp1}, \eqref{eqn:decomp1_term1} and \eqref{eqn:decomp1_term2}
gives the recursion
\begin{align}
	\varepsilon_t^* \leq \varepsilon_{t+1}^* + C_t e_t^{\psi},\quad t=0,\ldots,T-1.
\end{align}

We now consider the sequence $(\varepsilon_t^{\psi\cdot\phi})_{t=0}^T$.
Using Assumptions \ref{ass:function_classes}$(ii)$ \& \ref{ass:marginals} and the same arguments
in Equation \eqref{eqn:proof_terminaltime}, we obtain $\varepsilon_T^{\psi\cdot\phi}\leq C_T M_Te_T^{\psi}$.
As before, we consider the decomposition
\begin{align}\label{eqn:proof_decomp2}
	\varepsilon_t^{\psi\cdot\phi} = \xi_t^{\psi\cdot\phi}(\log(\phi_t/\tilde{\phi}_t))
						   + \xi_t^{\psi\cdot\phi}(\log(\tilde{\phi}_t/\phi_t^*)),
\end{align}
for time $t=0,\ldots,T-1$.
Again using Assumptions \ref{ass:function_classes}$(ii)$ \& \ref{ass:marginals}
and the same arguments in Equation \eqref{eqn:proof_consider1}, the first term of \eqref{eqn:proof_decomp2} satisfies
\begin{align}\label{eqn:proof_decomp2_term1}
	\xi_t^{\psi\cdot\phi}(\log(\phi_t/\tilde{\phi}_t)) \leq C_tM_t e_t^{\psi}.
\end{align}
By the log-sum inequality, we have
\begin{align}
	\log(\tilde{\phi}_t/\phi_t^*) = \log\left(\frac{q_{t+1}^{\psi}(\phi_{t+1})}{q_{t+1}^{\psi}(\phi_{t+1}^*)}\right)
	= \log\left(\frac{q_{t+1}(\psi_{t+1}\phi_{t+1})}{q_{t+1}(\psi_{t+1}\phi_{t+1}^*)}\right)
	\leq q_{t+1}^{\psi\cdot\phi}(\log(\phi_{t+1}/\phi_{t+1}^*)),
\end{align}
therefore the second term in Equation \eqref{eqn:proof_decomp2} satisfies
\begin{align}
	\xi_t^{\psi\cdot\phi}(\log(\tilde{\phi}_t/\phi_t^*)) \leq \xi_t^{\psi\cdot\phi}(q_{t+1}^{\psi\cdot\phi}(\log(\phi_{t+1}/\phi_{t+1}^*))). 	
%	= (\mu_{t-1}^*q_t^{\psi\cdot\phi})(q_{t+1}^{\psi\cdot\phi}(\log(\phi_{t+1}/\phi_{t+1}^*))),
\end{align}
%with $\mu_{-1}^*q_0^{\psi\cdot\phi}=q_0^{\psi\cdot\phi}$ for the case $t=0$.		
Using Assumption \ref{ass:marginals}$(ii)$ and the definition of the distributions $(\xi_t^{\psi\cdot\phi})_{t=0}^T$, we have
\begin{align}\label{eqn:proof_decomp2_term2}
	\xi_t^{\psi\cdot\phi}(\log(\tilde{\phi}_t/\phi_t^*)) \leq M_t \mu_t^*(q_{t+1}^{\psi\cdot\phi}(\log(\phi_{t+1}/\phi_{t+1}^*))) = M_t \xi_{t+1}^{\psi\cdot\phi}(\log(\phi_{t+1}/\phi_{t+1}^*)) = M_t\varepsilon_{t+1}^{\psi\cdot\phi}.
\end{align}
Combining \eqref{eqn:proof_decomp2}, \eqref{eqn:proof_decomp2_term1} and \eqref{eqn:proof_decomp2_term2} gives
\begin{align}
	\varepsilon_{t}^{\psi\cdot\phi} \leq M_t\varepsilon_{t+1}^{\psi\cdot\phi} + C_tM_t e_t^{\psi} ,\quad t=0,\ldots,T-1.
\end{align}
\end{proof}

\section{Conditional implementation of controlled SMC}\label{appendix:cond_smc}

Our adaptive SMC$^2$ algorithm in Section \ref{subsec:smc2} relies on a conditional implementation of controlled SMC (Algorithm \ref{alg:cSMC})
as detailed below.
\begin{algorithm}[h]
\protect\caption{\footnotesize{Conditional sequential Monte Carlo at inverse temperature $\lambda\in[0,1]$} \label{alg:condSMC}}
{\footnotesize
\textbf{Input}: number of particles $N$, policy $\psi=(\psi_t)_{t=0}^T$ and reference trajectory
$s_{0:T}^{(k)}=(s_{t}^{(l_{t})})_{t=0}^{T}$.

(1) For time $t=0$ and particle $n=1,\ldots,N$.
\ignorespacesafterend
\begin{description}[itemsep=0pt,parsep=0pt,topsep=0pt,labelindent=0.5cm]
	\item (1a) If $n\neq l_0$, sample state $s_{0}^{(n)}\sim q_0^{\psi}$.

	\item (1b) Compute normalized weights $W_0^{(n)}=w_0^{\psi}(s_{0}^{(n)})/\sum_{m=1}^Nw_0^{\psi}(s_{0}^{(m)})$.
\end{description}

(2) For time $t=1,\ldots,T$ and particle $n=1,\ldots,N$.

\begin{description}[itemsep=0pt,parsep=0pt,topsep=0pt,labelindent=0.5cm]
	\item (2a) If $n=l_t$, set ancestor $a_{t-1}^{(n)}=l_{t-1}$, else
	sample ancestor $a_{t-1}^{(n)}\sim r(\cdot|W_{t-1}^{(1)},\ldots,W_{t-1}^{(N)})$.

	\item (2b) If $n\neq l_t$, sample state $s_t^{(n)} \sim q_t^{\psi}(\cdot|s_{t-1}^{(a_{t-1}^{(n)})})$.

	\item (2c) Compute normalized weights $W_t^{(n)}=w_t^{\psi}(s_{t-1}^{(a_{t-1}^{(n)})},s_{t}^{(n)})/
	\sum_{m=1}^Nw_t^{\psi}(s_{t-1}^{(a_{t-1}^{(m)})},s_{t}^{(m)})$.
\end{description}

(3) Compute likelihood estimator $\hat{p}(y_{1:T}|\theta,\lambda)=\{\frac{1}{N}\sum_{n=1}^Nw_0^{\psi}(s_0^{(n)})\}
	\{\prod_{t=1}^T\frac{1}{N}\sum_{n=1}^Nw_t^{\psi}(s_{t-1}^{(a_{t-1}^{(n)})}, s_t^{(n)})\}$.

(4) Sample an ancestor $b_T\sim r(\cdot|W_{T}^{(1)},\ldots,W_{T}^{(N)})$ and set
ancestral lineage as $b_t = a_t^{(b_{t+1})}$ for $t=T-1,\ldots,0$.

\textbf{Output}: likelihood estimator $\hat{p}(y_{1:T}|\theta,\lambda)$ and trajectory $(s_t^{(b_t)})_{t=0}^T$.
}
\end{algorithm}

\section{Proofs for sequential Monte Carlo$^2$}\label{sec:proofsmc2}
This section provides a conceptual framework to understand our SMC$^2$ in Algorithm \ref{alg:SMC-sq} and
the necessary arguments to establish the consistency results of Section \ref{subsec:consistency}.

At inverse temperature $\lambda\in[0,1]$, the target distribution defined on $\mathbb{X}=\Theta\times\mathbb{S}^{T+1}$ is
\begin{align}\label{eqn:target_distribution}
	\pi(dx|\lambda) = \pi(d\theta, ds_{0:T} | \lambda) = p(d\theta, ds_{0:T} | y_{1:T},\lambda).
\end{align}
To facilitate our analysis, we define the following extended target distribution on
$\tilde{\mathbb{X}}=\Theta\times\mathbb{F}_{0:T}\times\mathbb{S}^{(T+1)N}\times\{1,\ldots,N\}^{TN}\times\{1,\ldots,N\}$
\begin{align}\label{eq:extended_target_pmmh}
\tilde{\pi}(dx|\lambda)&=\tilde{\pi}(d\theta,d\psi,ds_{0:T}^{(1:N)},a_{0:T-1}^{(1:N)},k|\lambda)\\&=\frac{p(d\theta)p_{\textrm{ADP}}(d\psi|\theta,\lambda)p_{\textrm{SMC}}(ds_{0:T}^{(1:N)},a_{0:T-1}^{(1:N)}|\theta,\lambda,\psi)W_{T}^{(k)}\hat{p}(y_{1:T}|\theta,\lambda)}{p(y_{1:T}|\lambda)}.\notag
\end{align}
In the above, $p_{\textrm{ADP}}(d\psi|\theta,\lambda)$ denotes a distribution on $\mathbb{F}_{0:T}=\mathbb{F}_{0}\times\cdots\times\mathbb{F}_{T}$, defined by how a policy $\psi$ for parameter $\theta$ at inverse temperature $\lambda$ is constructed using ADP. This generic description allows us to accommodate various policy learning strategies. The law of states $s_{0:T}^{(1:N)}=(s_{t}^{(n)})_{t=0,n=1}^{T,N}$ and ancestors $a_{0:T-1}^{(1:N)}=(a_{t}^{(n)})_{t=0,n=1}^{T-1,N}$ under controlled SMC in Algorithm \ref{alg:cSMC}, with policy $\psi$ at parameter $\theta$ and inverse temperature $\lambda$, is given by
\begin{align}
&p_{\textrm{SMC}}(ds_{0:T}^{(1:N)},a_{0:T-1}^{(1:N)}|\theta,\lambda,\psi)\\&=\left\{ \prod_{n=1}^{N}q_{0}^{\psi}(ds_{0}^{(n)}|\theta,\lambda)\right\} \left\{ \prod_{t=1}^{T}r(a_{t-1}^{(1)},\ldots,a_{t-1}^{(N)}|W_{t-1}^{(1)},\ldots,W_{t-1}^{(N)})\prod_{n=1}^Nq_{t}^{\psi}(ds_{t}^{(n)}|s_{t-1}^{(a_{t-1}^{(n)})})\right\} .\notag
\end{align}
Under the multinomial resampling scheme, we have
\begin{align}
r(a_{t}^{(1)},\ldots,a_{t}^{(N)}|W_{t}^{(1)},\ldots,W_{t}^{(N)})=\prod_{n=1}^{N}W_{t}^{(a_{t}^{(n)})}.
\end{align}
The notation $W_{T}^{(k)}$ refers to the normalized weight of particle $k=1,\ldots,N$ at time $T$,
and $\hat{p}(y_{1:T}|\theta,\lambda)$ denotes the controlled SMC likelihood estimator.
Equation \eqref{eq:extended_target_pmmh} can be rewritten as 	
\begin{align}\label{eq:extended_target_pmmh_rewrite}
&\tilde{\pi}(d\theta,d\psi,ds_{0:T}^{(1:N)},a_{0:T-1}^{(1:N)},k|\lambda)\\
&=\frac{1}{N^{T+1}}
\pi(d\theta, ds_{0:T}^{(k)} | \lambda)
p_{\textrm{ADP}}(d\psi|\theta,\lambda)
p_{\textrm{CSMC}}(ds_{0:T}^{-(k)},a_{0:T-1}^{-(k)}|s_{0:T}^{(k)},\theta,\lambda,\psi),\notag
\end{align}	
where the trajectory $s_{0:T}^{(k)}=(s_{t}^{(l_{t})})_{t=0}^{T}$ is formed by tracing the ancestral lineage of $s_{T}^{(k)}$, i.e. $l_{T}=k$ and $l_{t}=a_{t}^{(l_{t+1})}$ for $ t=T-1,\ldots,0$, and
\begin{align}
&p_{\textrm{CSMC}}(ds_{0:T}^{-(k)},a_{0:T-1}^{-(k)}|s_{0:T}^{(k)},\theta,\lambda,\psi)\\
&=\frac{p_{\textrm{SMC}}(ds_{0:T}^{(1:N)},a_{0:T-1}^{(1:N)}|\theta,\lambda,\psi)}{q_{0}^{\psi}(ds_{0}^{(l_{0})}|\theta,\lambda)\prod_{t=1}^{T}r(l_{t-1}|W_{t-1}^{(1)},\ldots,W_{t-1}^{(N)})q_{t}^{\psi}(ds_{t}^{(l_{t})}|s_{t-1}^{(l_{t-1})})}\notag
\end{align}
is the law of the other states $s_{0:T}^{-(k)}$ and ancestors $a_{0:T-1}^{-(k)}$
generated in the conditional implementation of controlled SMC (Algorithm \ref{alg:condSMC} of this Appendix).
Note from Equation \eqref{eq:extended_target_pmmh_rewrite} that the marginal distribution
of $(\theta,s_{0:T}^{(k)})$ under the extended target distribution on $\tilde{\mathbb{X}}$
is exactly the desired target distribution on $\mathbb{X}$ in Equation \eqref{eqn:target_distribution}.

To sample from the extended target distribution in Equation \eqref{eq:extended_target_pmmh},
Steps 3c and 3d of Algorithm \ref{alg:SMC-sq} define the following forward Markov transition kernel
$\tilde{f}_{\textrm{ADP-CSMC}}$ on $\tilde{\mathbb{X}}$
\begin{align}\label{eqn:forward_ADP_CSMC}
&\tilde{f}_{\textrm{ADP-CSMC}}(d\tilde{x}|x,\lambda)
=\tilde{f}_{\textrm{ADP-CSMC}}(d\tilde{\theta},d\tilde{\psi},d\tilde{s}_{0:T}^{(1:N)},\tilde{a}_{0:T-1}^{(1:N)},\tilde{k}|\theta,\psi,s_{0:T}^{(1:N)},a_{0:T-1}^{(1:N)},k,\lambda) \notag\\
&=\underbrace{p_{\textrm{ADP}}(d\tilde{\psi}|\theta,\lambda)}_{\textrm{Step 3c}}
~\underbrace{p_{\textrm{CSMC}}(d\tilde{s}_{0:T}^{-(k)},\tilde{a}_{0:T-1}^{-(k)}|s_{0:T}^{(k)},\theta,\lambda,\tilde{\psi})}_{\textrm{Step 3d}}
~\delta_{(\theta,s_{0:T}^{(k)},k)}(d\tilde{\theta},d\tilde{s}_{0:T}^{(k)},\tilde{k}).	
\end{align}
For the purpose of importance sampling, we associate $\tilde{f}_{\textrm{ADP-CSMC}}$ with a backward Markov transition kernel $\tilde{b}_{\textrm{ADP-CSMC}}$ defined on $\tilde{\mathbb{X}}$ as
\begin{align}\label{eqn:backward_kernel_ADP_CSMC}
&\tilde{b}_{\textrm{ADP-CSMC}}(dx|\tilde{x},\lambda) = \tilde{b}_{\textrm{ADP-CSMC}}(d\theta,d\psi,ds_{0:T}^{(1:N)},a_{0:T-1}^{(1:N)},k|\tilde{\theta},\tilde{\psi},\tilde{s}_{0:T}^{(1:N)},\tilde{a}_{0:T-1}^{(1:N)},\tilde{k},\lambda)\notag\\
&=\delta_{(\tilde{\theta},\tilde{s}_{0:T}^{(k)},\tilde{k})}(d\theta,ds_{0:T}^{(k)},k)p_{\textrm{ADP}}(d\psi|\theta,\lambda)p_{\textrm{CSMC}}(ds_{0:T}^{-(k)},a_{0:T-1}^{-(k)}|s_{0:T}^{(k)},\theta,\lambda,\psi).
\end{align}
To improve the sample diversity of resampled parameters, we consider PMMH Markov transitions in
Step 4 of Algorithm \ref{alg:SMC-sq}. In particular, Steps 4a, 4b and 4c
define the following proposal transition kernel $\tilde{q}_{\textrm{ADP-SMC}}$ on $\tilde{\mathbb{X}}$
\begin{align}\label{eqn:proposal_pmmh_extendedspace}
&\tilde{q}_{\textrm{ADP-SMC}}(d\tilde{x}|x,\lambda)
=\tilde{q}_{\textrm{ADP-SMC}}(d\tilde{\theta},d\tilde{\psi},d\tilde{s}_{0:T}^{(1:N)},\tilde{a}_{0:T-1}^{(1:N)},\tilde{k}|\theta,\psi,s_{0:T}^{(1:N)},a_{0:T-1}^{(1:N)},k,\lambda)\notag\\
&=\underbrace{h(d\tilde{\theta}|\theta,\lambda)}_{\textrm{Step 4a}}
~\underbrace{p_{\textrm{ADP}}(d\tilde{\psi}|\tilde{\theta},\lambda)}_{\textrm{Step 4b}}~\underbrace{p_{\textrm{SMC}}(d\tilde{s}_{0:T}^{(1:N)},\tilde{a}_{0:T-1}^{(1:N)}|\tilde{\theta},\lambda,\tilde{\psi})W_{T}^{(\tilde{k})}}_{\textrm{Step 4c}},
\end{align}
where $h(d\tilde{\theta}|\theta,\lambda)$ is a proposal transition kernel on $\Theta$ at
inverse temperature $\lambda$.
The resulting PMMH Markov transition kernel on $\tilde{\mathbb{X}}$ induced by the accept-reject procedure in Step 4d is
\begin{align}\label{eqn:pmmh_extendedspace}
	\tilde{f}_{\textrm{PMMH}}(d\tilde{x}|x,\lambda) &=
	\tilde{\alpha}(\tilde{x}|x,\lambda)\tilde{q}_{\textrm{ADP-SMC}}(d\tilde{x}|x,\lambda)~ +\\
	&\left(1-\int_{\tilde{\mathbb{X}}}\tilde{\alpha}(z|x,\lambda)\tilde{q}_{\textrm{ADP-SMC}}(dz|x,\lambda)\right)
	\delta_{x}(d\tilde{x}),\notag
\end{align}
where the acceptance probability $\tilde{\alpha}(\tilde{x}|x,\lambda)=\alpha(\tilde{\theta}|\theta,\lambda)$ 
is given in Equation \eqref{eqn:PMMH_acceptprob}. 
We denote the composition of $\tilde{f}_{\textrm{PMMH}}$ over $K\in\mathbb{N}$ iterations as
$\tilde{f}_{\textrm{PMMH}}^K$.
To perform importance sampling, we associate $\tilde{f}_{\textrm{PMMH}}^{K}$ with a backward Markov transition kernel $\tilde{b}_{\textrm{PMMH}}^{K}$ on $\tilde{\mathbb{X}}$ satisfying
\begin{align}\label{eqn:backward_kernel_PMMH}
\tilde{\pi}(d\tilde{x}|\lambda)\tilde{b}_{\textrm{PMMH}}^{K}(dx|\tilde{x},\lambda)
=\tilde{\pi}(dx|\lambda)\tilde{f}_{\textrm{PMMH}}^{K}(d\tilde{x}|x,\lambda).
\end{align}
To examine the effect of applying $\tilde{f}_{\textrm{ADP-CSMC}}$ and $\tilde{f}_{\textrm{PMMH}}$
for the variables $(\theta,s_{0:T}^{(k)})$ on the marginal space $\mathbb{X}$,
we will denote their respective marginal Markov transition kernels as
$f_{\textrm{ADP-CSMC}}(d\tilde{\theta},d\tilde{s}_{0:T}|\theta,s_{0:T},\lambda)$
and $f_{\textrm{PMMH}}(d\tilde{\theta},d\tilde{s}_{0:T}|\theta,s_{0:T},\lambda)$.
Similarly, the $K$-fold composition of $f_{\textrm{PMMH}}$ will be written as $f_{\textrm{PMMH}}^K$.
By composing Steps 3c, 3d and 4 of Algorithm \ref{alg:SMC-sq}, we obtain the following Markov
transition kernel on $\mathbb{X}$
\begin{align}\label{eqn:marginal_composed_kernel}
	&m(d\tilde{x}|x,\lambda) = m(d\tilde{\theta},d\tilde{s}_{0:T}|\theta,s_{0:T},\lambda) \\
	&= \int_{\mathbb{X}} f_{\textrm{ADP-CSMC}}(d\bar{\theta},d\bar{s}_{0:T}|\theta,s_{0:T},\lambda)
	f_{\textrm{PMMH}}^K(d\tilde{\theta},d\tilde{s}_{0:T}|\bar{\theta},\bar{s}_{0:T},\lambda),\notag
\end{align}
which was introduced in Section \ref{subsec:consistency}.
The following establishes some properties of the Markov kernels we have introduced.

\begin{lemma}\label{lem:properties_kernels}
For any inverse temperature $\lambda\in[0,1]$, the Markov transition kernels $\tilde{f}_{\textrm{ADP-CSMC}}$, $\tilde{f}_{\textrm{PMMH}}$,
$f_{\textrm{ADP-CSMC}}$, $f_{\textrm{PMMH}}$ and $m$ satisfy:
	\begin{enumerate}[label=(\roman*)]
		\item $\tilde{f}_{\textrm{ADP-CSMC}}$ and $\tilde{f}_{\textrm{PMMH}}$ are invariant with respect
		to the extended target distribution $\tilde{\pi}$ on $\tilde{\mathbb{X}}$ in Equation \eqref{eq:extended_target_pmmh};
		
		\item $f_{\textrm{ADP-CSMC}}$ and $f_{\textrm{PMMH}}$ are invariant with respect
		to the target distribution $\pi$ on $\mathbb{X}$ in Equation \eqref{eqn:target_distribution};
		
		\item $m$ is invariant with respect
		to the target distribution $\pi$ on $\mathbb{X}$ in Equation \eqref{eqn:target_distribution}.
	\end{enumerate}	
\end{lemma}

\begin{proof}
Using Equation \eqref{eq:extended_target_pmmh_rewrite}, by rewriting
$\tilde{f}_{\textrm{ADP-CSMC}}$ in Equation \eqref{eqn:forward_ADP_CSMC} as
\begin{align}
\tilde{f}_{\textrm{ADP-CSMC}}(d\tilde{\theta},d\tilde{\psi},d\tilde{s}_{0:T}^{(1:N)},\tilde{a}_{0:T-1}^{(1:N)},\tilde{k}|\theta,\psi,s_{0:T}^{(1:N)},a_{0:T-1}^{(1:N)},k,\lambda)\notag\\
=\tilde{\pi}(d\tilde{\psi},d\tilde{s}_{0:T}^{-(k)},\tilde{a}_{0:T-1}^{-(k)}|\theta,s_{0:T}^{(k)},k,\lambda)
~\delta_{(\theta,s_{0:T}^{(k)},k)}(d\tilde{\theta},d\tilde{s}_{0:T}^{(k)},\tilde{k}),	
\end{align}
it follows that Steps 3c and 3d of Algorithm \ref{alg:SMC-sq} can be seen as
sampling the variables $(\psi,s_{0:T}^{-(k)},a_{0:T-1}^{-(k)})$ from the conditional distribution of the extended target distribution.
Hence $\tilde{f}_{\textrm{ADP-CSMC}}$ is a Gibbs update which leaves the extended target distribution $\tilde{\pi}$ invariant.

Next, we show that $\tilde{f}_{\textrm{PMMH}}$ in Equation \eqref{eqn:pmmh_extendedspace} can be understood
as a standard Metropolis--Hastings transition kernel on the extended space $\tilde{\mathbb{X}}$,
with proposal transition kernel $\tilde{q}_{\textrm{ADP-SMC}}$ in Equation \eqref{eqn:proposal_pmmh_extendedspace} and
the extended target distribution $\tilde{\pi}$ as invariant distribution.
The Metropolis--Hastings acceptance probability of a transition from $x=(\theta,\psi,s_{0:T}^{(1:N)},a_{0:T-1}^{(1:N)},k)\in\tilde{\mathbb{X}}$ to
$\tilde{x}=(\tilde{\theta},\tilde{\psi},\tilde{s}_{0:T}^{(1:N)},\tilde{a}_{0:T-1}^{(1:N)},\tilde{k})\in\tilde{\mathbb{X}}$ is given by $\min\{1,R_{\textrm{MH}}(x,\tilde{x}|\lambda)\}$, where
$R_{\textrm{MH}}(x,\tilde{x}|\lambda)$ the Radon--Nikodym derivative between the measures
$\tilde{\pi}(d\tilde{x}|\lambda)\tilde{q}_{\textrm{ADP-SMC}}(dx|\tilde{x},\lambda)$
and $\tilde{\pi}(dx|\lambda)\tilde{q}_{\textrm{ADP-SMC}}(d\tilde{x}|x,\lambda)$ on
$\tilde{\mathbb{X}}\times\tilde{\mathbb{X}}$. Using the form of the proposal transition kernel
in Equation \eqref{eqn:proposal_pmmh_extendedspace} and
the extended target distribution in Equation \eqref{eq:extended_target_pmmh}, we have
\begin{align}
	R_{\textrm{MH}}(x,\tilde{x}|\lambda) =
	\frac{p(\tilde{\theta})\hat{p}(y_{1:T}|\tilde{\theta},\lambda)h(\theta|\tilde{\theta},\lambda)}
	{p(\theta)\hat{p}(y_{1:T}|\theta,\lambda)h(\tilde{\theta}|\theta,\lambda)},
\end{align}
where $\hat{p}(y_{1:T}|\theta,\lambda)$ and $\hat{p}(y_{1:T}|\tilde{\theta},\lambda)$
are the controlled SMC likelihood estimators at the parameter-policy configurations
$(\theta,\psi)$ and $(\tilde{\theta},\tilde{\psi})$, respectively.
Hence the Metropolis--Hastings acceptance probability $\min\{1,R_{\textrm{MH}}(x,\tilde{x}|\lambda)\}$
coincides with the acceptance probability $\tilde{\alpha}(\tilde{x}|x,\lambda)=\alpha(\tilde{\theta}|\theta,\lambda)$ in Equation \eqref{eqn:PMMH_acceptprob}.
Although $\tilde{f}_{\textrm{PMMH}}$ operates on the extended space, only the variables
$(\theta, s_{0:T}^{(k)}, \hat{p}(y_{1:T}|\theta,\lambda))$ have to be stored,
as described in Step 4d of Algorithm \ref{alg:SMC-sq}. This completes part $(i)$ of the proof.

The claim in part $(ii)$ follows from part $(i)$ and the fact that the extended target distribution $\tilde{\pi}$ on
$\tilde{\mathbb{X}}$ in Equation \eqref{eq:extended_target_pmmh_rewrite} admits the target distribution $\pi$ on
$\mathbb{X}$ in Equation \eqref{eqn:target_distribution} as marginal distribution.

Lastly, part $(iii)$ is an immediate consequence of part $(ii)$ as the
Markov transition kernel $m$ in Equation \eqref{eqn:marginal_composed_kernel} is a composition
of $f_{\textrm{ADP-CSMC}}$ and $f_{\textrm{PMMH}}^K$ that leave the target distribution $\pi$ on $\mathbb{X}$
invariant.
\end{proof}

With the above preliminaries in place, the following key results cast our SMC$^2$ algorithm
within the SMC sampler frameworks of \citet{del2006sequential} and \citet{beskos2016convergence}.

\begin{proposition}\label{prop:SMC_samplers}
The SMC$^2$ in Algorithm \ref{alg:SMC-sq} without adaptation in Steps 2 and 4 is an SMC sampler of \citet{del2006sequential}, operating on the extended space $\tilde{\mathbb{X}}$ with
initial distribution $\tilde{\pi}(dx|\lambda_0)$ and
the following properties at iteration $i=1,\ldots,I$:
\begin{enumerate}[label=(\roman*)]
	\item	the target distribution $\tilde{\pi}(dx|\lambda_i)$ defined in
	Equation \eqref{eq:extended_target_pmmh};
	
	\item the forward transition kernels $\tilde{f}_{\textrm{ADP-CSMC}}(d\tilde{x}|x,\lambda_i)$
	followed by $\tilde{f}_{\textrm{PMMH}}^K(d\tilde{x}|x,\lambda_i)$ defined in Equations
	\eqref{eqn:forward_ADP_CSMC} and \eqref{eqn:pmmh_extendedspace};
	
	\item the backward transition kernels $\tilde{b}_{\textrm{PMMH}}^{K}(dx|\tilde{x},\lambda_i)$ followed by
	$\tilde{b}_{\textrm{ADP-CSMC}}(dx|\tilde{x},\lambda_{i-1})$ defined in Equations
	\eqref{eqn:backward_kernel_PMMH} and \eqref{eqn:backward_kernel_ADP_CSMC}.
	
\end{enumerate}	
\end{proposition}

\begin{proof}
The initialization step of an SMC sampler would generate $P\in\mathbb{N}$ independent samples from
$\tilde{\pi}(dx|\lambda_0)$ at inverse temperature $\lambda_{0}=0$.
This corresponds to Step 1 of Algorithm \ref{alg:SMC-sq}, where only the parameters and trajectories
$(\theta^{(p)},s_{0:T}^{(p)})_{p=1}^{P}$ are sampled from
the marginal target distribution $\pi(d\theta,ds_{0:T}|\lambda_{0})$ on $\mathbb{X}$ as
the other variables are not needed to progress to the next iteration.

At iteration $i=1,\ldots,I-1$, suppose we have samples
\begin{align}\label{eqn:smc2_particles}
	(X^{(p)})_{p=1}^P=(\theta^{(p)},\psi^{(p)},s_{0:T}^{(p,1:N)},a_{0:T-1}^{(p,1:N)},k^{(p)})_{p=1}^{P}
\end{align}
approximating $\tilde{\pi}(dx|\lambda_{i-1})$.
In our algorithmic implementation, we only have to store the parameters and trajectories
$(\theta^{(p)},s_{0:T}^{(p)})_{p=1}^{P}$ at this stage,
where the trajectory $s_{0:T}^{(p)}=(s_{t}^{(p,l_{t}^{(p)})})_{t=0}^{T}$ is obtained by setting $l_{T}^{(p)}=k^{(p)}$
and $l_{t}^{(p)}=a_{t}^{(p,l_{t+1}^{(p)})}$ for $t=T-1,\ldots,0$.
It follows from Equation \eqref{eq:extended_target_pmmh_rewrite} that the samples $(\theta^{(p)},s_{0:T}^{(p)})_{p=1}^{P}$ will approximate the marginal target distribution $\pi(dx|\lambda_{i-1})$ on $\mathbb{X}$.

To approximate the target distribution $\tilde{\pi}(dx|\lambda_{i})$ at the next inverse temperature $\lambda_i$,
Steps 3c, 3d and 4 of Algorithm \ref{alg:SMC-sq} can be understood as moving the samples
$(X^{(p)})_{p=1}^P$ in Equation \eqref{eqn:smc2_particles} according to
the forward transition kernel $\bar{X}^{(p)}\sim\tilde{f}_{\textrm{ADP-CSMC}}(\cdot|X^{(p)},\lambda_i)$,
followed by $\tilde{X}^{(p)}\sim\tilde{f}_{\textrm{PMMH}}^K(\cdot|\bar{X}^{(p)},\lambda_i)$.
Using the backward transition kernel $\tilde{b}_{\textrm{PMMH}}^{K}(dx|\tilde{x},\lambda_i)$
followed by $\tilde{b}_{\textrm{ADP-CSMC}}(dx|\tilde{x},\lambda_{i-1})$,
the SMC sampler would assign an importance weight to each sample $(X^{(p)},\bar{X}^{(p)},\tilde{X}^{(p)})$.
This importance weight $W(x,\bar{x},\tilde{x}|\lambda_{i-1},\lambda_i)$ is given by the Radon--Nikodym derivative between the measures
\begin{align}
	\tilde{\pi}(d\tilde{x}|\lambda_{i})\tilde{b}_{\textrm{PMMH}}^{K}(d\bar{x}|\tilde{x},\lambda_i)
	\tilde{b}_{\textrm{ADP-CSMC}}(dx|\bar{x},\lambda_{i-1})
\end{align}
and
\begin{align}
\tilde{\pi}(dx|\lambda_{i-1})\tilde{f}_{\textrm{ADP-CSMC}}(d\bar{x}|x,\lambda_i)
\tilde{f}_{\textrm{PMMH}}^K(d\tilde{x}|\bar{x},\lambda_i).
\end{align}
on $\tilde{\mathbb{X}}\times\tilde{\mathbb{X}}\times\tilde{\mathbb{X}}$.
We will rewrite this importance weight as
\begin{align}
W(x,\bar{x},\tilde{x}|\lambda_{i-1},\lambda_i) =
R_{\textrm{ADP-CSMC}}(x,\bar{x}|\lambda_{i-1},\lambda_i)R_{\textrm{PMMH}}(\bar{x},\tilde{x}|\lambda_i),	
\end{align}
where $R_{\textrm{ADP-CSMC}}(x,\bar{x}|\lambda_{i-1},\lambda_i)$ is the Radon--Nikodym derivative between
the measures $\tilde{\pi}(d\bar{x}|\lambda_{i})\tilde{b}_{\textrm{ADP-CSMC}}(dx|\bar{x},\lambda_{i-1})$ and
$\tilde{\pi}(d{x}|\lambda_{i-1})\tilde{f}_{\textrm{ADP-CSMC}}(d\bar{x}|x,\lambda_i)$ on 
$\tilde{\mathbb{X}}\times\tilde{\mathbb{X}}$, and $R_{\textrm{PMMH}}(\bar{x},\tilde{x}|\lambda_i)$
is the Radon--Nikodym derivative between
the measures \\$\tilde{\pi}(d\tilde{x}|\lambda_{i})$$\tilde{b}_{\textrm{PMMH}}^{K}(d\bar{x}|\tilde{x},\lambda_i)$ and
$\tilde{\pi}(d\bar{x}|\lambda_{i})\tilde{f}_{\textrm{PMMH}}^K(d\tilde{x}|\bar{x},\lambda_i)$ on
$\tilde{\mathbb{X}}\times\tilde{\mathbb{X}}$.
Using the form of the extended target distribution in Equation \eqref{eq:extended_target_pmmh_rewrite},
we have
\begin{align}\label{eqn:density_adp_csmc}
	R_{\textrm{ADP-CSMC}}(x,\bar{x}|\lambda_{i-1},\lambda_i) =
	\frac{\pi(\theta,s_{0:T}^{(k)}|\lambda_i)}{\pi(\theta,s_{0:T}^{(k)}|\lambda_{i-1})}
	=\frac{p(y_{1:T}|\lambda_{i-1})}{p(y_{1:T}|\lambda_i)}
	\prod_{t=1}^Tg_{\theta}(y_t|s_{t-1}^{(l_{t-1})},s_t^{(l_t)})^{\lambda_i-\lambda_{i-1}}.
\end{align}
From the construction in Equation \eqref{eqn:backward_kernel_PMMH}, we have
$R_{\textrm{PMMH}}(\bar{x},\tilde{x}|\lambda_i)=1$.
Therefore the importance weight $W(x,\bar{x},\tilde{x}|\lambda_{i-1},\lambda_i) =
R_{\textrm{ADP-CSMC}}(x,\bar{x}|\lambda_{i-1},\lambda_i)$ corresponds to
the unnormalized weight in Step 3a of Algorithm \ref{alg:SMC-sq} and
Equation \eqref{eqn:adapt_weights}, up to the unknown normalization
constants $p(y_{1:T}|\lambda_{i-1})$ and $p(y_{1:T}|\lambda_{i})$.
Noting from Equation \eqref{eqn:density_adp_csmc} that the importance weight
only depends on the parameter and trajectory $(\theta,s_{0:T}^{(k)})$
before the forward Markov transitions
$\tilde{f}_{\textrm{ADP-CSMC}}$ and $\tilde{f}_{\textrm{PMMH}}^K$ are applied in Steps 3c, 3d and 4,
this justifies first performing resampling in Step 3b of Algorithm \ref{alg:SMC-sq}
\citep[Remark 1]{del2006sequential}.

Lastly, we clarify the storage requirements in Steps 3c and 3d.
For each particle $p=1,\ldots,P$, after sampling
\begin{align}
	(\bar{\theta}^{(p)},\bar{\psi}^{(p)},\bar{s}_{0:T}^{(p,1:N)},\bar{a}_{0:T-1}^{(p,1:N)},\bar{k}^{(p)})\sim \tilde{f}_{\textrm{ADP-CSMC}}(\cdot|\theta^{(p)},\psi^{(p)},s_{0:T}^{(p,1:N)},a_{0:T-1}^{(p,1:N)},k^{(p)},\lambda_{i}),
\end{align}
we can compute a controlled SMC likelihood estimator $\hat{p}(y_{1:T}|\bar{\theta}^{(p)},\lambda_{i})$
using the new particle system $(\bar{s}_{0:T}^{(p,1:N)},\bar{a}_{0:T-1}^{(p,1:N)})$.
In our algorithmic implementation, we only have to store $(\bar{\theta}^{(p)},\bar{s}_{0:T}^{(p)},\hat{p}(y_{1:T}|\bar{\theta}^{(p)},\lambda_{i}))_{p=1}^{P}$, where the trajectory $\bar{s}_{0:T}^{(p)}=(\bar{s}_{t}^{(p,\bar{l}_{t}^{(p)})})_{t=0}^{T}$ is obtained by setting $\bar{l}_{T}^{(p)}=\bar{k}^{(p)}$ and $\bar{l}_{t}^{(p)}=\bar{a}_{t}^{(p,\bar{l}_{t+1}^{(p)})}$ for $t=T-1,\ldots,0$. As before, Equation \eqref{eq:extended_target_pmmh_rewrite} implies
that the samples $(\bar{\theta}^{(p)},\bar{s}_{0:T}^{(p)})_{p=1}^{P}$ will approximate the marginal target distribution
$\pi(dx|\lambda_{i})$ on $\mathbb{X}$.

\end{proof}

\begin{proof}[Proof of Theorem \ref{thm:smc_sq}]
Following Proposition \ref{prop:SMC_samplers},
we note that our estimator $\hat{\pi}(\varphi)$ of the posterior expectation $\pi(\varphi)$ and our estimator $\hat{p}(y_{1:T})$ of the model evidence $p(y_{1:T})$ coincide with those of an SMC sampler when resampling
is performed at every iteration \citep[Equations (13) and (14)]{del2006sequential}.
Hence the WLLNs and CLTs in Equations \eqref{eqn:WLLN} and \eqref{eqn:CLT} follow from
\citet[Proposition 2]{del2006sequential},
\citet[Propositions 9.4.1 and 9.4.2]{del2004feynman} and \citet[Theorem 1]{chopin2004central}.

\end{proof}

\begin{proposition}\label{prop:adaptive_SMC_samplers}
%		\item The adaptive SMC$^2$ in Algorithm \ref{alg:SMC-sq} is an adaptive SMC sampler
%		of \citet{beskos2016convergence} that ...
The adaptive SMC$^2$ in Algorithm \ref{alg:SMC-sq} is an adaptive SMC sampler of \citet{beskos2016convergence} operating on the space $\mathbb{X}$ with
initial distribution $\pi(dx|\lambda_0)$ and
the following properties at iteration $i=1,\ldots,I$:
\begin{enumerate}[label=(\roman*)]
	\item	the target distribution $\pi(dx|\lambda_i)$ defined in
	Equation \eqref{eqn:target_distribution};
	
	\item the potential function $G_{i-1}(x)=\prod_{t=1}^Tg_{\theta}(y_t|s_{t-1},s_t)^{\lambda_i-\lambda_{i-1}}$
	for $x=(\theta,s_{0:T})\in\mathbb{X}$;
	
	\item the Markov transition kernel $m(d\tilde{x}|x,\lambda_i)$ defined in
	Equation \eqref{eqn:marginal_composed_kernel}.
	
\end{enumerate}
\end{proposition}

\begin{proof}
The algorithm is initialized in Step 1 of Algorithm \ref{alg:SMC-sq}
by simulating $P\in\mathbb{N}$ independent samples
$(X^{(p)})_{p=1}^P=(\theta^{(p)},s_{0:T}^{(p)})_{p=1}^{P}$ from $\pi(dx|\lambda_0)$ at inverse temperature $\lambda_{0}=0$.

At iteration $i\geq1$, suppose we have samples $(\theta^{(p)},s_{0:T}^{(p)})_{p=1}^{P}$ approximating
$\pi(dx|\lambda_{i-1})$.
Using these samples, the next inverse temperature $\lambda_i\in(\lambda_{i-1},1]$ is determined in
Step 2 of Algorithm \ref{alg:SMC-sq}. The adaptation criterion in Equation \eqref{eqn:adapt_inversetemp}
is exactly the same as the effective sample size criterion employed in
\citet[Equation (3.1)]{beskos2016convergence} with a pre-specified threshold $\kappa_{\mathrm{ESS}}\in(0,1)$.
Moreover, it follows from \citet[Lemma 1]{beskos2016convergence} that the
ESS criterion in Equation \eqref{eqn:adapt_inversetemp} is a strictly decreasing and continuous function
of $\lambda\in(\lambda_{i-1},1]$, which justifies computing $\lambda_i$ using a bisection method.
The adaptive SMC sampler would then perform resampling with weights that are
proportional to $G_{i-1}(X^{(p)})$ for $p=1,\ldots,P$. This corresponds to Steps 3a and 3b of
Algorithm \ref{alg:SMC-sq} since we have the relation $G_{i-1}(x)=w(\theta,s_{0:T}|\lambda_{i-1},\lambda_i)$
between the potential and the importance weight in Equation \eqref{eqn:adapt_weights}.

The adaptive SMC sampler then moves the resampled particles using an adaptive Markov transition kernel
$m_i$ that leaves the next target distribution $\pi(dx|\lambda_i)$ invariant.
In Algorithm \ref{alg:SMC-sq}, $m_i$ is defined by the composition of Steps 3c, 3d and 4.
In terms of the framework considered here, the proposal transition kernel $h_i$ in Step 4,
whose tuning parameters are determined adaptively as described in Section \ref{subsec:consistency},
corresponds to a specific choice of proposal transition kernel
$h(d\tilde{\theta}|\theta,\lambda_i)$ in Equation \eqref{eqn:proposal_pmmh_extendedspace}.
Hence the resulting adaptive Markov transition kernel $m_i$
corresponds to the Markov transition kernel $m(d\tilde{x}|x,\lambda_i)$ in Equation \eqref{eqn:marginal_composed_kernel}. By part $(iii)$ of Lemma \ref{lem:properties_kernels},
$m_i$ admits $\pi(dx|\lambda_i)$ as invariant distribution.

Lastly, the adaptive SMC sampler terminates after $I=\inf\{i\geq 1: \lambda_i\geq 1\}$ iterations,
which is implemented in Step 5 of Algorithm \ref{alg:SMC-sq}.

\end{proof}

\begin{proof}[Proof of Theorem \ref{thm:adaptive_smc_sq}]
Following Proposition \ref{prop:adaptive_SMC_samplers},
we note that our estimator $\hat{\pi}(\varphi)$ of the posterior expectation $\pi(\varphi)$ and our estimator $\hat{p}(y_{1:T})$ of the model evidence $p(y_{1:T})$ coincide with those of the
adaptive SMC sampler \citep[Equations (3.6) and (3.7)]{beskos2016convergence}.
In our context, Assumptions A5 and A6 of \citet{beskos2016convergence} are satisfied
under Assumption \ref{ass:adaptive_smc}($i$)-($iv$), so we may invoke
Theorems 3.1 and 3.2 of \citet{beskos2016convergence} to establish the WLLNs
in Equation \eqref{eqn:WLLN}.
Similarly, Assumptions A5, A6 and A7 of \citet{beskos2016convergence} hold 
under Assumption \ref{ass:adaptive_smc}, hence we can appeal to Theorem 3.4 of \citet{beskos2016convergence}
to obtain the CLTs in Equation \eqref{eqn:CLT}.

\end{proof}

\section{Dynamic stochastic general equilibrium model}\label{appendix:DSGE}

\subsection{State-space model representation}\label{appendix:dsge_ssm}

We detail how one can rewrite the model specification in Equations \eqref{eqn:dsge_xtransition} and \eqref{eqn:dsge_ztransition}
to obtain the state equations in Equations \eqref{eqn:dsge_state_transition} and \eqref{eqn:dsge_initial_state}.
We define $0_{d_{z}\times d_{x}}$ as the zero matrix of size $d_{z}\times d_{x}$ and
write $L(\theta)s=L_1(\theta)x + L_2(\theta)z$, where $L_1(\theta)\in\mathbb{R}^{d_x\times d_x}$ and
$L_2(\theta)\in\mathbb{R}^{d_x\times d_z}$ are sub-matrices of $L(\theta)=(L_1(\theta)~ L_2(\theta))\in\mathbb{R}^{d_x\times d}$.
It follows that Equations \eqref{eqn:dsge_state_transition} and \eqref{eqn:dsge_initial_state} hold if we define
\begin{equation}
	A(\theta) = \left(\begin{array}{cc}
			L_{1}(\theta) & L_{2}(\theta)\rho(\theta)\\
			0_{d_{z}\times d_{x}} & \rho(\theta)
			\end{array}\right)\in\mathbb{R}^{d\times d},\quad
	B(\theta) = \left(\begin{array}{c}
			L_{2}(\theta)\Sigma(\theta)\\
			\Sigma(\theta)
			\end{array}\right)\in\mathbb{R}^{d\times d_{z}},
\end{equation}
and
\begin{equation}
c_{\theta}(s_{t-1},\varepsilon_t) = \left(\begin{array}{c}
c(\theta)+Q_{\theta}(x_{t-1},z_t)\\
0_{d_{z}}
\end{array}\right)\in\mathbb{R}^{d},
\end{equation}
where $z_t$ is given by Equation \eqref{eqn:dsge_ztransition}.

\subsection{Optimal policy for linear Gaussian DSGE model}\label{appendix:dsge_lqg}

In this section, we consider the linear Gaussian state-space model that arises when
log-linearization approximations are employed to solve for the equilibrium conditions.
We show below that the optimal policy $\psi^*=(\psi_t^*)_{t=0}^T$ defined by Equation \eqref{eqn:optimal_psi}
has the form
\begin{equation}\label{eqn:dsge_optimalpolicy}
	\psi_t^*(s_t) = \exp(-Q_0(s_t;\tilde{\beta}_t)),\quad t=0,\ldots,T,
\end{equation}
and derive a backward recursion for its coefficients $(\tilde{\beta}_t)_{t=0}^T=(\tilde{A}_t,\tilde{b}_t,\tilde{c}_t)_{t=0}^T$.
Under the linear state transitions in Equations \eqref{eqn:dsge_state_transition} and \eqref{eqn:dsge_initial_state} (without the $c_{\theta}$ terms),
it follows that the optimal policy can be written as Equation \eqref{eqn:dsge_quadratic_functionclass} with coefficients $\beta_0^*=(A_0^*,b_0^*,c_0^*)$ and $(\beta_t^*)_{t=1}^T=(A_t^*,b_t^*,C_t^*,D_t^*,e_t^*,f_t^*)_{t=1}^T$ given by
\begin{equation}
	A_0^*=B(\theta)^\top\tilde{A}_0B(\theta),\quad b_0^*=B(\theta)^\top \tilde{b}_0,\quad c_0^*=\tilde{c}_0,
\end{equation}
and
\begin{eqnarray}
	A_t^* &=& B(\theta)^\top \tilde{A}_t B(\theta), \quad b_t^* = B(\theta)^\top \tilde{b}_t, \quad C_t^* = 2B(\theta)^\top\tilde{A}_tA(\theta), \\
	D_t^*  &= & A(\theta)^\top \tilde{A}_t A(\theta), \quad e_t^* = A(\theta)^\top\tilde{b}_t, \quad f_t^* = \tilde{c}_t.\notag
\end{eqnarray}
for $t=1,\ldots,T$ (see Appendix \ref{appendix:dsge_dimreduction}).

The weight functions of uncontrolled SMC can be written as
\begin{equation}
\begin{aligned}
	&w_0(s_0;\theta,\lambda) = 1 = \exp\left(-\lambda Q_0(s_0;\bar{\beta}_0)\right),\\
	&w_t(s_{t-1},s_t;\theta,\lambda)=g_{\theta}(y_t|s_t)^{\lambda}=\exp\left(-\lambda Q_0(s_t;\bar{\beta}_t)\right),\quad t=1,\ldots,T,
\end{aligned}
\end{equation}
where the coefficients
$\bar{\beta}_t=(\bar{A}_t,\bar{b}_t,\bar{c}_t)\in\mathbb{R}^{d\times d}_{\mathrm{sym}}\times\mathbb{R}^d\times\mathbb{R}$
are given by $(\bar{A}_0,\bar{b}_0,\bar{c}_0)=(0_{d\times d},0_d,0)$ for $t=0$, and
\begin{equation}
\begin{aligned}
	&\bar{A}_t = \frac{1}{2}E(\theta)^\top F(\theta)^{-1}E(\theta),\quad \bar{b}_t = - E(\theta)^\top F(\theta)^{-1}(y_t-d(\theta)),\quad\\
	&\bar{c}_t = \frac{1}{2}	(y_t-d(\theta))^\top F(\theta)^{-1}(y_t-d(\theta)) + \frac{1}{2}d_y\log(2\pi) + \frac{1}{2}\log\det(F(\theta)),
\end{aligned}
\end{equation}
for $t=1,\ldots,T$. At the terminal time $T$, we have $\tilde{\beta}_T=\lambda\bar{\beta}_T$ since $\psi_T^*=w_T$.
We then proceed inductively for $t=T-1,\ldots,0$. Assume that $\psi_{t+1}^*$ has the form as in Equation \eqref{eqn:dsge_optimalpolicy} with coefficients $\tilde{\beta}_{t+1}$. By completing the square, the conditional expectation
\begin{eqnarray}
	q_{t+1}(\psi_{t+1}^*|s_t,\theta) &= &\exp\left(-s_t^\top\{A(\theta)^\top \tilde{A}_{t+1}A(\theta) - 2J_{t+1}^\top K_{t+1} J_{t+1}\}s_t\right)\notag\\
	&\times&\exp\left(-s_t^\top \{A(\theta)^\top\tilde{b}_{t+1} -2 J_{t+1}^\top K_{t+1}B(\theta)^\top\tilde{b}_{t+1} \} \right)\\
	&\times&\exp \left(-\Big\{\tilde{c}_{t+1} - \frac{1}{2}\tilde{b}_{t+1}^\top B(\theta)K_{t+1}B(\theta)^\top\tilde{b}_{t+1}
	- \frac{1}{2}\log\det(K_{t+1})\Big\}\right),\notag
\end{eqnarray}
where $J_{t+1} = B(\theta)^\top \tilde{A}_{t+1}A(\theta)$ and $K_{t+1} = (I_{d_z} + 2B(\theta)^\top\tilde{A}_{t+1}B(\theta))^{-1}$.
It follows that the next iterate $\psi_t^*=w_tq_{t+1}(\psi_{t+1}^*)$ also has the form as in Equation \eqref{eqn:dsge_optimalpolicy} with coefficients
\begin{eqnarray}
	\tilde{A}_t &= & \lambda\bar{A}_t + A(\theta)^\top \tilde{A}_{t+1}A(\theta) - 2J_{t+1}^\top K_{t+1} J_{t+1},\notag\\
	\tilde{b}_t &= & \lambda\bar{b}_t + A(\theta)^\top\tilde{b}_{t+1} -2 J_{t+1}^\top K_{t+1}B(\theta)^\top\tilde{b}_{t+1},\\
	\tilde{c}_t &= & \lambda\bar{c}_t + \tilde{c}_{t+1} - \frac{1}{2}\tilde{b}_{t+1}^\top B(\theta)K_{t+1}B(\theta)^\top\tilde{b}_{t+1}
	- \frac{1}{2}\log\det(K_{t+1}).\notag
\end{eqnarray}

\subsection{Dimension reduction for policy learning}\label{appendix:dsge_dimreduction}

We derive the form of the mapping $\beta_t=\Lambda_{\theta}(\tilde{\beta}_t)$ in Equation \eqref{eqn:dsge_quadratic_functionclass} by equating coefficients in the relationship $Q_0(\tilde{s}_t;\tilde{\beta}_t)=Q(s_{t-1},\varepsilon_t;\beta_t)$, where $\tilde{\beta}_t=(\tilde{A}_t,\tilde{b}_t,\tilde{c}_t)\in\mathbb{R}_\mathrm{sym}^{d\times d}\times\mathbb{R}^d\times\mathbb{R}$ are the coefficients in the dimension reduced space and $\beta_t=(A_t,b_t,C_t,D_t,e_t,f_t)$ are the desired coefficients. By substituting the linearized state $\tilde{s}_t = A(\theta)s_{t-1} + B(\theta)\varepsilon_t$ and expanding terms, we have
\begin{align}
	&Q_0(\tilde{s}_t;\tilde{\beta}_t) = \tilde{s}_t^\top \tilde{A}_t \tilde{s}_t + \tilde{s}_t^\top \tilde{b}_t + \tilde{c}_t \notag\\
	&= (A(\theta)s_{t-1} + B(\theta)\varepsilon_t)^\top \tilde{A}_t (A(\theta)s_{t-1} + B(\theta)\varepsilon_t)
	+ (A(\theta)s_{t-1} + B(\theta)\varepsilon_t)^\top\tilde{b}_t + \tilde{c}_t\notag\\
	&= \varepsilon_t^\top A_t \varepsilon_t + \varepsilon_t^\top b_t + \varepsilon_t^\top C_t s_{t-1} +
	s_{t-1}^\top D_t s_{t-1} + s_{t-1}^\top e_t + f_t\notag\\
	&= Q(s_{t-1},\varepsilon_t;\beta_t),
\end{align}
where we set
\begin{equation}
\begin{aligned}
	&A_t = B(\theta)^\top \tilde{A}_t B(\theta), \quad b_t = B(\theta)^\top \tilde{b}_t, \quad C_t = 2B(\theta)^\top\tilde{A}_tA(\theta), \\
	 &D_t = A(\theta)^\top \tilde{A}_t A(\theta), \quad e_t = A(\theta)^\top\tilde{b}_t, \quad f_t = \tilde{c}_t.
\end{aligned}
\end{equation}

\subsection{Proposal transitions and weight functions of controlled SMC}\label{appendix:dsge_csmc}

For any matrix $A\in\mathbb{R}_\mathrm{sym}^{d\times d}$, we will write $A\succ 0$ if $A$ is positive definite.
Suppose we have a policy $\psi=(\psi_t)_{t=0}^T$ of the form in Equation \eqref{eqn:dsge_quadratic_functionclass} with coefficients $\beta_0=(A_0,b_0,c_0)$ and $(\beta_t)_{t=1}^T=(A_t,b_t,C_t,D_t,e_t,f_t)_{t=1}^T$. We shall first assume that the constraint $I_{d_z}+2A_t\succ 0$ is satisfied for all $t=0,\ldots,T$, and defer our discussion on how to impose these constraints to Appendix \ref{appendix:dsge_constraints}. At the initial time, we have the new initial distribution
\begin{equation}
	q_0^{\psi}(ds_0|\theta) = \delta_{\Phi_{\theta}^{(0)}(\varepsilon_0)}(ds_0)\mathcal{N}(\varepsilon_0;-K_0b_0,K_0)d\varepsilon_0,
\end{equation}
and the expectation
\begin{equation}
	q_0(\psi_0|\theta) = \det(K_0)^{1/2}\exp\left(\frac{1}{2}b_0^\top K_0b_0 - c_0\right),
\end{equation}
where $K_0=(I_{d_z} + 2A_0)^{-1}$. For time $t=1,\ldots,T$, the new proposal transition is
\begin{equation}
	q_t^{\psi}(ds_t|s_{t-1},\theta) = \delta_{\Phi_{\theta}(s_{t-1},\varepsilon_{t})}(ds_t)\mathcal{N}(\varepsilon_t;-K_t(b_t + C_ts_{t-1}),K_t)d\varepsilon_t,
\end{equation}
and the conditional expectation is 
\begin{equation}\label{eqn:dsge_condexp}
	q_t(\psi_t|s_{t-1},\theta) = \det(K_t)^{1/2}\exp\left(\frac{1}{2}(b_t + C_ts_{t-1})^\top K_t (b_t + C_ts_{t-1}) - \check{Q}(s_{t-1};\check{\beta}_t)\right),
\end{equation}
where $K_t = (I_{d_z} + 2A_t)^{-1}$ and $\check{Q}(z;\check{\beta})=z^{\top} D z + z^{\top}e + f$ is a quadratic function that depends on the coefficients $\check{\beta}=(D,e,f)\in\mathbb{R}_\mathrm{sym}^{d\times d}\times\mathbb{R}^d\times\mathbb{R}$.

\subsection{Policy learning under constraints}\label{appendix:dsge_constraints}

Let $\psi=(\psi_t)_{t=0}^T$ denote a current policy with coefficients $\beta_0=(A_0,b_0,c_0)$ and $(\beta_t)_{t=1}^T=(A_t,b_t,C_t,D_t,e_t,f_t)_{t=1}^T$, and $\phi=(\phi_t)_{t=0}^T$ denote its refinement with coefficients $\tilde{\beta}_0=(\tilde{A}_0,\tilde{b}_0,\tilde{c}_0)$ and $(\tilde{\beta}_t)_{t=1}^T=(\tilde{A}_t,\tilde{b}_t,\tilde{C}_t, \tilde{D}_t,\tilde{e}_t,\tilde{f}_t)_{t=1}^T$. The refined policy $\psi\cdot\phi=(\psi_t\cdot\phi_t)_{t=0}^T$ has to satisfy the constraints
\begin{equation}
	I_{d_z} + 2(A_t + \tilde{A}_t)\succ 0,\quad t=0,\ldots,T,
\end{equation}
to ensure that the covariance matrices appearing in Appendix \ref{appendix:dsge_csmc} are well-defined.
For any $t=0,\ldots,T$ and $\alpha\in[0,1/2)$ (we set $\alpha=0.4$ in our implementation), we decompose
\begin{equation}
	I_{d_z} + 2(A_t + \tilde{A}_t) = R + 2U_t,
\end{equation}
into a remainder term $R=(1-2\alpha)I_{d_z}$ and an update term $U_t=\alpha I_{d_z} + A_t + \tilde{A}_t$ that we will constrain to be positive definite. By an inductive argument, we will assume that $\alpha I_{d_z} + A_t$ is positive definite, which holds when we initialize $A_t$ as a zero matrix.

As the matrix $\tilde{A}_t$, estimated from the policy refinement, could be negative definite, we will introduce a learning rate $\kappa_t\in[0,1]$ in the update
\begin{equation}
	U_t(\kappa_t)=\alpha I_{d_z} + A_t + \kappa_t\tilde{A}_t,
\end{equation}
and seek the largest $\kappa_t$ such that $U_t(\kappa_t)$ is positive definite. First, we compute the matrix square root of $\alpha I_{d_z} + A_t$, i.e. a unique positive definite matrix $M_t\in\mathbb{R}_\mathrm{sym}^{d_z\times d_z}$ satisfying $\alpha I_{d_z} + A_t=M_t M_t$. The determinant of $U_t(\kappa_t)$ can be written as
\begin{equation}
	\det\left(U_t(\kappa_t)\right) = \det(M_t)^2\det\left(I_{d_z} + \kappa_t M_t^{-1}\tilde{A}_t M_t^{-1} \right).
\end{equation}
The matrix $M_t^{-1}\tilde{A}_t M_t^{-1}$ is symmetric as both $M_t$ and $\tilde{A}_t$ are symmetric, therefore its eigenvalues $\Lambda_1,\ldots,\Lambda_{d_z}$ are real. Hence we have
\begin{equation}
	\det\left(U_t(\kappa_t)\right) = \det(M_t)^2\prod_{i=1}^{d_z}(1+\kappa_t\Lambda_i) >0
\end{equation}
if and only if $1+\kappa_t\Lambda_{\min}>0$, where $\Lambda_{\min}=\min_{i=1,\ldots,d_z}\Lambda_i$ denotes the minimum eigenvalue. We set the learning rate as
\begin{equation}
	\kappa_t = \begin{cases}	
	1,\quad&\mbox{if }1+\kappa_t\Lambda_{\min}>0,\\
	\min(1,(\zeta-1)/\Lambda_{\min}),\quad&\mbox{if }1+\kappa_t\Lambda_{\min}\leq 0,
	\end{cases}
\end{equation}
where $\zeta>0$ is a small number (taken as $\zeta=2^{-52}$ in our implementation). After determining the learning rate, we update all coefficients using $\beta_t+\kappa_t\tilde{\beta}_t$.

\section{Consumption-based long-run risk model}\label{LRRsmc2}

\subsection{Model setup}

\noindent {\emph{Preferences.}} We consider an endowment economy with a representative agent who has recursive preferences as in  \citet{Epstein_Zin_1989} and \citet{weil1989equity}. The agent maximizes her lifetime utility, which is given recursively by
\begin{equation}\label{EqEpsteinZin}
\mathsf{V}_{t} = \left[ (1-\delta)\mathsf{C}_t^{\frac{1-\gamma}{\theta_v}} + \delta \left[\mathbb{E}_{t}(\mathsf{V}_{t+1}^{1-\gamma})\right]^{\frac{1}{\theta_v}}\right]^{\frac{\theta_v}{1-\gamma}},
\end{equation}
where $\mathsf{C}_t$ denotes consumption at time $t$, $\delta\in(0,1)$ is the agent's time preference parameter, $\gamma$ is the relative risk aversion parameter, $\psi$ is the elasticity of intertemporal substitution (EIS), $\theta_v = (1-\gamma)/(1-1/\psi)$, and $\mathbb{E}_t$ denotes conditional expectation with respect to information up to time $t$. This class of preferences allows for a separation between risk aversion and the EIS. When $\gamma > 1/\psi$, the agent prefers early resolution of uncertainty; when $\gamma < 1/\psi$, she prefers late resolution of uncertainty; and when $\theta_v = 1$, she has the standard constant relative risk aversion preferences and is neutral to the time of resolution of uncertainty. The agent's utility maximization is subject to the following budget constraint
\begin{equation*}
      \mathsf{W}_{t+1} = (\mathsf{W}_t - \mathsf{C}_t)\mathsf{R}_{w, t+1},
\end{equation*}
where $\mathsf{W}_t$ is the wealth of the agent, and $\mathsf{R}_{w,t}$ is the return on the wealth portfolio.

For any asset $i$ with ex-dividend price $\mathsf{P}_{i, t}$ and dividend $\mathsf{D}_{i,t}$, the standard Euler equation holds, i.e.
\begin{equation}\label{EqEuler}
       \mathbb{E}_{t}\left[\mathsf{M}_{t+1}\mathsf{R}_{i, t+1}\right] = 1,
\end{equation}
where $\mathsf{R}_{i,t+1} = (\mathsf{P}_{i, t+1} + \mathsf{D}_{i, t+1})/\mathsf{P}_{i, t}$, and $\mathsf{M}_t$ is the stochastic discount factor. In particular, for the risk-free asset, we have $\mathsf{R}_{f,t} = 1/\mathbb{E}_{t}[\mathsf{M}_{t+1}]$. For the recursive utility function defined in Equation \eqref{EqEpsteinZin}, the stochastic discount factor is given by
\begin{equation}\label{Eqsdf}
\mathsf{M}_{t+1} =\delta \left( \frac{\mathsf{C}_{t+1}}{\mathsf{C}_{t}}\right) ^{-\frac{1}{\psi }}\left( \frac{\mathsf{V}_{t+1}}{\left[ \mathbb{E}_{t}\left( \mathsf{V}_{t+1}^{1-\gamma
}\right) \right] ^{\frac{1}{1-\gamma }}}\right) ^{\frac{1}{\psi }-\gamma }.
\end{equation}
\citet{Epstein_Zin_1989} showed that the wealth-consumption ratio $\mathsf{W}_t/\mathsf{C}_t$ can be expressed in terms of the value function $\mathsf{V}_t$,
\begin{equation}\label{wcratio}
\frac{\mathsf{W}_{t}}{\mathsf{C}_{t}}=\frac{1}{1-\delta}\left(\frac{\mathsf{V}_{t}}{\mathsf{C}_{t}}\right) ^{1-1/\psi},
\end{equation}
which allows us to reformulate the stochastic discount factor given in Equation \eqref{Eqsdf} using the return on the wealth portfolio as follows
\begin{equation}\label{Eqsdf2}
\mathsf{M}_{t+1}=\delta^{\theta_v} \left( \frac{\mathsf{C}_{t+1}}{\mathsf{C}_{t}}\right)^{-\frac{\theta_v }{\psi }}\mathsf{R}_{w,t+1}^{\theta_v -1}.
\end{equation}
Therefore the Euler equation \eqref{EqEuler} implies that the return on the wealth portfolio $\mathsf{R}_{w,t}$ satisfies
\begin{equation}\label{EqWealthEuler}
       \mathbb{E}_{t}\left[\delta^{\theta_v} \left( \frac{\mathsf{C}_{t+1}}{\mathsf{C}_{t}}\right)^{-\frac{\theta_v }{\psi }}\mathsf{R}_{w,t+1}^{\theta_v}\right] = 1.
\end{equation}
\bigskip

\noindent {\emph{Fundamentals.}} Following \citet{bansal2004risks} and \citet{bansal_kiku_yaron_2012}, we assume that the log-consumption growth $\Delta \mathsf{c}_{t} = \log\left(\mathsf{C}_{t}/\mathsf{C}_{t-1}\right)$ consists of a persistent component $x_t$ and a transitory component,
\begin{align}\label{eqn:lrr_deltac_x}
    \Delta \mathsf{c}_{t} &= \mu + x_{t-1} + \sigma_{t-1}u_{c,t},\\
               x_{t} &= \rho x_{t-1} + \phi_x \sigma_{t-1} \varepsilon_{x, t},\notag
\end{align}
and that dividends are imperfectly correlated with consumption and their log-growth rate $\Delta \mathsf{d}_{t} = \log\left(\mathsf{D}_{t}/\mathsf{D}_{t-1}\right)$ has the following dynamics
\begin{equation}\label{eqn:lrr_dd}
	\Delta \mathsf{d}_{t} = \mu_d + \Phi x_{t-1} + \phi_{dc}\sigma_{t-1}u_{c,t} + \phi_{d}\sigma_{t-1}u_{d,t},
\end{equation}
where $u_{c,t},\varepsilon_{x,t},u_{d,t}\sim\mathcal{N}(0,1)$ independently and $\sigma^2_{t}$ is the conditional variance of consumption growth.
In Equations \eqref{eqn:lrr_deltac_x} and \eqref{eqn:lrr_dd}, we consider the parameter values
$\mu,\mu_d,\Phi,\phi_{dc}\in\mathbb{R}$, $\rho\in(-1,1)$ and $\phi_{x},\phi_{d}\in(0,\infty)$.

In the standard long-run risk model \citep{bansal2004risks,bansal_kiku_yaron_2012}, the conditional variance $\sigma^2_{t}$ is modelled as an autoregressive (AR) process which can take negative values.
To resolve this issue, we adopt the approach in \citet{fulop2020bayesian} that is based on the (non-negative) autoregressive gamma (ARG) process of \citet{gourieroux2006autoregressive}. The ARG transition for the conditional variance $\sigma^2_{t}$ is given by
\begin{equation}\label{argvol}
    \sigma_t^2 \sim \mathcal{G}(\phi_s + \varepsilon_{s,t}, c), \quad \varepsilon_{s,t} \sim \mathcal{P}(\nu \sigma^2_{t-1}/c),
\end{equation}
where $\mathcal{G}(a,b)$ denotes the gamma distribution with shape $a\in(0,\infty)$ and scale $b\in(0,\infty)$,
and $\mathcal{P}(r)$ denotes the Poisson distribution with rate $r\in(0,\infty)$.
In Equation \eqref{argvol}, $\nu\in(0,1)$ controls the persistence, $c\in(0,\infty)$ determines the scale, and the Feller condition $\phi_s\in(1,\infty)$ should be imposed to ensure positivity of the conditional variances.
By marginalizing over $\varepsilon_{s,t}$, it can be shown that the Markov transition of $\sigma^2_{t}$ is given by a non-central gamma distribution with
conditional mean $m_{\theta}^{(f)}(\sigma_{t-1}^2)=\phi_s c + \nu\sigma_{t-1}^2$
and conditional variance $v_{\theta}^{(f)}(\sigma_{t-1}^2)=\phi_s c^2 + 2c\nu\sigma_{t-1}^2$ \citep{gourieroux2006autoregressive,creal2017class}.
Moreover, the stationary distribution of the ARG process is $\mathcal{G}(\phi_s, c/(1-\nu))$ with the long-run mean $\bar{\sigma}^2 = \phi_s c/(1-\nu)$. \bigskip

\noindent {\emph{Model solution.}} \citet{bansal2004risks} first applied the log-linear approximation method of \citet{campbell1988stock} to solve the long-run risk model. Log-linearization has been widely employed to solve models with long-run risks \citep*[see, e.g.,][]{bansal_kiku_yaron_2012,bansal2016risks,beeler_campbell_2012,Schorfheide_Song_Yaron_2018}.
However, in a recent work, \citet{pohl2018higher} argued that log-linearization of models with long-run risks may generate large numerical errors when state variables are persistent. Hence we will solve the model using the collocation projection method \citep{judd1992projection} that can account for higher-order effects.

\subsection{Implementation details}

\noindent \emph{State-space model.} The $d=2$ latent state variables $s_{t}=(x_{t},\sigma_{t}^{2})\in\mathbb{S}=\mathbb{R}\times(0,\infty)$ are the persistent component of log-consumption growth $x_t$ and its conditional variance $\sigma_{t}^2$.
The time evolution of $(x_{t})_{t=0}^T$ is given by the AR process in Equation \eqref{eqn:lrr_deltac_x} with initialization at $x_0=0$.
We will write the associated Markov transition kernel as $f_{\theta}(dx_t|s_{t-1})=\mathcal{N}(x_t;\rho x_{t-1}, \phi_{x}^2\sigma_{t-1}^2)dx_t$.
The consumption volatility $(\sigma_{t}^2)_{t=0}^T$ follows an ARG process, whose Markov transition kernel is given by a non-central gamma distribution
\begin{align}\label{eqn:noncentral_gamma}
	f_{\theta}(d\sigma_t^2|\sigma_{t-1}^2) = \left(\frac{\sigma_{t}^{2}}{\nu\sigma_{t-1}^{2}}\right)^{(\phi_{s}-1)/2}\frac{1}{c}\exp\left(-\frac{(\sigma_{t}^{2}+\nu\sigma_{t-1}^{2})}{c}\right)I_{\phi_{s}-1}\left(\frac{2\sqrt{\nu\sigma_{t-1}^{2}\sigma_{t}^{2}}}{c}\right)d\sigma_t^2,
\end{align}
where $I_{\zeta}(x)=(x/2)^{\zeta}\sum_{i=0}^{\infty}(x^{2}/4)^{i}/\{i!\Gamma(\zeta+i+1)\}$ denotes a modified Bessel function of the first kind and $\Gamma$ is the gamma function. 
We will initialize $\sigma_0^2$ at the long-run mean $\bar{\sigma}^{2}=\phi_{s}c/(1-\nu)$. Hence the initial distribution $\mu_{\theta}(ds_0)=\delta_{s_0}(ds_0)$ is given by a Dirac measure at $s_{0}=(x_{0},\sigma_{0}^{2})=(0,\bar{\sigma}^{2})$ and the Markov transition kernel has the form
$f_{\theta}(ds_t | s_{t-1})=f_{\theta}(dx_t|s_{t-1})f_{\theta}(d\sigma_t^2|\sigma_{t-1}^2)$ with support on $\mathbb{S}$.

The $d_y=4$ observed variables $y_{t}=(\Delta \mathsf{c}_{t},\Delta \mathsf{d}_{t},\mathsf{m}_{t},\mathsf{r}_{t})\in\mathbb{Y}=\mathbb{R}^{4}$
are log-consumption growth $\Delta \mathsf{c}_{t}$, log-dividend growth $\Delta \mathsf{d}_{t}$, market return $\mathsf{m}_{t}$, and risk-free rate $\mathsf{r}_{t}$. The observation model for each $t=1,\ldots,T$ is
\begin{equation}
\begin{aligned}
	g_{\theta}(y_t |s_{t-1},s_t) &=
	\mathcal{N}(\Delta \mathsf{c}_{t} ; \mu + x_{t-1}, \sigma_{t-1}^2)\\
	&\times\mathcal{N}(\Delta \mathsf{d}_{t} ; \mu_{d} + \Phi x_{t-1} + \phi_{dc}(\Delta \mathsf{c}_{t}-\mu-x_{t-1}), \phi_{d}^2\sigma_{t-1}^2)\\
	&\times\mathcal{N}(\mathsf{m}_{t} ; \mathsf{M}_{\theta}(x_{t-1},\sigma_{t-1}^{2},x_{t},\sigma_{t}^{2},\Delta \mathsf{d}_{t}),\phi_{m}^2)\\
	&\times\mathcal{N}(\mathsf{r}_{t} ; \mathsf{R}_{\theta}(x_{t},\sigma_{t}^{2}),\phi_{r}^2),
\end{aligned}
\end{equation}
where the market return and risk-free rate are assumed to be measured with normally distributed errors, and $\phi_{m},\phi_{r}\in(0,\infty)$ denote the respective standard deviations of the measurement errors. The functions $\mathsf{M}_{\theta}:\mathbb{S}\times\mathbb{S}\times\mathbb{R}\rightarrow\mathbb{R}$ and $\mathsf{R}_{\theta}:\mathbb{S}\rightarrow\mathbb{R}$, which represent the market return and risk-free rate determined by the long-run risk model,
are linear when employing the log-linearization method, and highly non-linear under the collocation projection method.
They also depend on the preference parameters $(\delta,\gamma,\psi)$ introduced earlier.
Hence the $d_{\theta}=15$ unknown model parameters $\theta$ to be inferred are $\theta=(\delta,\gamma,\psi,\mu,\rho,\phi_{x},\bar{\sigma},\nu,\phi_s,\mu_{d},\Phi,\phi_{dc},\phi_{d},\phi_{m},\phi_{r})\in\Theta$.  As parameters in $\Theta$ do not necessarily satisfy the conditions of \citet{borovivcka2020necessary} that characterize existence and uniqueness of solutions to the long-run risk model, we shall impose these conditions as restrictions on our parameter space. \bigskip

\noindent \emph{Annealed controlled SMC.} We now discuss how to implement AC-SMC for the long-run risk model. We first specify the proposals $(q_t)_{t=0}^T$ and weight functions $(w_t)_{t=0}^T$ of the uncontrolled SMC.
As the initialization is deterministically given by $s_{0}=(x_{0},\sigma_{0}^{2})=(0,\bar{\sigma}^{2})$,
we set $q_0(ds_0|\theta)=\delta_{s_0}(ds_0)$.
Subsequently for $t=1,\ldots,T$, we adopt the AR transition in Equation \eqref{eqn:lrr_deltac_x} as our proposal transition for $x_t$,
i.e. set $q_t(dx_t|s_{t-1},\theta)=\mathcal{N}(x_t;\rho x_{t-1}, \phi_{x}^2\sigma_{t-1}^2)dx_t$.
Since it is difficult to work with the non-central gamma transition in Equation \eqref{eqn:noncentral_gamma} associated
with the ARG volatility process, we will approximate it with a log-normal transition in our proposal construction.
By matching the conditional mean $m_{\theta}^{(f)}(\sigma_{t-1}^2)$
and conditional variance $v_{\theta}^{(f)}(\sigma_{t-1}^2)$ of the ARG transition in Equation \eqref{argvol},
we define our proposal transition for $\sigma_t^2$ as
$q_t(d\sigma_t^2|\sigma_{t-1}^2,\theta)=\mathcal{LN}(\sigma_t^2;m_{\theta}^{(q)}(\sigma_{t-1}^2),
v_{\theta}^{(q)}(\sigma_{t-1}^2))d\sigma_t^2$, which denotes a log-normal distribution with parameters
\begin{equation}
	m_{\theta}^{(q)}(\sigma_{t-1}^2) = \log m_{\theta}^{(f)}(\sigma_{t-1}^2) - v_{\theta}^{(q)}(\sigma_{t-1}^2)/2,\quad
	v_{\theta}^{(q)}(\sigma_{t-1}^2) = \log\left( \frac{v_{\theta}^{(f)}(\sigma_{t-1}^2)}
	{m_{\theta}^{(f)}(\sigma_{t-1}^2)^2}+1\right).
\end{equation}
The proposal transition kernel for the state $s_{t}=(x_{t},\sigma_{t}^{2})$ is given by
$q_t(ds_t|s_{t-1},\theta)= q_t(dx_t|s_{t-1},\theta) q_t(d\sigma_t^2|\sigma_{t-1}^2,\theta)$.
Note that our choice of proposal is not temperature dependent.
For any inverse temperature $\lambda\in[0,1]$, we specify the weight functions as
\begin{equation}
	w_0(s_0;\theta,\lambda) = 1, \quad w_t(s_{t-1}, s_t;\theta,\lambda) = \frac{f_{\theta}(s_t|s_{t-1})
	g_{\theta}(y_t|s_{t-1},s_t)^{\lambda}}{q_t(s_t|s_{t-1},\theta)}, \quad t=1,\ldots,T,
\end{equation}
which satisfy the requirement in Equation \eqref{eqn:weights_requirement}. 

To learn policies, we employ the following function class
\begin{equation}\label{eqn:lrr_quadratic_functionclass}
	\mathbb{F}_t = \left\lbrace \psi_t(s_{t-1},s_t)=\exp(-Q(\ell(s_{t-1}),\ell(s_t);\beta_t)) : |\beta_t|^2\leq \xi\right\rbrace,\quad t=1,\ldots,T,
\end{equation}
where $\ell:\mathbb{S}\rightarrow\mathbb{R}\times\mathbb{R}$ is defined by a log-transformation
of the volatility $\ell(x,\sigma^2)=(x,\log\sigma^2)$, and $Q$ is the quadratic function
introduced in Equation \eqref{eqn:quadratic_function} with coefficients $\beta_t=(A_t,b_t,C_t,D_t,e_t,f_t)\in\mathbb{R}^{d\times d}_{\mathrm{sym}}\times\mathbb{R}^{d}\times \mathbb{R}^{d\times d}\times \mathbb{R}_\mathrm{sym}^{d\times d}\times\mathbb{R}^d\times\mathbb{R}$. 
Note that we do not have to learn the policy at time $t=0$ due to the deterministic initialization.
Policy initialization in AC-SMC amounts to setting all coefficients $(\beta_t)_{t=1}^T$ as zero.
We will fit functions from Equation \eqref{eqn:lrr_quadratic_functionclass} using ridge regression
with a pre-specified and suitably large shrinkage $\xi\in(0,\infty)$.

Given a fitted policy $\psi=(\psi_t)_{t=1}^T$ with estimated coefficients $(\beta_t)_{t=1}^T$,
the induced proposal transitions $(q_t^{\psi})_{t=1}^T$ can be sampled exactly, and the
conditional expectations appearing in the weight functions $(w_t^{\psi})_{t=1}^T$ can be evaluated analytically;
see, Appendix \ref{appendix:lrr_csmc} for precise expressions.
If $\phi=(\phi_t)_{t=1}^T$ is an approximation of the optimal refinement of $\psi$ (at a desired inverse temperature)
of the form in Equation \eqref{eqn:lrr_quadratic_functionclass} with estimated coefficients $(\tilde{\beta}_t)_{t=1}^T$,
then the refined policy $\psi\cdot\phi$ also has the same form with coefficients given by the update $\beta_t+\tilde{\beta}_t$ for all $t=1,\ldots,T$.

\subsection{Proposal transitions and weight functions of controlled SMC}\label{appendix:lrr_csmc}
Let $\eta_{\theta}(s_{t-1})=(\rho x_{t-1}, m_{\theta}^{(q)}(\sigma_{t-1}^2))$ and
$\Sigma_{\theta}(s_{t-1})=\mathrm{diag}(\phi_{x}^2\sigma_{t-1}^2, v_{\theta}^{(q)}(\sigma_{t-1}^2))$ denote
the mean vector and covariance matrix of the proposal transition in the transformed variable $z_t=\ell(s_t)$.
Suppose we have a policy $\psi=(\psi_t)_{t=1}^T$ of the form in Equation \eqref{eqn:lrr_quadratic_functionclass}
with coefficients $(\beta_t)_{t=1}^T=(A_t,b_t,C_t,D_t,e_t,f_t)_{t=1}^T$.
For each $t=1,\ldots,T$, assuming that $\Sigma_{\theta}(s_{t-1})^{-1}+2A_t\succ 0$ for all $s_{t-1}\in\mathbb{S}$\footnote{From our implementation, we find that these constraints are usually satisfied.
Otherwise, we impose them by modifying the coefficients update to $\beta_t+\kappa_t\tilde{\beta}_t$,
for some suitably small learning rate $\kappa_t\in(0,1)$.},
the new proposal transition can be written as
\begin{equation}
	q_t^{\psi}(ds_t|s_{t-1},\theta) = \mathcal{N}\left(z_t; K_t(s_{t-1};\theta)m_t(s_{t-1};\theta),
	K_t(s_{t-1};\theta)\right)dz_t\, \delta_{\ell^{-1}(z_t)}(ds_t),
\end{equation}
where $m_t(s_{t-1};\theta)=\Sigma_{\theta}(s_{t-1})^{-1}\eta_{\theta}(s_{t-1})-b_t-C_t\ell(s_{t-1})$,
$K_t(s_{t-1};\theta)=(\Sigma_{\theta}(s_{t-1})^{-1}+2A_t)^{-1}$ and $\ell^{-1}:\mathbb{R}\times\mathbb{R}\rightarrow\mathbb{S}$
denotes the inverse transformation of $\ell$.
Conditional expectations can be evaluated using
\begin{align}
	&q_t(\psi_t|s_{t-1},\theta) =\det\left(K_t(s_{t-1};\theta)\right)^{1/2}\det\left(\Sigma_{\theta}(s_{t-1})\right)^{-1/2}\notag\\
	&\times \exp\left(\frac{1}{2}(m_t^{\top}K_tm_t)(s_{t-1};\theta)-
	\frac{1}{2}( \eta_{\theta}^{\top} \Sigma_{\theta}^{-1} \eta_{\theta}) (s_{t-1}) -\check{Q}(\ell(s_{t-1});\check{\beta}_t) \right),
\end{align}
where $\check{Q}(z;\check{\beta})=z^{\top} D z + z^{\top}e + f$ is a quadratic function that depends on the coefficients
$\check{\beta}=(D,e,f)\in\mathbb{R}_\mathrm{sym}^{d\times d}\times\mathbb{R}^d\times\mathbb{R}$.

\subsection{Real data application}

We consider the U.S. real quarterly data on consumption, market dividends, market returns, and risk-free rates, ranging from 1947:Q2 to 2019:Q3. We refer readers to \citet{fulop2020bayesian} for details on how the data are constructed; Figure \ref{fig:datalrr} displays these time series.
\begin{figure}[!t]
\begin{center}
\includegraphics[width=1\textwidth, height=12cm]{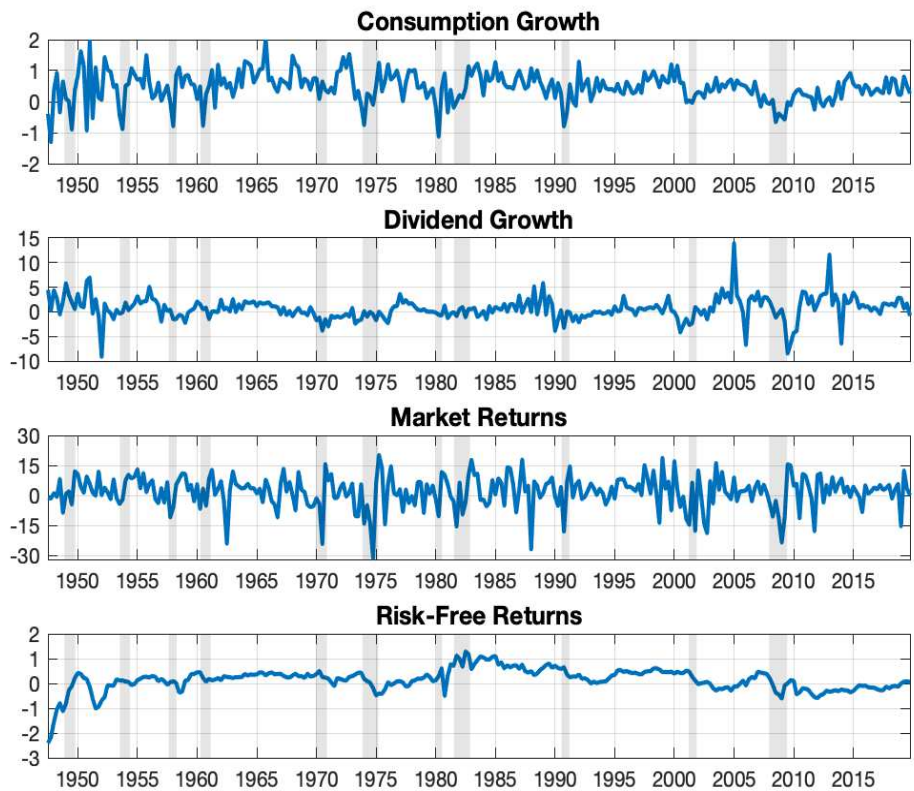}
\vspace{-1cm}
\caption{\textbf{\protect\small {Time Series Data}}}\label{fig:datalrr}
\end{center}
{\footnotesize\emph{Note:} Time series data for model estimation include consumption growth rates, dividend growth rates, market returns, and risk-free returns, all of which are in percentages. The data are in quarterly frequency ranging from 1947:Q2 to 2019:Q3, with a total of 290 quarters. The regions shaded in grey correspond to NBER recession periods.}
\end{figure}

We initialize our adaptive SMC$^2$ algorithm from a prior distribution that is component-wise independent, with marginal distributions that are similar to those used in \citet{fulop2020bayesian}. The left column of Table \ref{estimateslrr} details the exact distributional form, the support, and the hyper-parameters of the prior distribution for each model parameter.
The algorithmic settings of our adaptive SMC$^2$ are the same as those described in Section \ref{sec:dsge_realdata}. 
The right column of Table \ref{estimateslrr} reports the estimated posterior means, standard deviations, and (5, 95)\%-quantiles of model parameters.
We find that these parameter estimates are quite similar to those obtained by \citet{fulop2020bayesian} with a different SMC$^2$ algorithm,
based on the tempered likelihood approach of \citet{duan2015density} with likelihood estimators from a particle filter that employs unscented Kalman filter proposal distributions. 
\begin{table}[!t]
\begin{center}
\caption{\textbf{{\protect\footnotesize{Prior Distributions and Posterior Estimates of Parameters}}}}\label{estimateslrr}
{\onehalfspacing\footnotesize
\begin{tabular*}{\textwidth}{@{\extracolsep{\fill}}cllcccc}
\toprule
                    &  \multicolumn{2}{c}{Priors}  & \multicolumn{4}{c}{Posteriors}    \\
                       \cline{2-3}                                 \cline{4-7}

Parameter       &  Support  &  Distribution              & Mean  & Std     &  5\%  &  95\%    \\
\hline
$\delta$         & $(0, 1)$                   & $\mathcal{U}(0, 1)$                   &0.9988 & 0.0004 & 0.9980 & 0.9993  \\
$\gamma$     & $(0,\infty)$  & $\mathcal{TN}(8, 2)$              &7.4341  &  1.4268  &  5.1328  &  9.7995  \\
$\psi$             & $(0,\infty)$ & $\mathcal{TN}(2, 0.5)$            &1.3457  &  0.0892 &   1.2105  & 1.5009 \\
$\mu$            & $\mathbb{R}$      & $\mathcal{N}(\bar{\mu}, 10^{-5})$ &0.0047  &  0.0000  &  0.0047  & 0.0047 \\
$\rho$            & $(-1, 1)$                  &$\mathcal{U}(-1, 1)$                   & 0.9846 &   0.0035  &  0.9784 &   0.9895  \\
$\phi_x$         & $(0,\infty)$  & $\mathcal{TN}(0.10, 0.20)$     &0.2066  &  0.0300  &  0.1582  & 0.2564\\
$\bar{\sigma}$& $(0,\infty)$ & $\mathcal{TN}(0.004, 0.005)$ &0.0047 &   0.0001 &   0.0044  &  0.0049\\
$\nu$             & (-1, 1)                   & $\mathcal{U}(-1, 1)$                   &0.7161 &   0.0403  &  0.6447  &  0.7798 \\
$\phi_s$        & $(0,\infty)$   & $\mathcal{TN}(2, 4)$                 &1.8603 &   0.2763  &  1.3912  &  2.3104\\
$\mu_d$        &$\mathbb{R}$        & $\mathcal{N}(\bar{\mu}_d, 10^{-5})$&0.0064 &   0.0000  &  0.0064  &  0.0064 \\
$\Phi$            & $\mathbb{R}$       & $\mathcal{N}(3, 6)$                     &0.9495  &  0.0863  &  0.8228  &  1.0995 \\
$\phi_{dc}$    & $\mathbb{R}$       & $\mathcal{N}(3, 6)$                     &0.6478 &   0.2159  &  0.3491  &  1.0413 \\
$\phi_d$        & $(0,\infty)$    &$\mathcal{TN}(5, 6)$                 & 4.6717 &   0.2506 &   4.2618  &  5.0893  \\
$\phi_{m}$& $(0,\infty)$    & $\mathcal{TN}(0.03, 0.10)$     & 0.0811 &   0.0031  &  0.0762  &  0.0864 \\
$\phi_{r}$  & $(0,\infty)$    &$\mathcal{TN}(0.003, 0.01)$     & 0.0025 &   0.0001  &  0.0024 &  0.0027  \\
\hline
Model Evidence           &                                 &                               &  \multicolumn{4}{c}{$3.3895\times10^{3}$} \\
\bottomrule
\end{tabular*}}
\end{center}
{\footnotesize\emph{Note:} This table details the prior distributions and posterior estimates of model parameters.
The left column provides the exact distributional form, the support, and the hyper-parameters of the prior distribution for each model parameter.
$\mathcal{N}$ stands for the normal distribution, $\mathcal{TN}$ the truncated normal distribution, and $\mathcal{U}$ the uniform distribution. The model is estimated using our adaptive SMC$^2$ algorithm with $P=1,024$ parameter particles and $N=1,024$ state particles within AC-SMC.
The right column reports the posterior means, standard deviations, and (5, 95)\%-quantiles of model parameters with the standard deviations of  measurement errors treated as free parameters.}
\end{table}

However, we find that our adaptive SMC$^2$ algorithm has higher and more stable acceptance rates over the annealing iterations.
Figure \ref{fig:acceptratelrr} displays the acceptance rate of the final PMMH move at each iteration, which ranges around 0.48 to 0.58. 
While the acceptance rate in \citet{fulop2020bayesian} is comparable at the initial stages of annealing, it decreases over the annealing iterations to around 20\% as the algorithm terminates. 
Higher and more stable acceptance rates in Figure \ref{fig:acceptratelrr} indicate that AC-SMC is much more efficient than the particle filter employed in \citet{fulop2020bayesian} in terms of likelihood estimation. 
\begin{figure}[t]
\begin{center}
\includegraphics[scale=0.6]{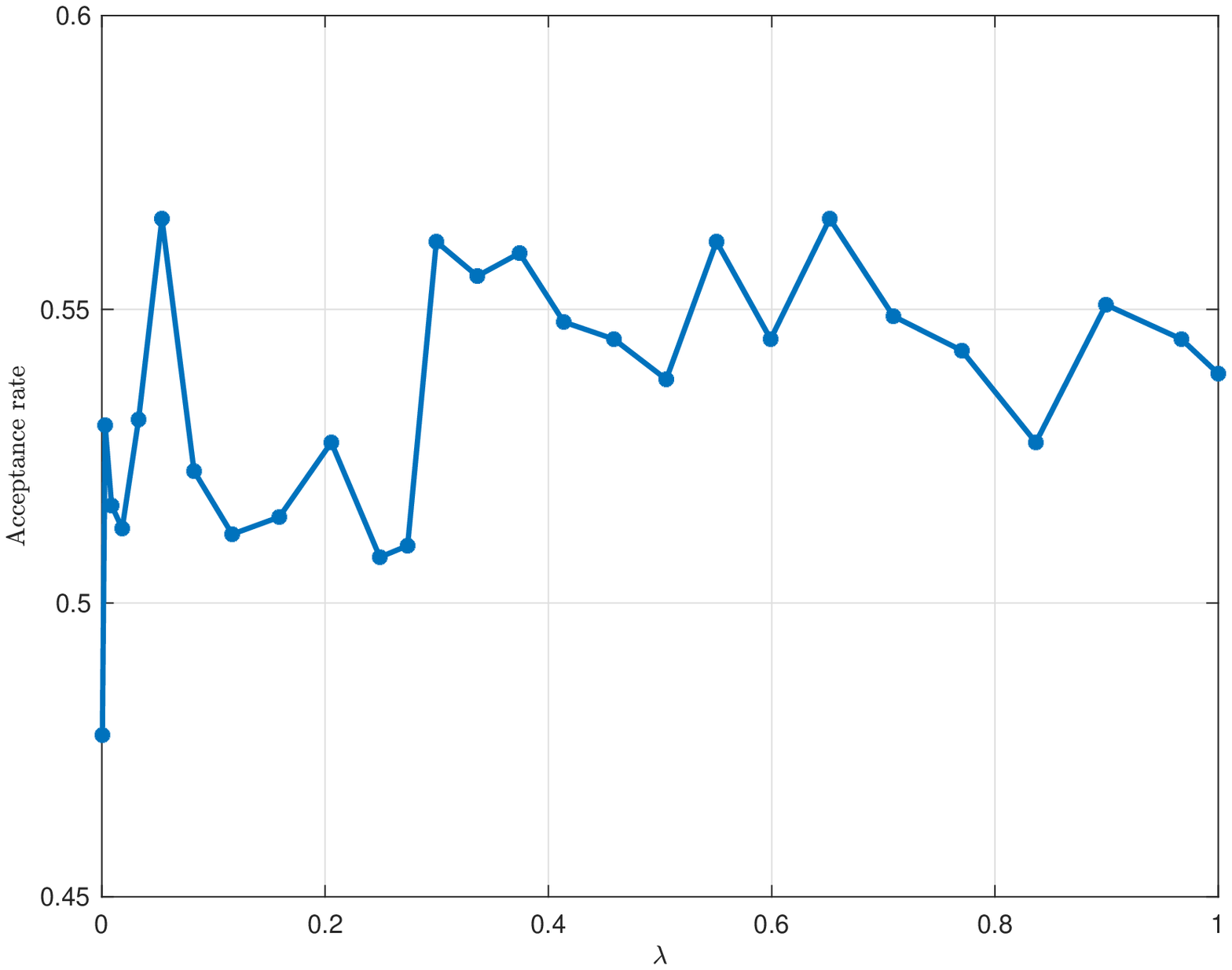}
\end{center}
\vspace{-0.5cm}
\caption{\textbf{\protect\small {Acceptance Rates}}}\label{fig:acceptratelrr}
\end{figure}

\end{document}